\documentclass[12pt]{article}
\usepackage{amsmath, amsthm, amssymb}
\usepackage{graphicx,psfrag,epsf}
\usepackage{graphics}
\usepackage{booktabs}
\usepackage{array}
\newcolumntype{H}{>{\setbox0=\hbox\bgroup}c<{\egroup}@{}}

\usepackage{multirow}
\usepackage{enumerate}
\usepackage{enumitem}
\usepackage[nodisplayskipstretch]{setspace}
\usepackage[ruled,vlined]{algorithm2e}
\usepackage{float}
\usepackage{bm}
\usepackage{bbm}
\usepackage{natbib}
\usepackage{hyperref}
\hypersetup{
	colorlinks = true,
	linkcolor = blue,
	citecolor = black,
	urlcolor=blue
}

\newtheorem{thm}{Theorem}
\newtheorem{thmsupp}{Theorem}
\newtheorem{cor}{Corollary}
\newtheorem{lemma}{Lemma}
\newtheorem{defn}{Definition}
\newtheorem{remark}{Remark}

\newcommand{\reals}{\mathbb{R}}
\newcommand{\yvec}{\mathbf{y}}
\newcommand{\xvec}{\mathbf{x}}
\newcommand{\xmat}{\mathbf{X}}
\newcommand{\bvec}{\bm{\beta}}
\newcommand{\dvec}{\bm{\delta}}

\newcommand{\ind}{\mathbbm{1}}
\newcommand{\iid}{\overset{iid}{\sim}}

\newcommand{\norm}[1]{\left\lVert#1\right\rVert}

\DeclareMathOperator{\E}{E}
\DeclareMathOperator*{\argmin}{arg\,min}
\DeclareMathOperator*{\initial}{initial}
\DeclareMathOperator*{\oracle}{oracle}
\DeclareMathOperator*{\diag}{diag}

\DeclareMathOperator*{\lasso}{lasso}
\DeclareMathOperator*{\sgn}{sgn}
\DeclareMathOperator*{\VEC}{vec}
\DeclareMathOperator*{\subG}{subG}
\DeclareMathOperator*{\subExp}{subExp}

\allowdisplaybreaks

\begin{document}

\def\spacingset#1{\renewcommand{\baselinestretch}%
{#1}\small\normalsize} \spacingset{1}


\title{\bf High-dimensional Censored Regression via the Penalized Tobit Likelihood}
\author{Tate Jacobson\\
	School of Statistics, University of Minnesota\\
	and \\
	Hui Zou\\
	School of Statistics, University of Minnesota}
\maketitle

\bigskip
\begin{abstract}
	High-dimensional regression and regression with a left-censored response are each well-studied topics. In spite of this, few methods have been proposed which deal with both of these complications simultaneously.
	The Tobit model---long the standard method for censored regression in economics---has not been adapted for high-dimensional regression at all. 
	To fill this gap and bring up-to-date techniques from high-dimensional statistics to the field of high-dimensional left-censored regression, we propose several penalized Tobit models.
	We develop a fast algorithm which combines quadratic minimization with coordinate descent to compute the penalized Tobit solution path.
	Theoretically, we analyze the Tobit lasso and Tobit with a folded concave penalty, bounding the $\ell_2$ estimation loss for the former and proving that a local linear approximation estimator for the latter possesses the strong oracle property.  
	Through an extensive simulation study, we find that our penalized Tobit models provide more accurate predictions and parameter estimates than other methods 
	on high-dimensional left-censored data. 
	We use a penalized Tobit model to analyze high-dimensional left-censored HIV viral load data from the AIDS Clinical Trials Group and identify potential drug resistance mutations in the HIV genome.  
	Appendices contain intermediate theoretical results and technical proofs.
\end{abstract}

\noindent%

\vfill

\begin{center}
	\fbox{
		\begin{minipage}{6in}
			This is an original manuscript of an article published by Taylor \& Francis in the \textit{Journal of Business and Economic Statistics} on March 15, 2023, available online: \url{https://www.tandfonline.com/doi/full/10.1080/07350015.2023.2182309}.
		\end{minipage}
	}
\end{center}

\vfill

\newpage
\spacingset{1.15} 

\section{Introduction}
In many regression problems, the dependent variable can only be observed within a restricted range. We say that such a response is \textit{censored} if we retain some information from the observations which fall outside of this range rather than losing them entirely. In particular, we still observe the predictors for these cases and know whether the unobserved response value fell below or above the range. Censored data appear in many disciplines, either as a consequence of the data collection process or due to the nature of the response itself. 
For instance, biological assays used to measure human immunodeficiency virus (HIV) viral load in plasma cannot detect viral concentrations below certain (known) thresholds. As such, the observed viral load is left-censored. 
Because censoring violates a key assumption of linear regression, ordinary least squares (OLS) estimates of the regression coefficients will be biased and inconsistent if the response is censored \citep{Amemiya1984}. Recognizing this, researchers in different disciplines have developed regression techniques to deal with various types of censoring. 
Among these, the Tobit model has long been the standard method for modeling a left-censored response in economics.

\cite{Tobin1958} originally developed the Tobit  model to study how annual expenditures on durable goods relate to household income. Noting that most low-income households spend $\$ 0$/year on durable goods, he designed the Tobit likelihood to treat response values at this (known) lower limit differently than those above the limit. He described the Tobit model as a ``hybrid of probit analysis and multiple regression," as it models the probability of the response falling at the lower limit using an approach similar to probit analysis while still treating the response as continuous \citep{Tobin1958}. 
Because left-censored data are common in household surveys and other micro-sample survey data, the Tobit model has enjoyed lasting popularity in economics and the social sciences. 
In the half-century since its introduction, it has been extended to handle right-censored and interval-censored data \citep{Amemiya1984} and has been adopted in other disciplines. 

In recent years, high-dimensional data have become increasingly common in many fields of study. 
This presents some researchers with the challenge of analyzing data with both high-dimensional covariates and a left-censored response.
Consider the HIV viral load example from earlier. There is now a sizable literature around modeling the relationship between HIV viral load and mutations in the HIV genome. Given the number of mutations that can occur, this is inherently a high-dimensional problem, where the number of predictors $p$ is much larger than the number of observations $n$. At the same time, the observed viral load is left-censored. Previous studies in this area have avoided the problem of having both high-dimensional covariates and a left-censored response by reducing HIV viral load to a binary response, such as $y =\ind_{\text{viral load} > 200 \text{ copies/mL} } $. In taking this approach, however, the modelers lose a great deal of information about the response. To directly model HIV viral load in this setting, researchers need techniques designed specifically for high-dimensional left-censored regression.

While high-dimensional regression and left-censored regression have been thoroughly studied as separate topics, few methods have been developed which handle both high-dimensional covariates and a left-censored response simultaneously.
\cite{Muller2016} and \cite{Zhou2016} have extended the least absolute deviation estimator of \cite{Powell1984} for high-dimensional data while \cite{Johnson2009}, \cite{Li2014}, and \cite{Soret2018} have extended the Buckley-James estimator \citep{Buckley1979}. To our knowledge, no existing methods directly extend the Tobit model.
Theoretically, this under-studied area has fallen behind the broader field of high-dimensional statistics, with estimators achieving weaker guarantees and requiring stronger assumptions. 
Among existing high-dimensional left-censored regression techniques, only \citeauthor{Muller2016}'s (\citeyear{Muller2016}) estimator has any theoretical guarantees in the setting where $p \gg n$. 
This estimator, however, does not achieve consistent model selection. On the other hand, the estimators of \cite{Johnson2009}, \cite{Zhou2016}, and \cite{Li2014} are shown to possess the weak oracle property, but only in the fixed $p$ case. We aim to improve on these high-dimensional left-censored regression techniques by developing an estimator which possesses the strong oracle property even when $p \gg n$.

In this study, we develop penalized Tobit models for high-dimensional censored regression.
The negative log-likelihood in \citeauthor{Tobin1958}'s \citeyearpar{Tobin1958} original formulation of the Tobit model is non-convex, creating technical problems for optimization in a high-dimensional setting. We use \citeauthor{Olsen1978}'s \citeyearpar{Olsen1978} convex reparameterization of the negative log-likelihood in our penalized Tobit models so that we can solve our problem using convex optimization methods. 
In particular, we leverage the fact that the negative log-likelihood satisfies the quadratic majorization condition to develop a generalized coordinate descent (GCD) algorithm \citep{Yang2013} for minimizing the penalized negative log-likelihood.

For our theoretical study, we analyze the Tobit lasso and Tobit with a folded concave penalty in a high-dimensional setting with $p \gg n$. 
We derive a bound for the $\ell_2$ estimation loss for the Tobit lasso estimator which holds with high probability. 
We introduce a local linear approximation (LLA) algorithm for Tobit regression with a folded concave penalty and prove that, when initialized with the Tobit lasso estimator, this algorithm finds the oracle estimator in one step and converges to it in two steps with probability rapidly converging to 1 as $n$ and $p$ diverge. To our knowledge, this makes the two-step LLA estimator the first estimator for high-dimensional left-censored regression to possess the strong oracle property.

We have implemented the GCD algorithm and the LLA algorithm (specifically with the SCAD penalty \citep{Fan2001}) in the \texttt{tobitnet} package in \texttt{R}, which is available at \url{https://github.com/TateJacobson/tobitnet}.

This paper is organized as follows.
In Section \ref{sec:tobit} we review the Tobit model and its statistical foundations. 
In Section \ref{sec:tobitnet} we introduce penalized Tobit models and develop our GCD algorithm to fit them. 
In Section \ref{sec:theory} we carry out our theoretical study of the Tobit lasso and our LLA algorithm for Tobit with a folded concave penalty.
Section \ref{sec:sims} presents the results of an extensive simulation study comparing our penalized Tobit models with penalized least-squares models and \citeauthor{Soret2018}'s (\citeyear{Soret2018}) high-dimensional Buckley-James estimator (the best available alternative for high-dimensional left-censored regression) in terms of their prediction, estimation, and selection performance.  
In Section \ref{sec:HIV} we analyze real high-dimensional left-censored data from the AIDS Clinical Trials Group, modeling the relationship between HIV viral load and HIV genotypic mutations using the two-step Tobit LLA estimator and \citeauthor{Soret2018}'s Buckley-James estimator in order to identify potential drug resistance mutations (DRMs).
Intermediate theoretical results and technical proofs are provided in Appendices \ref{sec:intermediate results} and \ref{sec:proofs}, respectively.

\section{The Tobit Model}
\label{sec:tobit}
Suppose that we observe a set of predictors, $x_1, \ldots, x_p$, and a response $y \geq c$ where $c$ is a known lower limit (for example, $c = 50$ if our response is HIV viral load and our assays cannot measure concentrations below $50 \text{ copies/mL}$).
In Tobit regression we assume that there exists a latent response variable $y^*$ such that $y = \max\{y^*, c\}$
and that $y^*$ comes from a linear model
$y^* =  \xvec'\bvec + \epsilon$,
where $\xvec = (1, x_1, \ldots, x_p)' \in \reals^{p+1}$, $\bvec = (\beta_0, \beta_1, \ldots, \beta_p)' \in \reals^{p+1}$, and $\epsilon \sim N(0, \sigma^2)$.  
In the following developments we assume that $c=0$ without loss of generality.

From this latent-variable formulation we can derive a likelihood for the censored response.
Let $\{(y_i,\xvec_i')\}_{i=1}^n$ be i.i.d copies of $(y,\xvec')$ and
define $d_i = \ind_{y_i > 0}$. Let  $\Phi(\cdot)$ denote the standard normal CDF. 
The Tobit likelihood is given by
$$L_n(\bvec, \sigma^2) = \prod_{i=1}^{n} \left[ \frac{1}{\sqrt{2\pi}\sigma}\exp\left\{-\frac{1}{2\sigma^2} (y_i - \xvec_i' \bvec)^2 \right\} \right]^{d_i} \left[ \Phi\left(\frac{- \xvec_i' \bvec}{\sigma} \right) \right]^{1-d_i} \text{.}$$
Noting that 
$P(y_i^* \leq 0) = P(\xvec_i' \bvec + \epsilon_i \leq 0) =  \Phi\left(\frac{- \xvec_i' \bvec}{\sigma} \right) \text{,}$
we see that this likelihood is a mixture of a normal density and a point mass at $0$.
After dropping an ignorable constant the log-likelihood is given by
$$\log L_n(\bvec, \sigma^2) = \sum_{i=1}^{n} d_i\left[ -\log(\sigma) -\frac{1}{2\sigma^2} (y_i - \xvec_i' \bvec)^2 \right] + (1-d_i) \log\left( \Phi\left(\frac{- \xvec_i' \bvec}{\sigma} \right) \right) \text{.}$$

\section{Penalized Tobit Regression}
\label{sec:tobitnet}
In high-dimensional regression, the most commonly used approach is to exploit sparsity in the regression coefficient vector. 
While we might initially consider simply adding a penalty term to the Tobit log-likelihood to create an objective function for penalized Tobit regression, $\log L_n(\bvec, \sigma^2)$ is not concave in $(\bvec, \sigma^2)$, frustrating this approach. Thankfully, \cite{Olsen1978} found that the reparameterization $\dvec = \bvec/\sigma$ and $\gamma^2 = \sigma^{-2}$  results in a concave log-likelihood:
$$\log L_n(\dvec, \gamma) = \sum_{i=1}^{n} d_i\left[ \log(\gamma) -\frac{1}{2} (\gamma y_i - \xvec_i' \dvec)^2 \right] + (1-d_i) \log\left( \Phi\left(- \xvec_i'\dvec \right) \right) \text{.}$$
Note that $\dvec$ and $\bvec$ must have the same degree of sparsity. 
We use Olsen's reparameterization to develop our penalized Tobit models. Our objective is to minimize
\begin{equation}
	R_n(\dvec, \gamma) = \ell_n(\dvec, \gamma) + P_{\lambda}(\dvec) \label{eqn:penalized tobit objective}
\end{equation}
with respect to $(\dvec, \gamma)$,
where $\ell_n(\dvec, \gamma) =  - \frac{1}{n} \log L_n(\dvec, \gamma)$ is our convex loss function (the \textit{Tobit loss} for short) and $P_{\lambda}(\dvec)$ is a penalty function. Note that, unlike with other loss functions, we cannot separate out the scale parameter $\gamma$ in the Tobit loss when estimating the regression coefficients $\dvec$.

Coordinate descent (CD) is currently the most popular algorithm for high-dimensional regression in the literature \citep{glmnet}. Given the relatively complex form of the Tobit loss, the standard CD algorithm requires solving a nonlinear convex program repeatedly for $p$ variables for many cycles until convergence. As a result, the computation time will be notably longer than for penalized least squares. 
Fortunately, another benefit of using Olsen's reparameterization is that the Tobit likelihood can be shown to enjoy a nice \textit{quadratic majorization condition} that serves as a foundation for using the majorization-minimization (MM) principle and coordinate descent to solve penalized Tobit regression. The combination of the MM principle and CD is called generalized coordinate descent \citep{Yang2013}. By using GCD, each of the coordinate-wise updates becomes a simple univariate quadratic minimization problem.  

For reference, we say that a univariate function $f:\reals \to \reals$ satisfies the \textit{quadratic majorization condition} if there exists $M \in \reals^+$ such that
$f(t+a)\leq f(t) + f'(t) a + \frac{M}{2}a^2$ for all $t,a \in \reals.$ Without loss of generality, we assume that our predictors are standardized---that is, $\frac{1}{n}\sum_{i = 1}^n x_{ij} = 0$ and $
\frac{1}{n}\sum_{i =1}^{n} x_{ij}^2 = 1$, for $j = 1, \ldots, p$.
Consider coordinate-wise updates of $\delta_0$ and $\delta_j$, $j = 1, \ldots, p$. 
We treat $\delta_0$ as a special case of $\delta_j$ in the following developments, keeping in mind that $x_{i0} = 1$ for all $i$. For ease of notation, let $\xvec_{i(-j)} = (x_{i0}, x_{i1}, \ldots , x_{i, j-1} , x_{i, j+1}, \ldots, x_{ip})' \in \reals^p$ and $\dvec_{(-j)} = (\delta_0, \delta_1, \ldots, \delta_{j-1}, \delta_{j+1}, \delta_p)' \in \reals^p$.

Let $\tilde{\dvec}$ and $\tilde{\gamma}$ denote the current values for $\dvec$ and $\gamma$. Let $j \in  \{0, 1, \ldots, p\}$ and leave $\tilde{\dvec}_{(-j)}$ and $\tilde{\gamma}$ fixed. Then the Tobit loss is viewed as a univariate function of $\delta_j$. After dropping ignorable constants (which have no impact in minimization), we can express the Tobit loss with respect to $\delta_j$ as
$$\ell_n(\delta_j| \tilde{\dvec}, \tilde{\gamma}) = \frac{1}{n}\sum_{i = 1}^n d_i \frac{1}{2} (\tilde{\gamma} y_i - \xvec_{i,(-j)}'\tilde{\dvec}_{(-j)} - x_{ij}\delta_j )^2  - (1-d_i) \log \Phi(-\xvec_{i,(-j)}'\tilde{\dvec}_{(-j)} - x_{ij}\delta_j ).$$ 
\begin{thm}\label{thm:GCD majorization}
	$\ell_n(\delta_j| \tilde{\dvec}, \tilde{\gamma})$ satisfies the quadratic majorization condition with $M = \frac{1}{n} \sum_{i = 1}^n x_{ij}^2$. Under the standardization of predictors, $M=1$. 
\end{thm}

To illustrate the whole process of GCD, we focus on the weighted lasso penalty:
$P_{\lambda}(\dvec)=\sum_{j = 1}^p \lambda w_j|\delta_j|$. When $w_j=1$ for all $j$, this penalty reduces to the lasso penalty. The weighted lasso penalty form will also be used in the computation of the folded-concave-penalized Tobit estimator (see the next section for details).

The standard coordinate descent algorithm needs to minimize $\ell_n(\delta_j| \tilde{\dvec}, \tilde{\gamma})+\lambda w_j |\delta_j|$, which requires another iterative procedure to find the minimizer. By Theorem~\ref{thm:GCD majorization}, we have that 
$$Q(\delta_j|  \tilde{\dvec}, \tilde{\gamma}) := \ell_n(\tilde{\delta}_j| \tilde{\dvec}, \tilde{\gamma}) + \ell_n'(\tilde{\delta}_j| \tilde{\dvec}, \tilde{\gamma})\cdot (\delta_j - \tilde{\delta}_j)  + \frac{1}{2} (\delta_j - \tilde{\delta}_j)^2$$ 
is a quadratic majorization function for $\ell_n(\delta_j| \tilde{\dvec}, \tilde{\gamma})$.
By the MM principle, we can simply minimize
$Q(\delta_j|  \tilde{\dvec}, \tilde{\gamma})
+\lambda w_j |\delta_j|$ to update $\delta_j$, while leaving the other parameters fixed at their current values. 
For $j = 1, \ldots, p$, we update $\delta_j$ as the minimizer of $Q(\delta_j|  \tilde{\dvec}, \tilde{\gamma}) + \lambda w_j |\delta_j|$, which is given by a soft-thresholding rule \citep{Tibshirani1996}:
$\hat{\delta}_j = S\left(\tilde{\delta}_j - \ell_n'(\tilde{\delta}_j| \tilde{\dvec}, \tilde{\gamma}), w_j\lambda \right),$
where $S(z,t) = (|z| - t)_{+} \text{sgn}(z)$.
For updating 
$\delta_0$, we use the minimizer of $Q(\delta_0|  \tilde{\dvec}, \tilde{\gamma})$, which is given by
$\hat{\delta}_0 = \tilde{\delta}_0 - \ell_n'(\tilde{\delta}_0| \tilde{\dvec}, \tilde{\gamma}) \text{.}$
Lastly, we need to update $\gamma$.
We can show that given $\tilde{\dvec}$, $\ell_n(\gamma | \tilde{\dvec})$ is minimized by
$\hat{\gamma} = 
\frac{\sum_{i = 1}^n d_i y_i ( \xvec_i'\tilde{\dvec}) + \sqrt{\left( \sum_{i = 1}^n d_iy_i ( \xvec_i'\tilde{\dvec}) \right)^2 + 4\left( \sum_{i = 1}^n d_i y_i^2  \right) \sum_{i=1}^n d_i  } }
{2 \sum_{i = 1}^n d_i y_i^2} $. For completeness, we show the whole GCD algorithm in  Algorithm~\ref{algo:gcd}.\\
\begin{algorithm}[H]\label{algo:gcd}
	\SetAlgoLined
	Initialize $(\tilde{\dvec}, \tilde{\gamma})$\;
	\Repeat{convergence }{
		Compute $\hat{\delta}_0 = \tilde{\delta}_0 - \ell_n'(\tilde{\delta}_0| \tilde{\dvec}, \tilde{\gamma})$; 		Set $\tilde{\delta}_0 = \hat{\delta}_0$ \;
		
		\For{$j = 1,\ldots, p$}{
			Compute $\hat{\delta}_j =  S\left(\tilde{\delta}_j - \ell_n'(\tilde{\delta}_j| \tilde{\dvec}, \tilde{\gamma}), w_j\lambda \right)$; 			Set $\tilde{\delta}_j = \hat{\delta}_j$ \;
		}
		
		Compute $\hat{\gamma} =  
		\frac{\sum_{i = 1}^n d_i y_i ( \xvec_i'\tilde{\dvec}) + \sqrt{\left( \sum_{i = 1}^n d_iy_i ( \xvec_i'\tilde{\dvec}) \right)^2 + 4\left( \sum_{i = 1}^n d_i y_i^2  \right) \sum_{i=1}^n d_i  } }
		{2 \sum_{i = 1}^n d_i y_i^2}$; 		Set $\tilde{\gamma} = \hat{\gamma}$\;
	}
	\caption{GCD algorithm for penalized Tobit with the weighted lasso penalty}
\end{algorithm} 

\section{Theoretical Results}\label{sec:theory}
Our objective function \eqref{eqn:penalized tobit objective} is flexible enough to accommodate a wide variety of penalties. In this section we offer theoretical studies of penalized Tobit estimators. We focus on the Tobit estimator with the lasso penalty and Tobit estimators with folded concave penalties.

\subsection{Setup}
Suppose that $\dvec^*$ and $\gamma^*$ are the true parameter values for $\dvec$ and $\gamma$. 
For notational convenience, we define $\Theta = (\theta_0, \ldots, \theta_{p+1})' := (\dvec', \gamma)'$.
We let $\mathcal{A} = \{j : \delta_j^* \neq 0\} \subseteq \{1, \ldots, p\}$ denote the true support set and define $\mathcal{A}' = \mathcal{A} \cup \{0, p+1 \}$ and $s = |\mathcal{A}|$. Under the sparsity assumption, $ s \ll p$.
Note that we continue to index $\dvec$, $\Theta$, and the columns of $\xmat$ from $0$ to $p+1$ to accommodate $\delta_0$ in $\dvec$.

We adopt the following notation throughout our analysis. For a matrix $\mathbf{A} \in [a_{ij}]_{n \times m}$ and sets of indices $\mathcal{S} \subseteq \{1, \ldots, m\}$ and $\mathcal{T} \subseteq \{1, \ldots, n\}$, we use $\mathbf{A}_{(\mathcal{S})}$ to denote the submatrix consisting of the \textit{columns} of $\mathbf{A}$ with indices in $\mathcal{S}$ and $\mathbf{A}_{\mathcal{T}}$ to denote the submatrix consisting of the \textit{rows} of $\mathbf{A}$ with indices in $\mathcal{T}$. We let $\lambda_{\max}(\mathbf{A})$ denote the largest eigenvalue of $\mathbf{A}$, let $\mathbf{A} \succ 0$ signify that $\mathbf{A}$ is positive definite, and let $\VEC(\mathbf{A}) \in \reals^{nm}$ denote the vectorization of $\mathbf{A}$. We define several matrix norms: the $\ell_{\infty}$-norm $\norm{\mathbf{A}}_{\infty} = \max_i \sum_j |a_{ij}|$, the $\ell_1$-norm $\norm{\mathbf{A}}_1 = \max_j \sum_i |a_{ij}|$, the $\ell_2$-norm $\norm{\mathbf{A}}_2 = \lambda_{\max}^{1/2}(\mathbf{A}'\mathbf{A})$, the entry-wise maximum $\norm{\mathbf{A}}_{\max} = \max_{(i,j)} |a_{i,j}|$, and the entry-wise minimum $\norm{\mathbf{A}}_{\min} = \min_{(i,j)} |a_{i,j}|$.
We let $\nabla_{\mathcal{S}} \log L_n(\Theta)$ and $\nabla_{\mathcal{S}}^2 \log L_n(\Theta)$ denote the gradient and Hessian, respectively, of $ \log L_n(\Theta)$ with respect to $\Theta_{\mathcal{S}}$. 

Because we handle censored and uncensored observations differently in the Tobit likelihood, we introduce notation to clearly differentiate between them.
Let $n_1$ denote the number of observations for which $y_i > 0$ and $n_0 = n - n_1$. Let $\xmat_1$ be the $n_1 \times (p+1)$ matrix of predictors corresponding to the observations for which $y_i > 0$ and let $\xmat_0$ be the $n_0 \times (p+1)$ matrix of predictors corresponding to the observations for which $y_i \leq 0$. Define $\yvec_0$ and $\yvec_1$ likewise. We then reorder our observations so that
$$ \xmat = 
\begin{bmatrix}
	\xmat_0 \\
	\xmat_1
\end{bmatrix} \text{ \: and \: } 
\yvec = 
\begin{bmatrix}
	\yvec_0 \\
	\yvec_1
\end{bmatrix} \text{.}$$ 

\subsection{The lasso-penalized Tobit estimator} \label{subsec:Tobit lasso theory}
Consider the lasso-penalized Tobit estimator found by minimizing 
\begin{equation}
	R_n(\dvec, \gamma) = \ell_n(\dvec, \gamma) + \lambda_{\rm lasso}\sum^p_{j=1}|\delta_j|. \label{eqn:lasso objective}
\end{equation}
We assume that the following restricted eigenvalue condition holds:
\begin{gather}
	\kappa = \min_{\mathbf{u} \in \mathcal{C}} 
	\frac{\E \left[ 
		\norm{
			\begin{bmatrix}
				- \xmat_1 & \yvec_1
			\end{bmatrix}
			\mathbf{u} }_{2}^2
		\right]}{n \norm{ \mathbf{u} }_{2}^2} \in (0, \infty)
	\tag{A0} \label{RE_condition}
\end{gather}
where $\mathcal{C} = \{ \mathbf{u} \neq \mathbf{0} : \norm{ \mathbf{u}_{\mathcal{A}'^c} }_{1} \leq 3 \norm{ \mathbf{u}_{\mathcal{A}'} }_{1} \} $.
It is worth pointing out that condition \eqref{RE_condition} is similar in spirit to the restricted eigenvalue condition used in lasso-penalized least squares  \citep{Bickel2009} but also has an important technical difference because the response variable appears together with the predictors---even in the fixed design setting, the entire matrix 
$\begin{bmatrix}
	- \xmat_1 & \yvec_1
\end{bmatrix}$
is random. Thus we must take the expectation in condition \eqref{RE_condition} so that $\kappa$ will be deterministic.

In the following result, we bound the $\ell_2$ estimation loss of the Tobit lasso estimator with high probability. Let $g(s) = \phi(s)/\Phi(s)$, where $\phi(\cdot)$ denotes the standard normal density function.
We assume the following:
\begin{enumerate}[label = (A\arabic*) , start = 1 ]
	\item $\max_j \norm{\xvec_{(j)} }_2 = O(\sqrt{n})$, $\sum_{i = 1}^n (\xvec_i'\dvec^*)^2 = O(n)$, $\max_{j,k}\sum_{i = 1}^nx_{ij}^2x_{ik}^2 = O(n)$,\\ $\max_j \sum_{i = 1}^n x_{ij}^2(2 + \xvec_i'\dvec^* + g(-\xvec_i'\dvec^*))^2 = O(n)$, and $\sum_{i = 1}^n (\xvec_i'\dvec^*)^2(2 + \xvec_i'\dvec^* + g(- \xvec_i'\dvec^*))^2 = O(n)$ where $j,k\in\{0, \ldots, p\}$; \label{assumption:denominator rates}
	\item $s = O\left( n^{\alpha_1} \right)$, $\log(p) = O(n^{\alpha_2})$, where $\alpha_1, \alpha_2 \in \left(0, \frac{1}{3}\right)$; \label{assumption: rate of p,s}
\end{enumerate}
and define $M_1 = \max_j n^{-1} \norm{\xvec_{(j)} }_2^2$ and $M_2 = 16 + 4n^{-1}\sum_{i = 1}^n (\xvec_i'\dvec^*)^2$.
\begin{thm}\label{cor:lasso asymptotic}
	Suppose that $Y_i^* = \xvec_i' \beta^* + \epsilon_i$ where $\epsilon_i \iid N(0, {\sigma^*}^2)$ and define $Y_i = Y_i^* \ind_{Y_i^* > 0}$ for $i = 1, \ldots, n$.
	Let $\hat{\Theta}^{\lasso}$ denote the solution to the lasso-penalized Tobit model \eqref{eqn:lasso objective} with penalty parameter $\lambda_{\lasso} = A \sqrt{\frac{\log p}{n}}$ where $A > \max \left\{ 4 \sqrt{M_1}, \frac{\sqrt{8 M_2}}{\gamma^*} \right\} $. If \eqref{RE_condition} - \emph{\ref{assumption: rate of p,s}} hold, then for large $n, p$
	\begin{equation*}
		\norm{ \hat{\Theta}^{\lasso}  - \Theta^* }_2 \leq \frac{3\sqrt{s+2}\lambda_{\lasso}}{\kappa}
	\end{equation*}
	with probability at least
	$1 - b_1 p^{-1} - b_2 e^{-b_3 n^{1-2\alpha_1}}$
	where $b_1,b_2,b_3$ are constants.
\end{thm}
\begin{remark}
	Under condition \emph{\ref{assumption: rate of p,s}}, $\sqrt{ \frac{(s+2)\log p}{n} } \to 0$ and, by extension, the $\ell_2$ estimation loss for $\hat{\Theta}_{\lasso}$ converges to $0$ as $n,p \to \infty$. As such, $\hat{\Theta}_{\lasso}$ is consistent under the $\ell_2$ norm.
\end{remark}
\noindent Theorem \ref{cor:lasso asymptotic} follows immediately from the more general finite-sample probability bound given in Theorem \ref{thm:a0_prob} (in Appendix \ref{sec:intermediate results}). Note that we cannot compute $\lambda_{\lasso} = A \sqrt{\frac{\log p}{n}}$ in practice as $\dvec^*$ and $\gamma^*$ are unknown.

\subsection{The folded-concave-penalized Tobit estimator}\label{subsec:Tobit LLA theory}
It is now well-understood that a lasso-penalized estimator often does not achieve consistent model selection unless a stringent ``irrepresentable condition" \citep{Bin2006,Zou2006} is assumed. To relax this condition, we can try to use a folded-concave-penalized Tobit estimator. We aim to minimize
$
R_n(\dvec, \gamma) = \ell_n(\dvec, \gamma) + P_{\lambda}(\dvec)
$
where
$P_{\lambda}(\dvec) = \sum_{j = 1}^p P_{\lambda}(|\delta_j|)$ is a \textit{folded concave penalty}, meaning that $P_{\lambda}(|t|)$ satisfies
\begin{enumerate}[label = (\roman*)]
	\item $P_{\lambda}(t)$ is increasing and concave in $t \in [0, \infty)$ with $P_{\lambda}(0) = 0$.
	\item $P_{\lambda}(t)$ is differentiable in $t  \in(0, \infty)$ with $P_{\lambda}'(0) := P_{\lambda}'(0+) \geq a_1 \lambda$.
	\item $P_{\lambda}'(t) \geq a_1 \lambda$ for $t \in (0, a_2\lambda]$
	\item $P_{\lambda}'(t) = 0$ for $t \in [a\lambda, \infty)$ where $a > a_2$,
\end{enumerate}
where $a$ is pre-specified and $a_1$ and $a_2$ are fixed positive constants which depend on the folded concave penalty we choose. Two well-known folded concave penalties are the SCAD penalty \citep{Fan2001}, the derivative of which is given by
$
P_{\lambda}'(t) = \lambda \ind_{t \leq \lambda} + \frac{(a\lambda - t)_+}{a - 1}\ind_{t > \lambda},
$
where $a > 2$, and the MCP \citep{Zhang2010}, the derivative of which is given by
$
P_{\lambda}'(t) = \left(\lambda - \frac{t}{a}\right)_+,
$
where $a > 1$. One can show that $a_1 = a_2 = 1$ for the SCAD penalty and $a_1 = 1 - a^{-1}$, $a_2 = 1$ for the MCP.

The strongest rationale for using a folded concave penalty is that it can produce an estimator with the \textit{strong oracle property}---that is, an estimator which is equal to the \textit{oracle estimator} with very high probability \citep{Fan2011,Fan2014}. The oracle estimator knows the true support set $\mathcal{A}$ beforehand and, as a result, delivers optimal estimation efficiency. For the Tobit model, the oracle estimator is given by
\begin{equation}
	\hat{\Theta}^{\oracle} = \argmin_{\Theta: \Theta_{{\mathcal{A}'}^c} = \mathbf{0}} \ell_n(\Theta) \label{eqn:oracle objective}
\end{equation}
Note that we cannot solve \eqref{eqn:oracle objective} in practice because $\mathcal{A}$ is unknown. Instead, the oracle estimator is a theoretical benchmark to compare our estimators against. Because our loss function is convex, the oracle estimator is unique, meaning that
\begin{equation}
	\nabla_j \ell_n (\hat{\Theta}^{\oracle}) = 0 \text{\: \: \: }\forall j \in \mathcal{A}' \label{oracle_prop}
\end{equation}
where $\nabla_j$ denotes the derivative with respect to the $j$th element of $\Theta$.

With a folded concave penalty function, the overall objective $R_n(\dvec, \gamma)$ may no longer be convex and, consequently, could have multiple local solutions. As such, theory should be developed for a specific, explicitly-defined local solution. We examine the local solution to the folded concave penalization problem which we obtain using the LLA algorithm \citep{Zou2008}. This choice is inspired by the general theory developed in 
\cite{Fan2014} where the authors established the strong oracle property of the LLA solution for a wide class of problems. We expect the same result holds for the Tobit model. 

The LLA algorithm turns the Tobit model with a folded concave penalty into a sequence of weighted Tobit lasso models. We can use Algorithm 1 to fit each weighted Tobit lasso model. Algorithm 2 shows the complete details of the LLA algorithm for the Tobit model with a folded concave penalty.

\begin{algorithm}[H]
	\SetAlgoLined
	Initialize $\hat{\Theta}^{(0)} = \hat{\Theta}^{\initial}$ and compute the adaptive weights 
	$$\hat{\mathbf{w}}^{(0)} = (\hat{w}_1^{(0)}, \ldots, \hat{w}_p^{(0)})' = 
	( P_{\lambda}'(|\hat{\delta}_1^{(0)} |), \ldots, P_{\lambda}'(|\hat{\delta}_p^{(0)} |) )' \text{.}$$
	
	\For{$m = 1,2, \ldots$}{
		Solve the following optimization problem
		$\hat{\Theta}^{(m)} = \argmin_{\Theta} \ell_n(\Theta) + \sum_{j = 1}^p \hat{w}_j^{(m-1)} \cdot |\delta_j|$\;
		Update the adaptive weight vector with $ \hat{w}_j^{(m)} = P_{\lambda}'(|\hat{\delta}_j^{(m)} |)$ for $j = 1, \ldots, p$.
	}
	\caption{The local linear approximation (LLA) algorithm}
\end{algorithm}

We aim to show that the LLA algorithm finds the oracle estimator in one step and converges to it in two steps with probability rapidly converging to $1$ as $n,p \to \infty$.
We define
{\small $
	Q_1 = \underset{j \in \mathcal{A} \cup \{0\}}{\max} \lambda_{\max}\left( \frac{1}{n} \xmat_{ (\mathcal{A} \cup \{0\}) }' \diag\{ |\xmat_{(j)} |\} \xmat_{(\mathcal{A} \cup \{0\})}   \right)
	$;
	$
	{Q}_2 = \norm{ \left( \E \left[ \frac{1}{n} \nabla_{\mathcal{A}'}^2 \log L_n(\Theta^*)  \right] \right)^{-1} }_{\infty}
	$; 
	$
	Q_3 = Q_2 \cdot \norm{ \E\left[\frac{1}{n}
		[\nabla^2 \log L_n(\Theta^*)]_{\mathcal{A}'^c, \mathcal{A}'}
		\right] }_{\infty}
	$;
	$H_{\mathcal{A}', \mathcal{A}'}^* = \E\left[ \frac{1}{n} \nabla_{\mathcal{A}'}^2 \log L_n(\Theta^*) \right] \otimes \E\left[ \frac{1}{n} \nabla_{\mathcal{A}'}^2 \log L_n(\Theta^*) \right]$,} where $\otimes$ denotes the Kronecker product; {\small $K_1 = \norm{\E\left[ \frac{1}{n} \nabla_{\mathcal{A}'}^2 \log L_n(\Theta^*) \right]}_{\infty}$; and $K_2 = \norm{(H_{\mathcal{A}', \mathcal{A}'}^*)^{-1}}_{\infty}$.} 
We assume the following:
\begin{enumerate}[label = (A\arabic*) , start = 3]
	\item $|| \dvec^*_{\mathcal{A}} ||_{\min} > (a+1)\lambda $ \label{min signal strength}
	\item $\E\left[ \nabla^2_{\mathcal{A}'} \log L_n(\Theta^*) \right] \succ 0$ \label{hessian positive definite}
	\item $Q_1 = O(1)$, $Q_2 = O(1)$, $Q_3 = O(1)$, $K_1 = O(1)$, and $K_2 = O(1)$; \label{assumption:Q,K rates}
	\item $\exists_{C_1, C_2 > 0}$ such that $Q_2 > C_1$ and $ \norm{
		\E\left[
		\frac{1}{n}
		[\nabla^2 \log L_n(\Theta^*)]_{\mathcal{A}'^c, \mathcal{A}'}
		\right]
	}_{\infty} > C_2$ for all $n$. \label{assumption:Q lower bounds}
\end{enumerate}

\begin{thm}\label{cor:lla asymptotics}
	Suppose that $Y_i^* = \xvec_i' \beta^* + \epsilon_i$ where $\epsilon_i \iid N(0, {\sigma^*}^2)$ and define $Y_i = Y_i^* \ind_{Y_i^* > 0}$ for $i = 1, \ldots, n$. 
	Let $\lambda = B \sqrt{\frac{\log p}{n}}$ with $B > \max\left\{ \frac{(9Q_3 + 2)\sqrt{4M_1}}{a_1}, \frac{(9Q_3 + 2)\sqrt{2M_2}}{a_1 \gamma^*}, 4 Q_2 \sqrt{M_1}, \frac{Q_2 \sqrt{8 M_2}}{\gamma^*} \right\} $. Let $a_0 = \min \{1, a_2\}$.
	Suppose that our initial estimator $\hat{\Theta}^{\initial}$ satisfies
	\begin{equation}
		||\hat{\dvec}^{\initial}_{(-0)} - \dvec^*_{(-0)} ||_{\max} \leq a_0 \lambda \label{assumption:initial estimator bound}
	\end{equation}
	If \emph{\ref{assumption:denominator rates}} - \emph{\ref{assumption:Q lower bounds}} hold,
	then for large $n, p$ the LLA algorithm initialized by $\hat{\Theta}^{\initial}$ finds $\hat{\Theta}^{\oracle}$ in one iteration with probability at least
	$ 1 - c_1 p^{-1} - c_2 e^{-c_3 n^{1-2\alpha_1}}$
	and converges to $\hat{\Theta}^{\oracle}$ after two iterations with probability at least
	$ 1 - d_1 p^{-1} - d_2 e^{-d_3 n^{1-2\alpha_1}}$
	where $c_1, c_2, c_3, d_1, d_2, d_3$ are constants.
\end{thm}
\begin{remark}
	Under condition \emph{\ref{assumption: rate of p,s}} both of the probability bounds in Theorem \ref{cor:lla asymptotics} rapidly converge to $1$ as $n,p \to \infty$.
\end{remark}
\noindent Theorem \ref{cor:lla asymptotics} follows immediately from Theorem \ref{thm:a1_a2_prob} (in Appendix \ref{sec:intermediate results}), which provides more general finite-sample bounds. Note that we cannot obtain $\lambda = B \sqrt{\frac{\log p}{n}}$ in practice as $\dvec^*$ and $\gamma^*$ are unknown.

All that remains is to pick an initial estimator which satisfies \eqref{assumption:initial estimator bound} with high probability.
We choose the Tobit lasso as our initial estimator since we already have an estimation loss bound from Theorem \ref{cor:lasso asymptotic}. 
The following corollary combines Theorems \ref{cor:lasso asymptotic} and \ref{cor:lla asymptotics} to bound the probability that the LLA algorithm initialized by $\hat{\Theta}^{\lasso}$ converges to the oracle estimator in two steps.
\begin{cor}\label{cor:a0_prob}
	Suppose that $Y_i^* = \xvec_i' \beta^* + \epsilon_i$ where $\epsilon_i \iid N(0, {\sigma^*}^2)$ and define $Y_i = Y_i^* \ind_{Y_i^* > 0}$ for $i = 1, \ldots, n$. 
	Define $A$ and $B$ as in Theorems \ref{cor:lasso asymptotic} and \ref{cor:lla asymptotics}.
	If conditions \eqref{RE_condition} - \emph{\ref{assumption:Q lower bounds}} hold, 
	$\lambda_{\lasso} = A\sqrt{\frac{\log p}{n}}$
	, and $ \lambda = \max\left\{ B\sqrt{\frac{\log p}{n}}, \frac{3 \sqrt{s+2} \lambda_{\lasso}}{a_0 \kappa} \right\} $, then for large $n, p$ the LLA algorithm initialized by $\hat{\Theta}^{\lasso}$ converges to $\hat{\Theta}^{\oracle}$ after two iterations with probability at least
	$ 1 -k_1 p^{-1} - k_2 e^{- k_3 n^{1-2\alpha_1}}$
	where $k_1, k_2, k_3$ are constants.
\end{cor}
\begin{remark}
	Under condition \emph{\ref{assumption: rate of p,s}} the probability bound in Corollary \ref{cor:a0_prob} rapidly converges to $1$ as $n,p \to \infty$. As such, Corollary \ref{cor:a0_prob} establishes that the two-step LLA estimator initialized by the Tobit lasso possesses the strong oracle property in a high-dimensional setting where $p \gg n$.
\end{remark}

\section{Simulation Study}\label{sec:sims}
In the following simulation study, we compare the Tobit lasso, the two-step Tobit LLA estimator with a SCAD penalty initialized by the Tobit lasso (Tobit LLA), the least-squares lasso, least-squares SCAD, and \citeauthor{Soret2018}'s(\citeyear{Soret2018}) high-dimensional Buckley-James estimator (SAWCT2018) to determine whether the penalized Tobit models provide an appreciable improvement in prediction, estimation, and selection performance on high-dimensional data with a left-censored response. 
We compare our methods to SAWCT2018 as it is the best available alternative for high-dimensional left-censored regression. 
For reference, we set $a = 3.7$ for Tobit LLA's SCAD penalty and the least-squares SCAD penalty throughout these simulations. 

For each simulation setting, we generate $100$ datasets with $100$ training observations and $5000$ test observations.
We generate an uncensored response from a linear model  $y_i^* =  \beta_0 + \xvec_i'\bvec + \epsilon_i$,
where $\xvec_i \sim N(0, \Sigma)$ and $\epsilon_i \sim N(0,\sigma^2)$, and left-censor it to create $y_i$ as follows. 
Let $q$ denote the proportion of the $y_i$ that are left-censored in a simulated dataset. 
We control $q$ by setting $c_q$ to be the $q$-quantile of the $y_i^*$ from both the training and test data and censoring the response at $c_q$---that is, we set $y_i = \max\{y_i^*, c_q\}$.  

We have four elements we can vary across our simulation settings: $\Sigma$, $q$, $p$, and the response generating parameters ($\beta_0$, $\bvec$, and $\sigma$).  
We run simulations with each of the following covariance structures for the predictors: independent, CS(0.5), CS(0.8), AR1(0.5), and AR1(0.8) (CS($\rho$) means that $\Sigma_{ij} = \rho$ for $i \neq j$, $\Sigma_{ii}=1$ for all $i$ and AR1($\rho$) means that $(\Sigma_{\rho})_{ij} = \rho^{|i-j|}$ for all $i,j$).
For each covariance structure, we generate datasets with every combination of $q \in \{ \frac{1}{8}, \frac{1}{4}, \frac{1}{2} \}$ and $p \in \{ 50, 500 \}.$ For all of these simulations, we set $\beta_0 = 3$, $\beta = (5, 1, 0.5, -2, 0.1, 0, \ldots , 0)$, and $\sigma = 1$. 
All together we examine 30 cases. 
We group our results into Tables \ref{table:independent full results}, \ref{table:cs 0.5 full results}, \ref{table:cs 0.8 full results}, \ref{table:ar1 0.5 full results}, and \ref{table:ar1 0.8 full results} based on the covariance structure of the predictors and vary $p$ and $q$ within these tables.

To assess the prediction performance of the models, we tune each of the models on the training data using 5-fold CV then compute the MSE of their predictions on the test data. In our simulation results, we report the average test MSE over 100 replications and give its standard error in parentheses. 
We also include prediction results for the ordinary least squares oracle model (OLS Oracle) and an ordinary least squares model with all of the predictors (OLS) for cases where $p \leq n$.

We use a variety of metrics to compare the parameter estimation and selection performance of the penalized models. 
To compare the accuracy of the parameter estimates, we report the $\ell_1$ loss $\lVert \hat{\bvec} - \bvec^* \rVert_1$ and $\ell_2$ loss $\lVert \hat{\bvec} - \bvec^* \rVert_2$. 
To assess the selection performance of these models, we report the number of false positive (FP) and false negative (FN) variable selections. In our simulation results, we report the average for each metric over 100 replications and give its standard error in parentheses.

\begin{table}
	\caption{Simulation Results with Independent Covariates}
	\label{table:independent full results}
	\vspace{0.4cm}
	\centering
	\begin{tabular}[t]{ll|l|c|cccc}
		q & p & Method & MSE & $\ell_2$ & $\ell_1$ & FP & FN\\
		\hline
		&  & Lasso & 2.37(0.03) & 1.5(0.06) & 3.14(0.09) & 4.6(0.4) & 0.8(0.1)\\
		
		&  & SCAD & 2.22(0.02) & 0.98(0.05) & 2.19(0.06) & 2(0.2) & 1(0.1)\\
		
		&  & Tobit Lasso & 1.08(0.01) & 0.24(0.01) & 1.61(0.05) & 7(0.3) & \textbf{0.6(0)}\\
		
		&  & Tobit LLA & \textbf{1.01(0.01)} & \textbf{0.15(0.01)} & \textbf{0.81(0.03)} & \textbf{1.1(0.1)} & 0.9(0)\\
		
		&  & SAWCT2018 & 1.09(0.01) & 0.28(0.01) & 1.47(0.04) & 4.4(0.2) & 0.7(0)\\
		
		&  & OLS Oracle & 2.1(0.01) & - & - & - & -\\
		
		& \multirow[t]{-7}{*}{\raggedright\arraybackslash 50} & OLS & 3.95(0.07) & - & - & - & -\\
		\cline{2-8}
		&  & Lasso & 2.45(0.04) & 1.87(0.07) & 3.85(0.12) & 9.6(0.8) & 1.2(0.1)\\
		
		&  & SCAD & 2.1(0.02) & 0.91(0.05) & 2.35(0.07) & 4.5(0.4) & 1.2(0.1)\\
		
		&  & Tobit Lasso & 1.23(0.01) & 0.49(0.02) & 2.61(0.05) & 14.5(0.5) & \textbf{0.9(0)}\\
		
		&  & Tobit LLA & \textbf{1.04(0.01)} & \textbf{0.22(0.01)} & \textbf{1.03(0.03)} & \textbf{2.3(0.2)} & 1(0)\\
		
		&  & SAWCT2018 & 1.28(0.02) & 0.6(0.02) & 2.51(0.07) & 10.6(0.6) & 1(0)\\
		
		\multirow[t]{-13}{*}{\raggedright\arraybackslash $\frac{1}{8}$} & \multirow[t]{-6}{*}{\raggedright\arraybackslash 500} & OLS Oracle & 1.93(0.01) & - & - & - & -\\
		\cline{1-8}
		&  & Lasso & 3.24(0.04) & 4.45(0.13) & 4.91(0.09) & 4(0.3) & 1.2(0.1)\\
		
		&  & SCAD & 3.05(0.03) & 3.42(0.11) & 4.07(0.08) & 2.1(0.2) & 1.4(0.1)\\
		
		&  & Tobit Lasso & 0.9(0.01) & 0.3(0.01) & 1.69(0.05) & 5.9(0.3) & \textbf{0.6(0)}\\
		
		&  & Tobit LLA & \textbf{0.84(0.01)} & \textbf{0.18(0.01)} & \textbf{0.88(0.03)} & \textbf{1(0.1)} & 0.9(0)\\
		
		&  & SAWCT2018 & 0.93(0.01) & 0.42(0.02) & 1.77(0.04) & 4.3(0.2) & 0.7(0.1)\\
		
		&  & OLS Oracle & 2.79(0.01) & - & - & - & -\\
		
		& \multirow[t]{-7}{*}{\raggedright\arraybackslash 50} & OLS & 5.49(0.1) & - & - & - & -\\
		\cline{2-8}
		&  & Lasso & 3.58(0.05) & 5.01(0.13) & 5.8(0.15) & 9.5(0.9) & 1.6(0.1)\\
		
		&  & SCAD & 3.18(0.03) & 3.23(0.11) & 4.36(0.08) & 6(0.5) & 1.5(0.1)\\
		
		&  & Tobit Lasso & 1.12(0.02) & 0.61(0.03) & 2.78(0.06) & 13.4(0.4) & \textbf{1(0)}\\
		
		&  & Tobit LLA & \textbf{0.94(0.01)} & \textbf{0.28(0.02)} & \textbf{1.13(0.03)} & \textbf{2(0.2)} & 1.1(0)\\
		
		&  & SAWCT2018 & 1.21(0.02) & 0.91(0.04) & 2.95(0.08) & 10.5(0.5) & \textbf{1(0)}\\
		
		\multirow[t]{-13}{*}{\raggedright\arraybackslash $\frac{1}{4}$} & \multirow[t]{-6}{*}{\raggedright\arraybackslash 500} & OLS Oracle & 2.88(0.02) & - & - & - & -\\
		\cline{1-8}
		&  & Lasso & 3.84(0.04) & 15.6(0.21) & 8.26(0.07) & 3.3(0.3) & 1.7(0.1)\\
		
		&  & SCAD & 3.66(0.03) & 13.73(0.21) & 7.57(0.07) & 2.2(0.2) & 1.9(0.1)\\
		
		&  & Tobit Lasso & 0.69(0.01) & 0.55(0.03) & 2.24(0.06) & 6.1(0.3) & \textbf{0.6(0.1)}\\
		
		&  & Tobit LLA & \textbf{0.6(0.01)} & \textbf{0.31(0.02)} & \textbf{1.14(0.04)} & \textbf{0.8(0.1)} & 1.1(0)\\
		
		&  & SAWCT2018 & 0.74(0.01) & 1.43(0.06) & 2.87(0.07) & 3.6(0.2) & 0.8(0.1)\\
		
		&  & OLS Oracle & 3.45(0.02) & - & - & - & -\\
		
		& \multirow[t]{-7}{*}{\raggedright\arraybackslash 50} & OLS & 6.49(0.1) & - & - & - & -\\
		\cline{2-8}
		&  & Lasso & 3.99(0.04) & 18.11(0.18) & 9.66(0.16) & 8.4(0.9) & 2.2(0.1)\\
		
		&  & SCAD & 3.72(0.03) & 15.21(0.19) & 8.48(0.06) & 5.5(0.5) & 2.2(0.1)\\
		
		&  & Tobit Lasso & 0.93(0.02) & 1.42(0.07) & 3.57(0.08) & 10.2(0.4) & \textbf{1.2(0)}\\
		
		&  & Tobit LLA & \textbf{0.69(0.01)} & \textbf{0.49(0.03)} & \textbf{1.49(0.04)} & \textbf{1.7(0.2)} & 1.4(0)\\
		
		&  & SAWCT2018 & 1.11(0.03) & 3.57(0.14) & 5.09(0.1) & 10.1(0.4) & 1.3(0)\\
		
		\multirow[t]{-13}{*}{\raggedright\arraybackslash $\frac{1}{2}$} & \multirow[t]{-6}{*}{\raggedright\arraybackslash 500} & OLS Oracle & 3.28(0.01) & - & - & - & -\\
	\end{tabular}
\end{table}

\begin{table}
	\caption{Simulation Results with CS(0.5) Covariates}
	\label{table:cs 0.5 full results}
	\vspace{0.4cm}
	\centering
	\begin{tabular}[t]{ll|l|c|cccc}
		q & p & Method & MSE & $\ell_2$ & $\ell_1$ & FP & FN\\
		\hline
		&  & Lasso & 2.02(0.02) & 1.6(0.07) & 3.33(0.08) & 6.5(0.3) & 0.9(0.1)\\
		
		&  & SCAD & 1.93(0.02) & 0.92(0.05) & 2.05(0.07) & \textbf{1.6(0.2)} & 1.2(0.1)\\
		
		&  & Tobit Lasso & 1.07(0.01) & 0.41(0.02) & 1.94(0.05) & 7.3(0.3) & \textbf{0.6(0)}\\
		
		&  & Tobit LLA & \textbf{1(0.01)} & \textbf{0.25(0.01)} & \textbf{1.04(0.03)} & 1.8(0.2) & 0.9(0)\\
		
		&  & SAWCT2018 & 1.08(0.01) & 0.45(0.02) & 1.72(0.04) & 5.4(0.2) & 0.7(0.1)\\
		
		&  & OLS Oracle & 1.8(0.01) & - & - & - & -\\
		
		& \multirow[t]{-7}{*}{\raggedright\arraybackslash 50} & OLS & 3.48(0.08) & - & - & - & -\\
		\cline{2-8}
		&  & Lasso & 2.37(0.03) & 2.61(0.09) & 4.97(0.15) & 15.8(0.8) & 1.5(0.1)\\
		
		&  & SCAD & 2.03(0.02) & 1.03(0.05) & 2.34(0.06) & \textbf{4.3(0.3)} & 1.6(0.1)\\
		
		&  & Tobit Lasso & 1.19(0.01) & 0.84(0.03) & 3.11(0.06) & 15.8(0.4) & \textbf{1.2(0)}\\
		
		&  & Tobit LLA & \textbf{1(0.01)} & \textbf{0.36(0.02)} & \textbf{1.41(0.04)} & 5.3(0.3) & \textbf{1.2(0)}\\
		
		&  & SAWCT2018 & 1.22(0.01) & 0.95(0.03) & 3.2(0.07) & 15.6(0.5) & \textbf{1.2(0)}\\
		
		\multirow[t]{-13}{*}{\raggedright\arraybackslash $\frac{1}{8}$} & \multirow[t]{-6}{*}{\raggedright\arraybackslash 500} & OLS Oracle & 1.87(0.01) & - & - & - & -\\
		\cline{1-8}
		&  & Lasso & 2.91(0.03) & 4.83(0.15) & 5.31(0.11) & 5.9(0.3) & 1.1(0.1)\\
		
		&  & SCAD & 2.74(0.03) & 3.31(0.13) & 3.9(0.09) & 1.8(0.2) & 1.6(0.1)\\
		
		&  & Tobit Lasso & 0.91(0.01) & 0.47(0.02) & 2.03(0.05) & 7.2(0.3) & \textbf{0.7(0.1)}\\
		
		&  & Tobit LLA & \textbf{0.84(0.01)} & \textbf{0.28(0.02)} & \textbf{1.1(0.03)} & \textbf{1.4(0.1)} & 1(0)\\
		
		&  & SAWCT2018 & 0.95(0.01) & 0.67(0.03) & 2.01(0.05) & 5(0.2) & 0.8(0)\\
		
		&  & OLS Oracle & 2.57(0.01) & - & - & - & -\\
		
		& \multirow[t]{-7}{*}{\raggedright\arraybackslash 50} & OLS & 4.81(0.08) & - & - & - & -\\
		\cline{2-8}
		&  & Lasso & 2.99(0.03) & 5.94(0.17) & 6.46(0.13) & 13.3(0.7) & 1.8(0.1)\\
		
		&  & SCAD & 2.69(0.03) & 3.34(0.14) & 4.26(0.09) & 5.7(0.4) & 1.9(0.1)\\
		
		&  & Tobit Lasso & 1.06(0.01) & 1.07(0.05) & 3.35(0.08) & 14.9(0.5) & \textbf{1.2(0)}\\
		
		&  & Tobit LLA & \textbf{0.89(0.01)} & \textbf{0.47(0.03)} & \textbf{1.58(0.05)} & \textbf{5.4(0.4)} & \textbf{1.2(0.1)}\\
		
		&  & SAWCT2018 & 1.11(0.02) & 1.33(0.05) & 3.58(0.08) & 14(0.5) & \textbf{1.2(0.1)}\\
		
		\multirow[t]{-13}{*}{\raggedright\arraybackslash $\frac{1}{4}$} & \multirow[t]{-6}{*}{\raggedright\arraybackslash 500} & OLS Oracle & 2.41(0.01) & - & - & - & -\\
		\cline{1-8}
		&  & Lasso & 3.15(0.03) & 15.52(0.21) & 8.63(0.08) & 5.4(0.3) & 1.6(0.1)\\
		
		&  & SCAD & 3.12(0.03) & 13.32(0.24) & 7.58(0.1) & 2(0.2) & 2.2(0.1)\\
		
		&  & Tobit Lasso & 0.68(0.01) & 0.86(0.05) & 2.65(0.08) & 6.7(0.3) & \textbf{0.9(0.1)}\\
		
		&  & Tobit LLA & \textbf{0.63(0.01)} & \textbf{0.55(0.05)} & \textbf{1.52(0.06)} & \textbf{1.3(0.1)} & 1.3(0.1)\\
		
		&  & SAWCT2018 & 0.77(0.01) & 2.13(0.1) & 3.4(0.08) & 4.2(0.2) & 1(0.1)\\
		
		&  & OLS Oracle & 2.87(0.02) & - & - & - & -\\
		
		& \multirow[t]{-7}{*}{\raggedright\arraybackslash 50} & OLS & 5.48(0.08) & - & - & - & -\\
		\cline{2-8}
		&  & Lasso & 3.48(0.03) & 18.3(0.26) & 9.76(0.13) & 10.1(0.7) & 2.4(0.1)\\
		
		&  & SCAD & 3.25(0.02) & 14.92(0.26) & 8.14(0.07) & \textbf{2.8(0.4)} & 2.8(0)\\
		
		&  & Tobit Lasso & 0.94(0.02) & 2.21(0.1) & 4.27(0.09) & 12.7(0.4) & \textbf{1.5(0.1)}\\
		
		&  & Tobit LLA & \textbf{0.7(0.01)} & \textbf{0.82(0.05)} & \textbf{1.99(0.06)} & 4.4(0.3) & 1.6(0.1)\\
		
		&  & SAWCT2018 & 1.02(0.02) & 3.81(0.15) & 5.36(0.1) & 12.9(0.4) & 1.6(0.1)\\
		
		\multirow[t]{-13}{*}{\raggedright\arraybackslash $\frac{1}{2}$} & \multirow[t]{-6}{*}{\raggedright\arraybackslash 500} & OLS Oracle & 2.85(0.01) & - & - & - & -\\
	\end{tabular}
\end{table}

\begin{table}
	\centering
	\caption{Simulation Results with CS(0.8) Covariates}
	\label{table:cs 0.8 full results}
	\vspace{0.4cm}
	\begin{tabular}[t]{ll|l|c|cccc}
		q & p & Method & MSE & $\ell_2$ & $\ell_1$ & FP & FN\\
		\hline
		&  & Lasso & 2(0.02) & 3.02(0.15) & 4.52(0.14) & 6(0.3) & 1.2(0.1)\\
		
		&  & SCAD & 1.94(0.02) & 1.83(0.13) & 2.76(0.11) & \textbf{0.9(0.1)} & 2.1(0.1)\\
		
		&  & Tobit Lasso & 1.06(0.01) & 1.02(0.05) & 2.74(0.08) & 6(0.2) & \textbf{1(0.1)}\\
		
		&  & Tobit LLA & \textbf{1.02(0.01)} & \textbf{0.74(0.05)} & \textbf{1.77(0.06)} & 1.2(0.1) & 1.5(0.1)\\
		
		&  & SAWCT2018 & 1.16(0.01) & 1.61(0.07) & 2.71(0.07) & 3.4(0.2) & 1.2(0.1)\\
		
		&  & OLS Oracle & 1.79(0.01) & - & - & - & -\\
		
		& \multirow[t]{-7}{*}{\raggedright\arraybackslash 50} & OLS & 3.38(0.05) & - & - & - & -\\
		\cline{2-8}
		&  & Lasso & 2.11(0.03) & 4.56(0.17) & 6.34(0.22) & 13.6(0.7) & 1.9(0.1)\\
		
		&  & SCAD & 1.94(0.02) & 2.45(0.17) & 3.15(0.11) & \textbf{1.3(0.2)} & 2.5(0.1)\\
		
		&  & Tobit Lasso & 1.21(0.01) & 1.98(0.08) & 4(0.08) & 11.3(0.4) & \textbf{1.6(0.1)}\\
		
		&  & Tobit LLA & \textbf{1.07(0.01)} & \textbf{1.09(0.07)} & \textbf{2.26(0.07)} & 4.1(0.3) & 1.8(0.1)\\
		
		&  & SAWCT2018 & 1.22(0.01) & 2.12(0.08) & 3.91(0.07) & 10.5(0.3) & \textbf{1.6(0.1)}\\
		
		\multirow[t]{-13}{*}{\raggedright\arraybackslash $\frac{1}{8}$} & \multirow[t]{-6}{*}{\raggedright\arraybackslash 500} & OLS Oracle & 1.65(0.01) & - & - & - & -\\
		\cline{1-8}
		&  & Lasso & 2.52(0.03) & 5.66(0.22) & 5.81(0.12) & 5.4(0.3) & 1.6(0.1)\\
		
		&  & SCAD & 2.54(0.03) & 4.07(0.26) & 4.32(0.18) & \textbf{0.8(0.1)} & 2.3(0.1)\\
		
		&  & Tobit Lasso & 0.92(0.01) & 1.01(0.05) & 2.74(0.08) & 5.8(0.2) & \textbf{1.1(0.1)}\\
		
		&  & Tobit LLA & \textbf{0.89(0.01)} & \textbf{0.86(0.05)} & \textbf{1.87(0.06)} & 1(0.1) & 1.7(0)\\
		
		&  & SAWCT2018 & 1.02(0.01) & 1.8(0.08) & 2.94(0.07) & 3.3(0.2) & 1.4(0.1)\\
		
		&  & OLS Oracle & 2.29(0.01) & - & - & - & -\\
		
		& \multirow[t]{-7}{*}{\raggedright\arraybackslash 50} & OLS & 4.32(0.07) & - & - & - & -\\
		\cline{2-8}
		&  & Lasso & 2.71(0.03) & 8.37(0.26) & 7.56(0.23) & 10.8(0.7) & 2.4(0.1)\\
		
		&  & SCAD & 2.57(0.02) & 5.55(0.27) & 4.92(0.13) & \textbf{1.1(0.2)} & 2.7(0)\\
		
		&  & Tobit Lasso & 1.05(0.02) & 2.29(0.11) & 4.17(0.09) & 11.3(0.4) & \textbf{1.6(0.1)}\\
		
		&  & Tobit LLA & \textbf{0.93(0.01)} & \textbf{1.37(0.09)} & \textbf{2.48(0.08)} & 3.8(0.3) & 2(0.1)\\
		
		&  & SAWCT2018 & 1.09(0.02) & 2.66(0.11) & 4.14(0.09) & 9.8(0.3) & 1.7(0.1)\\
		
		\multirow[t]{-13}{*}{\raggedright\arraybackslash $\frac{1}{4}$} & \multirow[t]{-6}{*}{\raggedright\arraybackslash 500} & OLS Oracle & 2.2(0.01) & - & - & - & -\\
		\cline{1-8}
		&  & Lasso & 2.74(0.01) & 17.09(0.27) & 9.26(0.12) & 4.8(0.3) & 2.1(0.1)\\
		
		&  & SCAD & 2.74(0.02) & 14.64(0.28) & 7.96(0.1) & \textbf{0.5(0.1)} & 2.9(0)\\
		
		&  & Tobit Lasso & \textbf{0.63(0.01)} & \textbf{1.56(0.09)} & 3.34(0.08) & 6(0.3) & \textbf{1.2(0.1)}\\
		
		&  & Tobit LLA & 0.65(0.01) & 1.67(0.12) & \textbf{2.54(0.09)} & 1(0.1) & 1.9(0.1)\\
		
		&  & SAWCT2018 & 0.78(0.01) & 4.06(0.15) & 4.34(0.08) & 2.8(0.2) & 1.7(0.1)\\
		
		&  & OLS Oracle & 2.56(0.01) & - & - & - & -\\
		
		& \multirow[t]{-7}{*}{\raggedright\arraybackslash 50} & OLS & 4.93(0.07) & - & - & - & -\\
		\cline{2-8}
		&  & Lasso & 3(0.02) & 19.32(0.31) & 10.18(0.19) & 8.4(0.6) & 2.8(0)\\
		
		&  & SCAD & 2.85(0.02) & 14.99(0.23) & 8.03(0.06) & \textbf{0.4(0.1)} & 3(0)\\
		
		&  & Tobit Lasso & 0.82(0.02) & 3.69(0.19) & 5.11(0.11) & 9.7(0.3) & \textbf{2.1(0.1)}\\
		
		&  & Tobit LLA & \textbf{0.73(0.02)} & \textbf{2.46(0.17)} & \textbf{3.14(0.09)} & 2.8(0.3) & 2.4(0.1)\\
		
		&  & SAWCT2018 & 0.89(0.02) & 5.25(0.2) & 5.51(0.11) & 7.2(0.3) & 2.3(0.1)\\
		
		\multirow[t]{-13}{*}{\raggedright\arraybackslash $\frac{1}{2}$} & \multirow[t]{-6}{*}{\raggedright\arraybackslash 500} & OLS Oracle & 2.65(0.01) & - & - & - & -\\
	\end{tabular}
\end{table}

\begin{table}
	\centering
	\caption{Simulation Results with AR1(0.5) Covariates}
	\label{table:ar1 0.5 full results}
	\vspace{0.4cm}
	\begin{tabular}[t]{ll|l|c|cccc}
		q & p & Method & MSE & $\ell_2$ & $\ell_1$ & FP & FN\\
		\hline
		&  & Lasso & 2.25(0.02) & 1.58(0.07) & 3.03(0.07) & 3.8(0.3) & 1.3(0.1)\\
		
		&  & SCAD & 2.13(0.02) & 1.09(0.05) & 2.34(0.06) & 1.8(0.2) & 1.4(0.1)\\
		
		&  & Tobit Lasso & 1.02(0.01) & 0.32(0.01) & 1.7(0.04) & 6(0.3) & \textbf{0.8(0)}\\
		
		&  & Tobit LLA & \textbf{0.97(0.01)} & \textbf{0.26(0.02)} & \textbf{1.05(0.03)} & \textbf{1.1(0.1)} & 1(0.1)\\
		
		&  & SAWCT2018 & 1.03(0.01) & 0.38(0.02) & 1.55(0.03) & 3.9(0.2) & 1(0)\\
		
		&  & OLS Oracle & 2.02(0.01) & - & - & - & -\\
		
		& \multirow[t]{-7}{*}{\raggedright\arraybackslash 50} & OLS & 3.88(0.07) & - & - & - & -\\
		\cline{2-8}
		&  & Lasso & 2.78(0.04) & 2.37(0.09) & 4.18(0.13) & 8.9(0.7) & 1.8(0)\\
		
		&  & SCAD & 2.48(0.03) & 1.33(0.07) & 2.83(0.09) & 4.5(0.4) & 1.7(0.1)\\
		
		&  & Tobit Lasso & 1.25(0.01) & 0.69(0.03) & 2.73(0.05) & 13(0.5) & 1.3(0)\\
		
		&  & Tobit LLA & \textbf{1.07(0.01)} & \textbf{0.36(0.02)} & \textbf{1.35(0.04)} & \textbf{2.4(0.2)} & \textbf{1.2(0)}\\
		
		&  & SAWCT2018 & 1.3(0.02) & 0.82(0.03) & 2.64(0.06) & 9.7(0.6) & 1.4(0.1)\\
		
		\multirow[t]{-13}{*}{\raggedright\arraybackslash $\frac{1}{8}$} & \multirow[t]{-6}{*}{\raggedright\arraybackslash 500} & OLS Oracle & 2.23(0.01) & - & - & - & -\\
		\cline{1-8}
		&  & Lasso & 3.34(0.04) & 4.66(0.14) & 5.08(0.11) & 4.3(0.3) & 1.5(0.1)\\
		
		&  & SCAD & 3.23(0.04) & 3.64(0.13) & 4.19(0.1) & 2.3(0.2) & 1.7(0.1)\\
		
		&  & Tobit Lasso & 0.93(0.01) & 0.4(0.02) & 1.94(0.05) & 6.5(0.3) & \textbf{0.9(0)}\\
		
		&  & Tobit LLA & \textbf{0.88(0.01)} & \textbf{0.33(0.02)} & \textbf{1.19(0.04)} & \textbf{1(0.1)} & 1.2(0)\\
		
		&  & SAWCT2018 & 0.94(0.01) & 0.53(0.02) & 1.87(0.03) & 4.2(0.2) & 1(0)\\
		
		&  & OLS Oracle & 2.98(0.02) & - & - & - & -\\
		
		& \multirow[t]{-7}{*}{\raggedright\arraybackslash 50} & OLS & 5.64(0.1) & - & - & - & -\\
		\cline{2-8}
		&  & Lasso & 3.72(0.05) & 5.39(0.15) & 6(0.13) & 9.1(0.8) & 1.9(0)\\
		
		&  & SCAD & 3.51(0.04) & 3.74(0.14) & 4.85(0.09) & 7.1(0.5) & 2(0.1)\\
		
		&  & Tobit Lasso & 1.15(0.02) & 0.86(0.03) & 2.95(0.05) & 12.6(0.4) & 1.4(0)\\
		
		&  & Tobit LLA & \textbf{0.97(0.01)} & \textbf{0.46(0.02)} & \textbf{1.5(0.04)} & \textbf{2.5(0.2)} & \textbf{1.3(0)}\\
		
		&  & SAWCT2018 & 1.23(0.02) & 1.13(0.04) & 3.18(0.07) & 11.1(0.6) & 1.5(0.1)\\
		
		\multirow[t]{-13}{*}{\raggedright\arraybackslash $\frac{1}{4}$} & \multirow[t]{-6}{*}{\raggedright\arraybackslash 500} & OLS Oracle & 3.03(0.01) & - & - & - & -\\
		\cline{1-8}
		&  & Lasso & 3.88(0.03) & 15.66(0.19) & 8.44(0.07) & 3.4(0.3) & 1.8(0.1)\\
		
		&  & SCAD & 3.87(0.04) & 14.01(0.2) & 7.84(0.09) & 2(0.2) & 2.3(0.1)\\
		
		&  & Tobit Lasso & 0.68(0.01) & 0.67(0.03) & 2.3(0.05) & 5.4(0.2) & \textbf{1.1(0)}\\
		
		&  & Tobit LLA & \textbf{0.63(0.01)} & \textbf{0.53(0.03)} & \textbf{1.52(0.04)} & \textbf{0.9(0.1)} & 1.5(0.1)\\
		
		&  & SAWCT2018 & 0.73(0.01) & 1.44(0.06) & 2.87(0.06) & 3(0.2) & 1.3(0)\\
		
		&  & OLS Oracle & 3.54(0.01) & - & - & - & -\\
		
		& \multirow[t]{-7}{*}{\raggedright\arraybackslash 50} & OLS & 6.78(0.1) & - & - & - & -\\
		\cline{2-8}
		&  & Lasso & 4.08(0.04) & 16.79(0.21) & 9.1(0.11) & 6(0.6) & 2.3(0)\\
		
		&  & SCAD & 4.02(0.03) & 14.17(0.21) & 8.25(0.08) & 3.8(0.5) & 2.7(0.1)\\
		
		&  & Tobit Lasso & 0.82(0.02) & 1.45(0.06) & 3.56(0.08) & 10(0.4) & 1.6(0)\\
		
		&  & Tobit LLA & \textbf{0.64(0.01)} & \textbf{0.7(0.04)} & \textbf{1.83(0.05)} & \textbf{2.3(0.2)} & \textbf{1.5(0.1)}\\
		
		&  & SAWCT2018 & 0.94(0.02) & 2.93(0.1) & 4.57(0.09) & 8.5(0.5) & 1.7(0)\\
		
		\multirow[t]{-13}{*}{\raggedright\arraybackslash $\frac{1}{2}$} & \multirow[t]{-6}{*}{\raggedright\arraybackslash 500} & OLS Oracle & 3.52(0.02) & - & - & - & -\\
	\end{tabular}
\end{table}

\begin{table}
	\centering
	\caption{Simulation Results with AR1(0.8) Covariates}
	\label{table:ar1 0.8 full results}
	\vspace{0.4cm}
	\begin{tabular}[t]{ll|l|c|cccc}
		q & p & Method & MSE & $\ell_2$ & $\ell_1$ & FP & FN\\
		\hline
		&  & Lasso & 2.01(0.02) & 2.01(0.1) & 3.57(0.1) & 4.5(0.3) & 1.6(0.1)\\
		
		&  & SCAD & 1.93(0.02) & 1.65(0.09) & 2.72(0.07) & \textbf{1.1(0.1)} & 2.1(0)\\
		
		&  & Tobit Lasso & 1.04(0.01) & \textbf{0.69(0.03)} & 2.17(0.06) & 4.8(0.2) & \textbf{1.3(0.1)}\\
		
		&  & Tobit LLA & \textbf{0.99(0.01)} & \textbf{0.69(0.05)} & \textbf{1.72(0.05)} & 1.5(0.1) & 1.8(0)\\
		
		&  & SAWCT2018 & 1.06(0.01) & 0.9(0.04) & 2.06(0.04) & 2.6(0.2) & 1.6(0)\\
		
		&  & OLS Oracle & 1.8(0.01) & - & - & - & -\\
		
		& \multirow[t]{-7}{*}{\raggedright\arraybackslash 50} & OLS & 3.42(0.07) & - & - & - & -\\
		\cline{2-8}
		&  & Lasso & 2.53(0.04) & 3.42(0.13) & 4.69(0.09) & 9.1(0.6) & 2.1(0)\\
		
		&  & SCAD & 2.23(0.02) & 2.18(0.1) & 3.3(0.06) & \textbf{3.9(0.4)} & 2.5(0.1)\\
		
		&  & Tobit Lasso & 1.26(0.02) & \textbf{1.51(0.06)} & 3.31(0.05) & 11.9(0.4) & \textbf{2(0)}\\
		
		&  & Tobit LLA & \textbf{1.19(0.02)} & 1.69(0.1) & \textbf{2.84(0.07)} & 4.1(0.4) & 2.2(0.1)\\
		
		&  & SAWCT2018 & 1.3(0.02) & 1.57(0.06) & 3.46(0.07) & 12(0.5) & \textbf{2(0)}\\
		
		\multirow[t]{-13}{*}{\raggedright\arraybackslash $\frac{1}{8}$} & \multirow[t]{-6}{*}{\raggedright\arraybackslash 500} & OLS Oracle & 1.9(0.01) & - & - & - & -\\
		\cline{1-8}
		&  & Lasso & 2.91(0.03) & 5.08(0.18) & 5.45(0.1) & 4(0.3) & 1.8(0)\\
		
		&  & SCAD & 2.84(0.03) & 4.23(0.19) & 4.53(0.13) & 1.6(0.2) & 2.5(0.1)\\
		
		&  & Tobit Lasso & 0.91(0.01) & 0.83(0.04) & 2.36(0.06) & 4.8(0.2) & \textbf{1.4(0.1)}\\
		
		&  & Tobit LLA & \textbf{0.85(0.01)} & \textbf{0.71(0.05)} & \textbf{1.78(0.05)} & \textbf{1.5(0.2)} & 1.8(0)\\
		
		&  & SAWCT2018 & 0.95(0.01) & 1.14(0.05) & 2.34(0.05) & 2.5(0.2) & 1.7(0)\\
		
		&  & OLS Oracle & 2.55(0.01) & - & - & - & -\\
		
		& \multirow[t]{-7}{*}{\raggedright\arraybackslash 50} & OLS & 5.04(0.09) & - & - & - & -\\
		\cline{2-8}
		&  & Lasso & 3.62(0.04) & 7.15(0.2) & 6.6(0.09) & 8.1(0.6) & 2.4(0.1)\\
		
		&  & SCAD & 3.28(0.02) & 4.81(0.2) & 5.07(0.1) & 4.2(0.4) & 2.8(0)\\
		
		&  & Tobit Lasso & 1.2(0.02) & \textbf{1.99(0.09)} & 3.72(0.07) & 10.5(0.4) & 2(0)\\
		
		&  & Tobit LLA & \textbf{1.08(0.02)} & 2.03(0.1) & \textbf{3.04(0.07)} & \textbf{3.3(0.3)} & 2.4(0.1)\\
		
		&  & SAWCT2018 & 1.25(0.03) & 2.14(0.09) & 4.11(0.09) & 12.1(0.5) & \textbf{1.9(0)}\\
		
		\multirow[t]{-13}{*}{\raggedright\arraybackslash $\frac{1}{4}$} & \multirow[t]{-6}{*}{\raggedright\arraybackslash 500} & OLS Oracle & 2.82(0.02) & - & - & - & -\\
		\cline{1-8}
		&  & Lasso & 3.33(0.03) & 15.47(0.23) & 8.59(0.08) & 3.3(0.3) & 2.3(0.1)\\
		
		&  & SCAD & 3.27(0.02) & 13.93(0.24) & 7.81(0.07) & \textbf{1(0.2)} & 2.9(0)\\
		
		&  & Tobit Lasso & 0.67(0.01) & \textbf{1.14(0.06)} & 2.74(0.07) & 4.1(0.2) & \textbf{1.6(0.1)}\\
		
		&  & Tobit LLA & \textbf{0.65(0.01)} & 1.33(0.1) & \textbf{2.31(0.09)} & \textbf{1(0.1)} & 2.1(0.1)\\
		
		&  & SAWCT2018 & 0.75(0.01) & 2.2(0.08) & 3.34(0.06) & 1.8(0.1) & 1.9(0)\\
		
		&  & OLS Oracle & 3.05(0.02) & - & - & - & -\\
		
		& \multirow[t]{-7}{*}{\raggedright\arraybackslash 50} & OLS & 5.93(0.1) & - & - & - & -\\
		\cline{2-8}
		&  & Lasso & 3.71(0.04) & 18.32(0.2) & 9.58(0.16) & 6.1(0.7) & 2.6(0)\\
		
		&  & SCAD & 3.36(0.02) & 15.27(0.22) & 8.38(0.07) & 2.6(0.3) & 3(0)\\
		
		&  & Tobit Lasso & 1.04(0.03) & 3.33(0.15) & 4.51(0.08) & 8(0.4) & \textbf{2.1(0)}\\
		
		&  & Tobit LLA & \textbf{0.86(0.02)} & \textbf{2.98(0.11)} & \textbf{3.53(0.05)} & \textbf{1.7(0.2)} & 2.9(0)\\
		
		&  & SAWCT2018 & 1.15(0.03) & 4.81(0.19) & 5.75(0.1) & 9.7(0.5) & \textbf{2.1(0)}\\
		
		\multirow[t]{-13}{*}{\raggedright\arraybackslash $\frac{1}{2}$} & \multirow[t]{-6}{*}{\raggedright\arraybackslash 500} & OLS Oracle & 3.07(0.02) & - & - & - & -\\
	\end{tabular}
\end{table}

\subsection{Prediction results}
We see a remarkably consistent pattern in our prediction results: in all 30 simulation settings
the penalized Tobit models attain the two lowest average test MSEs, with Tobit LLA delivering the best prediction performance in 29 of 30 settings. 
We see a clear gap in prediction performance separating the three methods which account for censoring (the Tobit lasso, Tobit LLA, and SAWCT2018) from the least squares methods in that the average test MSEs for the OLS methods are, at minimum, nearly double those of the Tobit models and SAWCT2018.
As the proportion of censored observations $q$ increases, the test MSEs for the least-squares models climb upwards while the test MSEs for the Tobit models and SAWCT2018 largely remain stable. In particular, we see that the average test MSEs for the least-squares lasso and SCAD are around five times the average test MSE of Tobit LLA when $q = \frac{1}{2}$.
The message here is clear: failing to account for censoring in the data can come at a steep price in terms of prediction accuracy, especially when the proportion of censored observations is high.

Narrowing our focus to the three models which account for censoring, we see that both of the penalized Tobit models achieve lower average test MSEs than SAWCT2018 in all 30 simulation settings and that Tobit LLA achieves the lowest average test MSE in most cases, often by a comfortable margin. In particular, we see that Tobit LLA gains a larger edge over the Tobit lasso and SAWCT2018 in simulations settings with $p = 500$ relative to those with $p = 50$. 

\subsection{Estimation results}
Turning to estimation performance, we see patterns similar to those that emerged in our prediction comparison. The penalized Tobit models' estimates have the two lowest average $\ell_2$ losses in all 30 simulation settings (the Tobit LLA estimates have the lowest average $\ell_2$ loss overall in 27 of 30 settings). In addition, the Tobit LLA estimates deliver the lowest average $\ell_1$ loss in every simulation setting.

As in the prediction comparison there is a clear gap between the least squares methods and the models which account for censoring, with the latter consistently having far lower $\ell_2$ and $\ell_1$ estimation losses. In many cases, the average $\ell_2$ losses for the Tobit and SAWCT2018 estimates differ from those of the least squares estimates by an order of magnitude.
Additionally, we once again find that the gap in estimation performance between the models that account for censoring and the least squares models grows as the proportion of censored observations increases to $q = \frac{1}{2}$.

Among the models which account for censoring, the penalized Tobit models' estimates consistently achieve lower average $\ell_2$ losses than the SAWCT2018 estimates. Shifting our focus to the $\ell_1$ loss, we see that the Tobit LLA estimates achieve markedly lower average $\ell_1$ losses than the Tobit lasso and SAWCT2018 estimates in every setting. The competition between the Tobit lasso and SAWCT2018, however, is closer, with the Tobit lasso estimates achieving a lower average $\ell_1$ loss than the SAWCT2018 estimates in just 18 of 30 settings.

\subsection{Selection results}
Our variable selection results are somewhat mixed. While the penalized Tobit models and SAWCT2018 consistently deliver lower average false negative counts than the least squares models, the differences are relatively small. At the same time, the SCAD and Tobit LLA models consistently make fewer false positive variable selections than the other models. Beyond that, neither SCAD nor Tobit LLA appears to have a clear edge in making fewer false positive selections, though Tobit LLA has a lower average false positive count in 19 of 30 settings. 

Overall, the penalized Tobit models deliver comparable (if slightly superior) selection performances to the least squares models and SAWCT2018 in this study.
These results further suggest that modelers may prefer to use the Tobit lasso if their goal is to minimize false negative variable selections and Tobit LLA if their goal is to minimize false positive variable selections.

\subsection{Takeaways}
Tobit LLA clearly outperformed competing methods in this simulation study, providing more accurate predictions and parameter estimates than the alternatives. Because it also has stronger theoretical guarantees than the Tobit lasso, we ultimately recommend Tobit LLA for analyzing high-dimensional left-censored data.

\section{HIV Viral Load and Drug Resistance} \label{sec:HIV}
Due to its short replication cycle and high mutation rate, human immunodeficiency virus (HIV) can rapidly develop drug resistance mutations (DRMs) in HIV-infected patients receiving antiretroviral therapy. To counter this, guidelines recommended physicians regularly monitor HIV viral load and, if a patient's treatment regimen is failing to suppress the virus, conduct genotypic testing to check for DRMs so they may update the patient's drug regimen appropriately \citep{Shafer2002}.

There is a substantial literature devoted to identifying DRMs and quantifying the degree of resistance they provide against different antiretroviral treatments \citep{Shafer2006}.
One way to accomplish this is by modeling the relationship between HIV viral load and mutations in the virus's genome. This poses two difficulties: (1) the observed viral load is left-censored because the assays used to measure it cannot detect concentrations below certain thresholds and (2) genome data are inherently high-dimensional.
As we established in our simulation study, it is necessary to use a model which accounts for censoring when analyzing these kind of data.
As such, we will use Tobit LLA and SAWCT2018 to model HIV viral load and identify potential DRMs.

Our data for this example come from the OPTIONS trial by the AIDS Clinical Trials Group \citep{Gandhi2020} and were downloaded from the Stanford HIV Drug Resistance Database \citep{Shafer2006}. The OPTIONS trial study population consisted of 413 HIV-infected individuals receiving protease inhibitor (PI)-based treatment and experiencing virological failure. Each participant was given an optimized antiretroviral regimen based on their viral drug resistance and treatment history. Participants with moderate drug resistance were randomly assigned to either add nucleoside reverse transcriptase inhibitors (NRTIs) to their optimized regimens or omit NRTIs from their optimized regimens.
Participants with highly drug-resistant HIV all received optimized regimens which included NRTIs.

We use Tobit LLA and SAWCT2018 to model HIV viral load 12 weeks after drug regimen assignment as a function of HIV genotypic mutations, current drug regimen, baseline viral load, observation week, and HIV subtype using a sample with $p \gg n$ and a moderate amount of left-censoring.
Our data come from the $n = 407$ participants who returned for their 12-week follow-up evaluations and include $p = 1295$ predictors, most of which are indicators for protease (PR) and reverse transcriptase (RT) gene mutations.
The assays used to measure HIV viral load in the OPTIONS trial had a detection threshold of 50 copies/mL. At their 12-week evaluations, $35.6 \%$ of study participants had viral loads which were at or below this lower limit and, consequently, undetectable. Given this limited information about these censored viral loads, investigators recorded them as falling at the lower limit of 50 copies/mL.
We use $\log_{10}$-HIV viral load as our response, as it is often assumed to be normally distributed \citep{Soret2018}.

We start by comparing the prediction performance of Tobit LLA and SAWCT2018 in terms of the Tobit loss in order to assess overall model fit. 
We randomly split the data into a training set of 326 observations and a test set of 81 observations, using stratified sampling to ensure that the training and test sets have similar proportions of left-censored observations. We repeat this process 50 times. Within each of the 50 training sets, we tune Tobit LLA and SAWCT2018 using 5-fold CV. Table \ref{table:HIV tobit prediction} reports the average Tobit loss across the 50 test sets, with the standard error in parentheses, for each model.

\begin{table}
	\centering
	\caption{Prediction Accuracy on HIV Viral Load Data}
	\label{table:HIV tobit prediction}
	\vspace{0.4cm}
	\begin{tabular}[t]{lc}
		Model & Tobit Loss \\
		\hline
		SAWCT2018 & 2.04 (0.02)\\
		Tobit LLA & \textbf{1.45 (0.01)}\\
	\end{tabular}
\end{table}

Our primary interest is in the predictors selected by the models, as they may include potential DRMs. 
We tune Tobit LLA and SAWCT2018 using 5-fold CV then fit them to the entire dataset.
Tobit LLA selects a sparse model with only three predictors: the RT mutation M184V, baseline viral load, and whether the participant is taking raltegravir (RAL), 
an integrase strand transfer inhibitor (INSTI) included in some of the patients' optimized regimens. 
SAWCT2018, on the other hand, selects 51 mutations (including M184V), baseline viral load, and whether the patient is taking RAL or the protease inhibitor saquinavir.
While it is possible that the 50 other mutations selected by SAWCT2018 include additional DRMs, the superior prediction performance of the sparse Tobit LLA model suggests that M184V is uniquely important for predicting HIV viral load in this population. It seems far more likely that SAWCT2018 is selecting unimportant mutations, reducing its utility as a method for identifying potential DRMs.

The Tobit LLA model provides some interesting insights into HIV drug resistance.
Most importantly, M184V stands out as the sole mutation selected by Tobit LLA.
This selection is supported by other research:
based on an extensive review of the HIV drug resistance literature, the Stanford HIV Drug Resistance Database lists M184V as a major NRTI resistance mutation \citep{Shafer2006}.
It is also notable that Tobit LLA did not select any NRTIs as important predictors. This is consistent with \citeauthor{Gandhi2020}'s (\citeyear{Gandhi2020}) finding that participants who added NRTIs to their regimes did not experience significantly higher rates of virological failure than those who omitted NRTIs from their regimes.

\section{Discussion}\label{sec:discussion}
As high-dimensional data become increasingly common across disciplines, we expect the need for reliable, theoretically-supported techniques for high-dimensional left-censored regression to grow.
The penalized Tobit models we introduce in this paper fill several gaps in the literature for high-dimensional left-censored regression. They are among the first models in this area with theoretical guarantees in the setting where $p \gg n$ and the lasso-initialized two-step LLA estimator for folded-concave penalized Tobit regression is the very first to possess the strong oracle property. In addition, our penalized Tobit models provide the first high-dimensional extensions of the enduringly popular Tobit model.

Our penalized Tobit models also perform well empirically.
In an extensive simulation study, our penalized Tobit models delivered superior prediction and estimation performance relative to least squares models and the best available alternative for high-dimensional left-censored regression.
When applied to real high-dimensional left-censored HIV viral load data, the Tobit LLA estimator delivered more accurate predictions and selected a more parsimonious model than the best available alternative.

\section*{Funding}
This work is supported in part by NSF DMS 1915842 and 2015120.

\bibliographystyle{taj-asa}
\bibliography{tobit}

\appendix

\renewcommand{\thethmsupp}{S.\arabic{thmsupp}}
\renewcommand{\thelemma}{S.\arabic{lemma}}
\renewcommand{\thedefn}{S.\arabic{defn}}

\section{Intermediate Theoretical Results}\label{sec:intermediate results}
In this section, we present intermediate theoretical results for the Tobit lasso and Tobit with a folded concave penalty. We start with a more general finite sample probability bound for the Tobit lasso.

\begin{thmsupp}\label{thm:a0_prob}
	Suppose that $Y_i^* = \xvec_i' \beta^* + \epsilon_i$ where $\epsilon_i \iid N(0, {\sigma^*}^2)$ and define $Y_i = Y_i^* \ind_{Y_i^* > 0}$ for $i = 1, \ldots, n$.
	Let $\hat{\Theta}^{\lasso}$ denote the solution to the lasso-penalized Tobit model with penalty parameter $\lambda_{\lasso}$.
	If assumption \eqref{RE_condition} holds, then
	\begin{equation}
		\norm{ \hat{\Theta}^{\lasso}  - \Theta^* }_2 \leq \frac{3\sqrt{s+2}\lambda_{\lasso}}{\kappa}
	\end{equation} 
	with probability at least
	{\footnotesize \begin{align*}
			1 
			& - 2(p+1)\exp\left( - \frac{n \lambda_{\lasso}^2 }{8M_1} \right) - 2 \exp \left( - \frac{n}{2} \min \left\{ \frac{\lambda_{\lasso} \gamma^*}{8}, \frac{\lambda_{\lasso}^2 {\gamma^*}^2 }{4 M_2} \right\} \right) \\
			& - 2(p+1)^2\exp\left(-\frac{2n \left(\frac{\kappa}{32(s+2)}\right)^2}{ \max_{j,k\in\{0, \ldots, p\}} n^{-1}\sum_{i = 1}^nx_{ij}^2x_{ik}^2} \right) \\
			& - 4(p+1)\exp\left(-\frac{2n \left(\frac{\kappa}{32(s+2)}\right)^2{\gamma^*}^2}{\max_{j\in\{0, \ldots, p\}}n^{-1}\sum_{i = 1}^n x_{ij}^2(2 + \xvec_i'\dvec^* + g(-\xvec_i'\dvec^*))^2} \right)\\
			& - 2 \exp\left(-\frac{n}{2} \min \left\{ \frac{\left(\frac{\kappa}{32(s+2)}\right){\gamma^*}^2}{8}, \frac{\left(\frac{\kappa}{32(s+2)}\right)^2 {\gamma^*}^4}{\frac{65}{2} + n^{-1}\sum_{i = 1}^n \frac{1}{2} (\xvec_i'\dvec^*)^2(2 + \xvec_i'\dvec^* + g(- \xvec_i'\dvec^*))^2 + 8 (\xvec_i'\dvec^*)^2} \right\} \right)  \text{.}
	\end{align*}}
\end{thmsupp}

Our next two theorems adapt the main results of~\cite{Fan2014} for the penalized Tobit model with a folded concave penalty. The first theorem provides conditions under which the LLA algorithm finds the oracle estimator after a single iteration.

\begin{thmsupp}\label{thm:one LLA step finds oracle}
	Suppose that $\dvec^*$ satisfies \emph{\ref{min signal strength}}.
	Consider the Tobit loss with a folded concave penalty $P_{\lambda}(|t|) $ satisfying conditions \textup{(i) - (iv)}. Let $a_0 = \min \{1, a_2\}$. Under the events
	$$ \mathcal{E}_0  = \{ ||\hat{\dvec}^{\initial}_{(-0)} - \dvec^*_{(-0)} ||_{\max} \leq a_0 \lambda \} \text{ \; and \; } \mathcal{E}_1 = \{|| \nabla_{\mathcal{A}'^c} \ell_n (\hat{\Theta}^{\oracle}) ||_{\max} < a_1 \lambda \} $$
	the LLA algorithm initialized by $\hat{\Theta}^{\initial}$ finds $\hat{\Theta}^{\oracle}$ after one iteration.
\end{thmsupp}
\noindent The next theorem identifies an additional condition which guarantees that the LLA algorithm will return to the oracle estimator in its second iteration.
\begin{thmsupp}\label{thm:two LLA steps converge to oracle}
	Consider the Tobit loss with a folded concave penalty $P_{\lambda}(|t|) $ satisfying conditions \textup{(i) - (iv)}. Under $\mathcal{E}_1$ and the additional event
	$$\mathcal{E}_2 = \{ || \hat{\dvec}_{\mathcal{A}}^{\oracle} ||_{\min} > a\lambda\} $$
	if $ \hat{\Theta}^{\oracle}$ is obtained, the LLA algorithm will find $\hat{\Theta}^{\oracle}$ again in the next iteration, that is, it converges to $\hat{\Theta}^{\oracle}$ in the next iteration and the oracle estimator is a fixed point.
\end{thmsupp}

As an immediate corollary to Theorems \ref{thm:one LLA step finds oracle} and \ref{thm:two LLA steps converge to oracle}, we see that the probability that the LLA algorithm finds the oracle estimator in one step is at least $1 - P(\mathcal{E}_0^c) - P(\mathcal{E}_1^c)$ and the probability that the LLA algorithm converges to the oracle estimator in two steps is at least $1 - P(\mathcal{E}_0^c) - P((\mathcal{E}_1 \cap \mathcal{E}_2)^c)$. In the next theorem, we provide finite-sample bounds for $P(\mathcal{E}_1^c)$ and $P((\mathcal{E}_1 \cap \mathcal{E}_2)^c)$.

\begin{thmsupp}\label{thm:a1_a2_prob}
	Suppose that $Y_i^* = \xvec_i' \beta^* + \epsilon_i$ where $\epsilon_i \iid N(0, {\sigma^*}^2)$ and define $Y_i = Y_i^* \ind_{Y_i^* > 0}$ for $i = 1, \ldots, n$. If assumptions \emph{\ref{min signal strength}} and \emph{\ref{hessian positive definite}} hold, 
	then the LLA algorithm initialized by $\hat{\Theta}^{\initial}$ finds $\hat{\Theta}^{\oracle}$ and one iteration with probability at least $1 - P(\mathcal{E}_0^c) - P(\mathcal{E}_1^c)$ and converges to $\hat{\Theta}^{\oracle}$ after two iterations with probability at least $1 - P(\mathcal{E}_0^c) - P((\mathcal{E}_1 \cap \mathcal{E}_2)^c)$, where
	{\footnotesize
		\begin{align*}
			P( \mathcal{E}_1^c )
			& \leq 2(s+1) \exp \left( - \frac{ n C_2^2(s, \lambda) }{ 2 M_1 } \right) 
			+ 2 \exp\left( - \frac{n}{2} \min \left\{ \frac{C_2(s, \lambda) \gamma^*}{4},  \frac{ C_2^2(s, \lambda) {\gamma^*}^2  }{ M_2 } \right \} \right) \\ 
			& + 2(p-s)\exp \left( - \frac{n a_1^2 \lambda^2}{8 M_1 } \right) 
			+ 2 (s + 1)^2 \exp \left( - \frac{2nC_{Q_2}^2(s)}{ \max_{j,k\in\mathcal{A} \cup \{0\}} n^{-1} \sum_{i = 1}^{n} x_{ij}^2x_{ik}^2 } \right) \notag \\ 
			& + 4(s+1) \exp \left( - \frac{2nC_{Q_2}^2(s){\gamma^*}^2}{ \max_{j\in\mathcal{A} \cup \{0\} } n^{-1}\sum_{i = 1}^{n} x_{ij}^2 (2 + \xvec_i'\dvec^* + g(- \xvec_i'\dvec^*) )^2 } \right) \notag \\ 
			& + 2 \exp\left(-\frac{n}{2} \min \left\{ \frac{C_{Q_2}(s){\gamma^*}^2}{8}, \frac{C_{Q_2}^2(s) {\gamma^*}^4}{34 + n^{-1}\sum_{i = 1}^n \frac{1}{2}(\xvec_i'\dvec^*)^2 (2 + \xvec_i'\dvec^* + g(- \xvec_i'\dvec^*))^2 + 8 (\xvec_i'\dvec^*)^2} \right\} \right)\\
			& +  2(s+1)(p-s) \exp \left( - \frac{2 n C_{Q_3}^2(s)}{ \max_{j \in (\mathcal{A} \cup \{0\})^c, k\in\mathcal{A} \cup \{0\}} n^{-1}\sum_{i = 1}^{n} x_{ij}^2x_{ik}^2 } \right) \notag \\
			& + 2(p - s) \exp \left( - \frac{2nC_{Q_3}^2(s){\gamma^*}^2}{ \max_{j\in(\mathcal{A} \cup \{0\})^c } n^{-1}\sum_{i = 1}^{n} x_{ij}^2 (2 + \xvec_i'\dvec^* + g(- \xvec_i'\dvec^*) )^2 } \right)
	\end{align*}}
	and 
	{\footnotesize 
		\begin{align*}
			P((\mathcal{E}_1 & \cap  \mathcal{E}_2)^c)\\
			& \leq 2(s+1) \exp \left( - \frac{ n C_3^2(s, \lambda) }{ 2 M_1 } \right) 
			+ 2 \exp\left( - \frac{n}{2} \min \left\{ \frac{C_3(s, \lambda) \gamma^*}{4},  \frac{ C_3^2(s, \lambda) {\gamma^*}^2  }{ M_2 } \right \} \right) \\ 
			& + 2(p-s)\exp \left( - \frac{n a_1^2 \lambda^2}{8 M_1 } \right) 
			+  2 (s + 1)^2 \exp \left( - \frac{2nC_{Q_2}^2(s)}{ \max_{j,k\in\mathcal{A} \cup \{0\}} n^{-1} \sum_{i = 1}^{n} x_{ij}^2x_{ik}^2 } \right) \notag \\ 
			& + 4(s+1) \exp \left( - \frac{2nC_{Q_2}^2(s){\gamma^*}^2}{ \max_{j\in\mathcal{A} \cup \{0\} } n^{-1}\sum_{i = 1}^{n} x_{ij}^2 (2 + \xvec_i'\dvec^* + g(- \xvec_i'\dvec^*) )^2 } \right) \notag \\ 
			& + 2 \exp\left(-\frac{n}{2} \min \left\{ \frac{C_{Q_2}(s){\gamma^*}^2}{8}, \frac{C_{Q_2}^2(s) {\gamma^*}^4}{34 + n^{-1}\sum_{i = 1}^n \frac{1}{2}(\xvec_i'\dvec^*)^2 (2 + \xvec_i'\dvec^* + g(- \xvec_i'\dvec^*))^2 + 8 (\xvec_i'\dvec^*)^2} \right\} \right)\\
			& + 2(s+1)(p-s) \exp \left( - \frac{2 n C_{Q_3}^2(s)}{ \max_{j \in (\mathcal{A} \cup \{0\})^c, k\in\mathcal{A} \cup \{0\}} n^{-1}\sum_{i = 1}^{n} x_{ij}^2x_{ik}^2 } \right) \notag \\
			& + 2(p - s) \exp \left( - \frac{2nC_{Q_3}^2(s){\gamma^*}^2}{ \max_{j\in(\mathcal{A} \cup \{0\})^c } n^{-1}\sum_{i = 1}^{n} x_{ij}^2 (2 + \xvec_i'\dvec^* + g(- \xvec_i'\dvec^*) )^2 } \right)
	\end{align*}}
	{\small
		where $C_1(s) = \min \left\{ \frac{1}{38.7 Q_1 Q_2^2 (s+1)}, \frac{ {\gamma^*}^3 }{144 Q_2^2}, \frac{\gamma^*}{6 Q_2} \right \}$,
		$C_2(s,\lambda) = \min \left \{ C_1(s), \frac{a_1 \lambda}{9Q_3 +2} \right \} $,\\ $C_3(s,\lambda) = \min \left \{ C_1(s), \frac{a_1 \lambda}{9Q_3 +2},\frac{1}{2 Q_2} ( \norm{ \dvec_{\mathcal{A}}^* }_{\min} - a\lambda  )  \right \} $, 
		$C_{Q_2}(s) =  \frac{1}{s+2}\min\left\{ \frac{1}{6K_1K_2}, \frac{1}{6K_1^3K_2}, \frac{Q_2}{4 K_2}  \right\}$, and
		$C_{Q_3}(s) = \frac{1}{2(s+2)}
		\norm{
			\E\left[
			\frac{1}{n}
			[\nabla^2 \log L_n(\Theta^*)]_{\mathcal{A}'^c, \mathcal{A}'}
			\right]
		}_{\infty}$.}
\end{thmsupp}

\section{Technical Proofs}\label{sec:proofs}
\subsection{Quadratic majorization}
\begin{proof}[Proof of Theorem \ref{thm:GCD majorization}]
	\label{proof:thm:majorization}
	Let $j \in \{ 0, \ldots, p \}$. Define $L_{ij}(t) = d_i \frac{1}{2} (\tilde{\gamma} y_i - \xvec_{i,(-j)}'\tilde{\dvec}_{(-j)} - x_{ij}t )^2  - (1-d_i) \log \Phi(-\xvec_{i,(-j)}'\tilde{\dvec}_{(-j)} - x_{ij}t )$ for $ i = 1, \ldots, n$, so that $\ell_n(\delta_j| \tilde{\dvec}, \tilde{\gamma}) = \frac{1}{n}\sum_{i = 1}^n L_{ij}(\delta_j)$.
	
	Let $i \in \{ 1, \ldots, n \}$. Recall that $g(s) = \phi(s)/\Phi(s)$. 
	We find that
	\begin{align*}
		L_{ij}'(t)
		& =  - d_i (\tilde{\gamma} y_i - \xvec_{i,(-j)}'\tilde{\dvec}_{(-j)} - x_{ij}t )x_{ij}  + (1-d_i) \frac{ \phi(-\xvec_{i,(-j)}'\tilde{\dvec}_{(-j)} - x_{ij}t ) }{ \Phi(-\xvec_{i,(-j)}'\tilde{\dvec}_{(-j)} - x_{ij}t ) } x_{ij}\\
		& =  - d_i (\tilde{\gamma} y_i - \xvec_{i,(-j)}'\tilde{\dvec}_{(-j)} - x_{ij}t )x_{ij}  + (1-d_i) g(-\xvec_{i,(-j)}'\tilde{\dvec}_{(-j)} - x_{ij}t )  x_{ij} \text{.}
	\end{align*}
	We see that
	\begin{equation*}
		g'(s) = \frac{ - s\phi(s)\Phi(s) - \phi^2(s) }{\Phi^2(s)} = - \frac{\phi(s)}{\Phi(s)} \left[s + \frac{\phi(s)}{\Phi(s)} \right] = -g(s)\left[ s + g(s) \right]\text{.}
	\end{equation*}
	Let $s(t) = -\xvec_{i,(-j)}'\tilde{\dvec}_{(-j)} - x_{ij}t $. We find
	\begin{align*}
		L_{ij}''(t) 
		& = d_i x_{ij}^2 + (1-d_i)g'(s(t))s'(t) x_{ij} \\
		& = d_i x_{ij}^2 + (1-d_i) g(s(t))[s(t) + g(s(t))] x_{ij}^2 \text{.}
	\end{align*}
	\cite{Sampford1953} shows that
	$ 0 < g(s)[s + g(s)] < 1$ 
	for all $s \in \reals$.
	As such,
	\begin{equation*}
		L_{ij}''(t) = d_i x_{ij}^2 + (1-d_i) g(s(t))[s(t) + g(s(t))]x_{ij}^2 \leq x_{ij}^2
	\end{equation*}
	for all $t \in \reals$.
	Lemma 1 in \cite{Yang2013} establishes that if $f:\reals \to \reals$ is twice differentiable and there exists $M^*$ such that $f''(t) \leq M^*$ for all $t$, then $f$ satisfies the quadratic majorization condition with $M = M^*$.
	As such, $L_{ij}(t)$ satisfies the quadratic majorization condition with $M_{ij} =  x_{ij}^2$.
	
	As a consequence, we see that for $j = 0, 1, \ldots, p$,
	\begin{align*}
		\ell_n(\delta_j| \tilde{\dvec}, \tilde{\gamma})  
		& =  \frac{1}{n}\sum_{i = 1}^n L_{ij}(\delta_j) \\
		& \leq  \frac{1}{n} \sum_{i=1}^n L_{ij}(\tilde{\delta}_j) + L_{ij}'(\tilde{\delta}_j) (\delta_j - \tilde{\delta}_j) + \frac{1}{2} x_{ij}^2  (\delta_j - \tilde{\delta}_j)^2\\
		& =  \frac{1}{n} \sum_{i=1}^n L_{ij}(\tilde{\delta}_j) + \frac{1}{n} \sum_{i=1}^n  L_{ij}'(\tilde{\delta}_j) (\delta_j - \tilde{\delta}_j) + \frac{1}{2n} \sum_{i=1}^{n} x_{ij}^2 (\delta_j - \tilde{\delta}_j)^2 \\
		& = \ell_n(\tilde{\delta}_j| \tilde{\dvec}, \tilde{\gamma}) + \ell_n'(\tilde{\delta}_j| \tilde{\dvec}, \tilde{\gamma}) (\delta_j - \tilde{\delta}_j) + \frac{1}{2} \left(\frac{1}{n} \sum_{i=1}^{n} x_{ij}^2\right) (\delta_j - \tilde{\delta}_j)^2 \text{.} 
	\end{align*}
	That is,  $\ell_n(\delta_j| \tilde{\dvec}, \tilde{\gamma})$ satisfies the quadratic majorization condition with $M = \frac{1}{n} \sum_{i=1}^{n} x_{ij}^2$.
\end{proof}

\subsection{Convergence of the LLA algorithm}
We present our extensions of the results in \cite{Fan2014} early in this appendix because they are relatively straightforward and can stand on their own, unlike the considerably more involved proofs of Theorems \ref{thm:a0_prob} and \ref{thm:a1_a2_prob} which we present in the next section.
\begin{proof}[Proof of Theorem \ref{thm:one LLA step finds oracle}]
	We only need to make a single change to the proof of Theorem 1 in~\cite{Fan2014} to adapt it for the Tobit loss. We modify inequality (11) in their proof to account for $\delta_0$ and $\gamma$ as follows: since $\nabla_j \ell_n (\hat{\Theta}^{\oracle}) = 0 \text{\;} \forall j \in \mathcal{A}'$, we see that
	\begin{align*}
		\ell_n(\Theta) 
		& \geq \ell_n(\hat{\Theta}^{\oracle}) + \sum_{j = 0}^{p+1} \nabla_j \ell_n (\hat{\Theta}^{\oracle})(\theta_j - \hat{\theta}_j^{\oracle}) \\
		& = \ell_n(\hat{\Theta}^{\oracle}) + \sum_{j \in \mathcal{A}'^c} \nabla_j \ell_n (\hat{\Theta}^{\oracle})(\theta_j - \hat{\theta}_j^{\oracle}) \text{.}
	\end{align*}
	With this small modification, the remainder of the original proof holds.
\end{proof}

\begin{proof}[Proof of Theorem \ref{thm:two LLA steps converge to oracle}]
	The simple modification we made for the proof of Theorem 1 is also sufficient to adapt the proof of Theorem 2 in~\cite{Fan2014} for the Tobit loss.
\end{proof}

\subsection{Tobit lasso and folded concave Tobit}
The probability bound results for the Tobit lasso and Tobit with a folded concave penalty require far more involved proofs than the results we've covered so far.
Before we prove these results we must establish some properties of the Tobit log-likelihood for reference.
Recall that the Tobit log-likelihood is given by
$$\log L_n(\dvec, \gamma)  =  \sum_{i=1}^{n} d_i\left[ \log(\gamma) -\frac{1}{2} (\gamma y_i - \xvec_i' \dvec)^2 \right] + (1-d_i) \log\left( \Phi\left(- \xvec_i'\dvec \right) \right) $$
and that $g(s) = \phi(s)/\Phi(s)$. 
The gradient of the log-likelihood can be expressed as
$$\nabla \log L_n(\dvec, \gamma) = 
\begin{bmatrix}
	\xmat_1' (\gamma \yvec_1 - \xmat_1 \dvec) - \xmat_0' \mathbf{g}(\dvec) \\
	n_1 \gamma^{-1} - \yvec_1' (\gamma \yvec_1 - \xmat_1 \dvec)
\end{bmatrix}
$$
where $\mathbf{g}(\dvec)= (g(-\xvec_1'\dvec), \ldots, g(-\xvec_{n_0}'\dvec) )'$. Define $h(s) = g(s)(s + g(s)) $. We showed in the proof of Theorem \ref{thm:GCD majorization} that
$g'(s) = -g(s)(s + g(s)) = -h(s)$. We can write the Hessian of the Tobit log-likelihood as
\begin{equation}
	\nabla^2 \log L_n(\dvec, \gamma) = 
	- \begin{bmatrix}
		\xmat'\\
		- \yvec'
	\end{bmatrix}
	\begin{bmatrix}
		\mathbf{D}(\dvec) & \mathbf{0}\\
		\mathbf{0} & \mathbf{I}
	\end{bmatrix}
	\begin{bmatrix}
		\xmat & -\yvec
	\end{bmatrix}
	- \begin{bmatrix}
		\mathbf{0} & \mathbf{0}\\
		\mathbf{0} & n_1 \gamma^{-2}
	\end{bmatrix}
	\label{logL_Hessian}
\end{equation}
where $\mathbf{D}(\dvec)$ is a $n_0 \times n_0$ diagonal matrix with $[\mathbf{D}(\dvec)]_{ii} = h(-\xvec_i'\dvec)$.

\subsubsection{Probability bounding lemmas}
We rely on a bevy of lemmas to compute the probability bounds in Theorems \ref{thm:a0_prob} and \ref{thm:a1_a2_prob}. We present all of these lemmas together in this section as their proofs rely on similar arguments and so that readers can easily find them.

We leverage the properties of sub-Gaussian and sub-exponential random variables, which we define as follows:
\begin{defn}[Sub-Gaussian]
	We say that a random variable $X$ with $\E X = \mu$ is \emph{sub-Gaussian} with variance proxy $\sigma^2 \geq 0$ if its moment generating function satisfies
	$$ \E \left[ e^{t(X - \mu)}\right] \leq e^{\frac{\sigma^2 t^2}{2}} \text{ \; \; } \forall t \in \reals \text{.}$$
	We denote this by $X \sim \subG(\sigma^2)$.
\end{defn}

\begin{defn}[Sub-Exponential]
	We say that a random variable $X$ with $\E X = \mu$ is \emph{sub-exponential} with parameters $\sigma^2 \geq 0$, $\alpha \geq 0$ if its moment generating function satisfies
	$$ \E \left[ e^{t(X - \mu)}\right] \leq e^{\frac{\sigma^2 t^2}{2}} \text{ \; \; } \forall |t| < \frac{1}{\alpha} \text{.}$$
	We denote this by $X \sim \subExp(\sigma^2, \alpha)$.
\end{defn}

\begin{lemma}\label{lemma:subGaussian_score}
	Suppose that $Y^* = \xvec' \beta + \epsilon$ where $\epsilon \sim N(0, \sigma^2)$ and define $Y = Y^* \ind_{Y^* > 0}$. Let $\log L_1(\dvec, \gamma)$ denote the Tobit log-likelihood for a single observation. Then for $j \in \{0, \ldots, p \}$,
	$$ \E[e^{t \nabla_j \log L_1(\delta, \gamma) }] \leq e^{t^2  x_{j}^2/2} \text{\; \;} \forall t \in \reals \text{.} $$
	That is, $\nabla_j \log L_1(\delta, \gamma) \sim \subG(x_{j}^2)$.
\end{lemma}

\begin{proof}
	Let $j \in \{0, \ldots, p \}$. We can express the derivative of $\log L_1(\dvec, \gamma$) with respect to $\delta_j$ as
	\begin{equation*}
		\nabla_j \log L_1(\dvec, \gamma) = \left[ (\gamma Y^* - \xvec' \dvec) \ind_{Y^* > 0} - g(-\xvec'\dvec) \ind_{Y^* \leq 0} \right]x_{j} \text{.}
	\end{equation*}
	It is straightforward to show that $\E[\nabla_j \log L_1(\dvec, \gamma)] = 0$.
	Let $t \in \reals$. One can show that
	\begin{equation}
		\E [ e^{t \nabla_j \log L_1(\delta, \gamma) } ]
		= e^{ -g(-\xvec'\dvec)tx_j }\Phi(-\xvec'\dvec) 
		+ e^{t^2x_j^2/2}(1 - \Phi(-\xvec'\dvec - t x_j ) ) \text{.}
	\end{equation}
	Let $b,s \in \reals$.
	We see that $e^{-g(b)s}\Phi(b) \leq e^{s^2/2}\Phi(b-s)$ if and only if 
	$$-g(b)s + \log(\Phi(b))  = \log \left( e^{-g(b)s}\Phi(b) \right) \leq \log \left( e^{s^2/2}\Phi(b-s) \right) = s^2/2 + \log( \Phi(b-s) ) \text{.}$$
	Define
	$f(s) = s^2/2 + \log(\Phi(b-s)) + g(b)s - \log(\Phi(b)) $. 
	We see that 
	$f'(s) = s - g(b-s) + g(b)$
	and find $f'(0) = 0$ and $f(0) = 0$.
	Moreover, $f''(s) = 1 + g'(b-s) > 0$ for all $s \in \reals$, as $-1 < g'(b-s) < 0$.
	As such, $f(s)$ has a unique minimum at $s = 0$, meaning that $f(s) \geq 0$ for all $s \in \reals$.
	This implies that
	$ e^{-g(b)s}\Phi(b) \leq e^{s^2/2}\Phi(b-s) \text{\; \;} \forall b,s \in \reals \text{,}$
	so
	$$e^{-g(-\xvec'\dvec)t x_j}\Phi(-\xvec'\dvec) \leq e^{t^2 x_j^2/2}\Phi(-\xvec'\dvec-t x_j)$$
	and
	$$ \E[e ^{t \nabla_j \log L_1(\delta, \gamma)}] \leq e^{t^2  x_{j}^2/2} \text{\; \;} \forall t \in \reals \text{.} $$
	
\end{proof}

\begin{lemma}\label{lemma:subExponential_score_gamma}
	Suppose that $Y^* = \xvec' \beta + \epsilon$ where $\epsilon \sim N(0, \sigma^2)$ and define $Y = Y^* \ind_{Y^* > 0}$. Let $\log L_1(\dvec, \gamma)$ denote the Tobit log-likelihood for a single observation. Then
	$$ \E[e^{t \gamma \nabla_{\gamma} \log L_1(\dvec, \gamma) }] \leq e^{ \frac{(16 + 4 (\xvec ' \dvec)^2) t^2 }{2}} \text{\; \;} \forall |t| < \frac{1}{4}\text{.} $$
	That is, $\gamma \nabla_{\gamma} \log L_1(\dvec, \gamma) \sim \subExp( 16 + 4 (\xvec ' \dvec)^2  , 4)$.
\end{lemma}

\begin{proof}
	It is straightforward to show that $\E[\nabla_{\gamma} \log L_1(\dvec, \gamma)] = 0$.
	Using standard calculus techniques, we derive
	\begin{equation}
		\E[ e^{t \gamma \nabla_{\gamma} \log L_1(\dvec, \gamma) } ] 
		= \Phi( - \xvec' \dvec ) + (1 + 2t)^{-1/2} e^{t + \frac{t^2 (\xvec' \dvec)^2}{2(1 + 2t)}} \left( 1 - \Phi\left( - \frac{\xvec' \dvec (1 + t)}{(1 + 2t)^{1/2}} \right)  \right)
	\end{equation}
	for $t> - \frac{1}{2}$.
	First, we aim to show 
	\begin{equation}
		(1 + 2t)^{-1/2} e^{t + \frac{t^2 (\xvec' \dvec)^2}{2(1 + 2t)}} \leq e^{ \frac{(16 + 4 (\xvec'\dvec)^2) t^2}{2} } \text{ \: \: } \forall |t| < \frac{1}{4} \text{.}
		\label{eqn:subExp_bound1}
	\end{equation}
	We see that \eqref{eqn:subExp_bound1} holds if and only if
	\begin{equation}
		-\frac{1}{2}\log(1 + 2t) + t + \frac{t^2 (\xvec' \dvec)^2}{2(1 + 2t)} - \frac{(16 + 4 (\xvec'\dvec)^2) t^2}{2} \leq 0 \text{ \: \: } \forall |t| < \frac{1}{4} \text{.}
		\label{eqn:subExp_f_bound}
	\end{equation}
	Define $g(t) = -(1 + 2t)\log(1 + 2t) + 2(1 + 2t)t + t^2 (\xvec' \dvec)^2 - (16 + 4 (\xvec'\dvec)^2) t^2(1 + 2t) $.  Note that $g(0) = 0$. We derive
	\begin{align*}
		g'(t) 
		& = -2 \log(1 + 2t) - 2 + 2 + 8t - (16 + 3(\xvec' \dvec)^2)2t - (32 + 8(\xvec' \dvec)^2)3t^2 \\ 
		& = -2 \log(1 + 2t) - (24 + 6(\xvec' \dvec)^2)t - (96 + 24(\xvec' \dvec)^2)t^2 \text{.}
	\end{align*}
	Suppose that $t \in [ -  \frac{1}{4}, 0 ]$. Since $-2 \log(1 + 2t)$ is decreasing for all $t > -\frac{1}{2}$, we see that
	\begin{align*}
		g'(t) 
		& \geq -2\log(1)  - (24 + 6(\xvec' \dvec)^2)t - (96 + 24(\xvec' \dvec)^2)t^2 \\ 
		& = - t( 1 + 4t  ) (24 + 6(\xvec' \dvec)^2)\\
		& \geq 0 \text{.}
	\end{align*}
	Now let $t \in [0,  \frac{1}{4}]$. We see
	\begin{align*}
		g'(t) 
		& \leq -2\log(1)  - (24 + 6(\xvec' \dvec)^2)t - (96 + 24(\xvec' \dvec)^2)t^2 \\ 
		& = - t( 1 + 4t  ) (24 + 6(\xvec' \dvec)^2)\\
		& \leq 0 \text{.}
	\end{align*}
	Altogether, we have shown that $g(0) = 0$, $g(t)$ is increasing on $[ -  \frac{1}{4}, 0 ]$, and $g(t)$ is decreasing on $[0,  \frac{1}{4}]$. Together these imply that $g(t) \leq 0$ for $t \in [- \frac{1}{4}, \frac{1}{4}]$, which, in turn, implies \eqref{eqn:subExp_f_bound}. 
	
	Next, we aim to show
	\begin{equation}
		\Phi( - \xvec' \dvec ) \leq  e^{ \frac{(16 + 4 (\xvec'\dvec)^2) t^2}{2} } \Phi\left( - \frac{\xvec' \dvec (1 + t)}{(1 + 2t)^{1/2}} \right)  \text{ \: \: } \forall |t| < \frac{1}{4} \text{.}
		\label{eqn:subExp_bound2}
	\end{equation}
	Define $f(t) = \frac{(16 + 4 (\xvec'\dvec)^2) t^2}{2} + \log \left( \Phi\left( - \frac{\xvec' \dvec (1 + t)}{(1 + 2t)^{1/2}} \right) \right) - \log( \Phi( - \xvec' \dvec ) )$. 
	Note that $f(0) = 0$. 
	We find
	\begin{align*}
		f'(t) 
		& =  (16 + 4 (\xvec'\dvec)^2)t + g\left( - \frac{\xvec' \dvec (1 + t)}{(1 + 2t)^{1/2}} \right) \left( \frac{ (1 + 2t)^{1/2} - (1+t)(1 + 2t)^{-1/2} }{1 + 2t} \right)(- \xvec' \dvec )\\
		& = (16 + 4 (\xvec'\dvec)^2)t - \xvec' \dvec g\left( - \frac{\xvec' \dvec (1 + t)}{(1 + 2t)^{1/2}} \right)\frac{ (1 + 2t) - (1+t)}{(1 + 2t)^{3/2}} \\
		& = (16 + 4 (\xvec'\dvec)^2)t - \xvec' \dvec g\left( - \frac{\xvec' \dvec (1 + t)}{(1 + 2t)^{1/2}} \right)  \frac{ t }{(1 + 2t)^{3/2}} \\
		& = t \left[ 16 + 4 (\xvec'\dvec)^2 - \xvec' \dvec g\left( - \frac{\xvec' \dvec (1 + t)}{(1 + 2t)^{1/2}} \right)  (1 + 2t)^{-3/2} \right] \text{.}
	\end{align*}
	Define $v(t) = 16 + 4 (\xvec'\dvec)^2 - \xvec'\dvec g\left( - \frac{\xvec' \dvec (1 + t)}{(1 + 2t)^{1/2}} \right)  (1 + 2t)^{-3/2}$. 
	We will show $v(t) \geq 0$ for all $t \in [- \frac{1}{4}, \frac{1}{4}]$ by working through the following two cases separately: (i) $\xvec'\dvec < 0$ and (ii) $\xvec'\dvec \geq 0$.
	
	Suppose that $\xvec'\dvec < 0$. Recall that $g(s) \geq 0 $ for all $s \in \reals$. As a consequence, we have $- \xvec' \dvec g\left( - \frac{\xvec' \dvec (1 + t)}{(1 + 2t)^{1/2}} \right)  (1 + 2t)^{-3/2} \geq 0$ and, by extension, $v(t) \geq 16 + 4 (\xvec'\dvec)^2 \geq 0$ for $t \in [- \frac{1}{4}, \frac{1}{4}]$.
	
	Now suppose that $\xvec'\dvec \geq 0$. Lemma 2.1 of \cite{Kesavan1985} implies that
	$ g(-s) \leq s + \sqrt{\frac{2}{\pi} } $ for $s \geq 0$. Since $\frac{\xvec' \dvec (1 + t)}{(1 + 2t)^{1/2}} \geq 0$ for $t > -\frac{1}{2}$, we have
	$$ g\left( - \frac{\xvec' \dvec (1 + t)}{(1 + 2t)^{1/2}} \right) \leq \frac{\xvec' \dvec (1 + t)}{(1 + 2t)^{1/2}} + \sqrt{\frac{2}{\pi} } \text{.}$$
	Thus we find that when $\xvec'\dvec \geq 0$,
	\begin{align*}
		v(t) 
		& \geq 16 + 4 (\xvec'\dvec)^2 - \xvec' \dvec \left(\frac{\xvec' \dvec (1 + t)}{(1 + 2t)^{1/2}} + \sqrt{\frac{2}{\pi} } \right)  (1 + 2t)^{-3/2}\\
		& = 16 + 4 (\xvec'\dvec)^2 - (\xvec'\dvec)^2(1+t)(1 + 2t)^{-2} - \xvec'\dvec(1 + 2t)^{-3/2}\sqrt{\frac{2}{\pi}}
	\end{align*}
	for $t >-\frac{1}{2}$. Note that both $(1+t)(1 + 2t)^{-2}$ and $(1 + 2t)^{-3/2}$ are decreasing on $(-\frac{1}{2}, \infty)$, so $\max_{t \in [-1/4,1/4 ]}(1+t)(1 + 2t)^{-2} = (3/4)(1/2)^{-2} = 3$ and  $\max_{t \in [-1/4,1/4 ]}(1 + 2t)^{-3/2} = (1/2)^{-3/2} = 2\sqrt{2}$. Therefore
	\begin{align*}
		v(t) 
		& \geq 16 + 4 (\xvec'\dvec)^2 - 3(\xvec'\dvec)^2 - \xvec'\dvec2\sqrt{2}\sqrt{\frac{2}{\pi}} \\
		& =  (\xvec'\dvec)^2 - \xvec'\dvec \frac{4}{\sqrt{\pi}} + 16 \text{.}
	\end{align*}
	We see that $\frac{2}{\sqrt{\pi}} = \argmin_s s^2 - s \frac{4}{\sqrt{\pi}} + 16$ and $ (\frac{2}{\sqrt{\pi}})^2 -  (\frac{2}{\sqrt{\pi}})\frac{4}{\sqrt{\pi}} + 16 = - \frac{4}{\pi} + 16 > 0$. As such, $v(t) > 0$ for $t \in [- \frac{1}{4}, \frac{1}{4}]$ when $\xvec'\dvec \geq 0$.
	
	Since $v(t) \geq 0$ for $t \in [- \frac{1}{4}, \frac{1}{4}]$, we see that $f'(t) = tv(t) \leq 0$ for $ t \in [- \frac{1}{4}, 0]$ and $f'(t) \geq 0$ for $t \in [0, \frac{1}{4}]$. Taken together, our findings that $f(0) = 0$, $f(t)$ is decreasing on $[ -  \frac{1}{4}, 0 ]$, and $f(t)$ is increasing on $[0,  \frac{1}{4}]$ imply that $f(t) \geq 0$ for $t \in [- \frac{1}{4}, \frac{1}{4}]$. \eqref{eqn:subExp_bound2} immediately follows.
	
	Together \eqref{eqn:subExp_bound1} and \eqref{eqn:subExp_bound2} imply that
	\begin{equation}
		\E[e^{t \gamma \nabla_{\gamma} \log L_1(\delta, \gamma) }] \leq e^{ \frac{ (16 + 4 (\xvec ' \dvec)^2 ) t^2}{2} } \text{\; \;} \forall |t| < \frac{1}{4}\text{.}
	\end{equation}
\end{proof}

\begin{lemma}
	Suppose that $Y_i^* = \xvec_i' \beta + \epsilon_i$ where $\epsilon \iid N(0, \sigma^2)$ and define $Y_i = Y_i^* \ind_{Y_i^* > 0}$ for $i = 1, \ldots, n$. Then for $c > 0$
	\begin{align*}
		P\bigg( \bigg \lVert\nabla_{\mathcal{A}'}^2 & \log L_n(\Theta)  - \E\left[ \nabla_{\mathcal{A}'}^2 \log L_n(\Theta) \right] \bigg \rVert_{\max} > c \bigg)\\
		& \leq  2 (s + 1)^2 \exp \left( - \frac{2c^2}{ \max_{j,k\in\mathcal{A} \cup \{0\}} \sum_{i = 1}^{n} x_{ij}^2x_{ik}^2 } \right) \\
		& + 4(s+1) \exp \left( - \frac{2c^2\gamma^2}{ \max_{j\in\mathcal{A} \cup \{0\} } \sum_{i = 1}^{n} x_{ij}^2 (2 + \xvec_i'\dvec + g(- \xvec_i'\dvec) )^2 } \right)\\
		& + 2 \exp\left(-\frac{1}{2} \min \left\{ \frac{c\gamma^2}{8}, \frac{c^2 \gamma^4}{34n + \sum_{i = 1}^n \frac{1}{2} (\xvec_i'\dvec)^2 (2 + \xvec_i'\dvec + g(- \xvec_i'\dvec))^2 + 8 (\xvec_i'\dvec)^2} \right\} \right)  \text{.}
	\end{align*}
	\label{lemma:hessian prob bound}
\end{lemma}

\begin{proof}
	Let $c > 0$.
	Recall that the Tobit log-likelihood can be expressed as
	$$\log L_n(\dvec, \gamma)  =  \sum_{i=1}^{n} d_i\left[ \log(\gamma) -\frac{1}{2} (\gamma Y_i^* - \xvec_i' \dvec)^2 \right] + (1-d_i) \log\left( \Phi\left(- \xvec_i'\dvec \right) \right) $$
	where $d_i = \ind_{Y_i^* > 0}$.
	The Hessian with respect to $\Theta_{\mathcal{A}'}$, $\nabla_{\mathcal{A}'}^2 \log L_n(\Theta^*)$, has three types of entries: (i) $\frac{\partial^2}{\partial \delta_j \partial \delta_k }\log L_n (\Theta)$, (ii) $\frac{\partial^2}{\partial \delta_j \partial \gamma }\log L_n (\Theta)$, and (iii) $\frac{\partial^2}{\partial \gamma^2 }\log L_n (\Theta)$ where $j, k \in \mathcal{A} \cup \{0\}$. We will bound the upper tail probabilities for each type of entry using the Chernoff bound.
	
	Starting with (i), we  see that
	$$ \frac{\partial^2}{\partial \delta_j \partial \delta_k }\log L_n (\Theta) = \sum_{i = 1}^n - d_i x_{ij}x_{ik} - (1 - d_i) h(- \xvec_i'\dvec) x_{ij}x_{ik} = \sum_{i = 1}^n - x_{ij}x_{ik} [d_i + (1 - d_i)h(- \xvec_i'\dvec)] \text{.} $$
	Since $0 < h(s)<1 $ for all $s \in \reals$, we see that $0 < d_i + (1 - d_i)h(- \xvec_i'\dvec) < 1$. As such, by Hoeffding's lemma
	$ d_i + (1 - d_i)h(- \xvec_i'\dvec) \sim \subG(\frac{1}{4})$, so $- x_{ij}x_{ik} [d_i + (1 - d_i)h(- \xvec_i'\dvec)] \sim \subG(\frac{x_{ij}^2x_{ik}^2}{4})$. Since the $Y_i$ are independent, this implies 
	$\sum_{i = 1}^n - x_{ij}x_{ik} [d_i + (1 - d_i)h(- \xvec_i'\dvec)] \sim \subG(\frac{1}{4} \sum_{i=1}^n x_{ij}^2x_{ik}^2 ) $. Applying a Chernoff bound, we have
	\begin{equation}
		P\left( \left| \frac{\partial^2}{\partial \delta_j \partial \delta_k }\log L_n (\Theta) - \E\left[ \frac{\partial^2}{\partial \delta_j \partial \delta_k }\log L_n (\Theta) \right] \right| > c  \right) \leq 2 \exp \left( - \frac{2 c^2}{\sum_{i = 1}^{n} x_{ij}^2x_{ik}^2 } \right) \text{.}
	\end{equation}
	
	Moving on to (ii), we see that 
	$ \frac{\partial^2}{\partial \delta_j \partial \gamma }\log L_n (\Theta) = \sum_{i=1}^n d_i Y_i^* x_{ij}$. 
	A simple calculation yields $ \E[\gamma d_i Y_i^* x_{ij}] = (\xvec_i'\dvec(1 - \Phi(-\xvec_i'\dvec)) + \phi(-\xvec_i'\dvec) )x_{ij}$.
	We know that 
	$
	\frac{\partial}{\partial \delta_j } \log L_1 (\Theta) =  \gamma d_i  Y_i^* x_{ij} - d_i \xvec_i' \dvec x_{ij} - (1 - d_i) g(-\xvec'\dvec) x_{ij}.
	$
	As such, we see that
	\begin{align*}
		\gamma d_i Y_i^* x_{ij} - \E[\gamma d_i Y_i^* x_{ij}] 
		= & \text{ } \frac{\partial}{\partial \delta_j } \log L_1 (\Theta) + [d_i - (1 - \Phi(-\xvec_i'\dvec))]\xvec_i'\dvec x_{ij}\\ 
		& + [(1 - d_i) - \Phi(-\xvec_i'\dvec)]g(-\xvec_i'\dvec)x_{ij} \text{.}
	\end{align*}
	Lemma \ref{lemma:subGaussian_score} establishes that $\frac{\partial}{\partial \delta_j } \log L_1 (\Theta) \sim \subG(x_{ij}^2)$ and	Hoeffding's lemma yields that 
	$[d_i - (1 - \Phi(-\xvec_i'\dvec))]\xvec_i'\dvec x_{ij} \sim \subG\left( \frac{1}{4}(\xvec_i'\dvec)^2 x_{ij}^2 \right)$ 
	and 
	$[(1 - d_i) - \Phi(-\xvec_i'\dvec)]g(-\xvec_i'\dvec)x_{ij} \sim \subG\left( \frac{1}{4}g^2(-\xvec_i'\dvec)x_{ij}^2 \right)$. 
	With a simple application of Hoeffding's inequality one can show that if $Z_1 \sim \subG(\sigma_1^2)$ and $Z_2 \sim \subG(\sigma_1^2)$, then $Z_1 + Z_2 \sim \subG( (\sigma_1 + \sigma_2)^2 )$. 
	As such, we see that
	$
	\gamma d_i Y_i^* x_{ij} \sim \subG\left( \frac{1}{4}x_{ij}^2(2 + \xvec_i'\dvec + g(-\xvec_i'\dvec))^2  \right)
	$.
	Since the $Y_i$ are independent, this implies that
	$
	\frac{\partial^2}{\partial \delta_j \partial \gamma }\log L_n (\Theta) \sim \subG\left( \frac{1}{4\gamma^2} \sum_{i = 1}^n x_{ij}^2(2 + \xvec_i'\dvec + g(-\xvec_i'\dvec))^2 \right)
	$. Applying a Chernoff bound, we find
	\begin{multline}
		P\bigg( \bigg| \frac{\partial^2}{\partial \delta_j \partial \gamma }\log L_n (\Theta)  - \E\bigg[ \frac{\partial^2}{\partial \delta_j \partial \gamma }  \log L_n (\Theta) \bigg]  \bigg| >  c  \bigg)  \\ 
		\leq 2 \exp\left( - \frac{2 c^2 \gamma^2}{\sum_{i = 1}^n x_{ij}^2(2 + \xvec_i'\dvec + g(-\xvec_i'\dvec))^2} \right). 
	\end{multline}
	
	Finally, we move to (iii). We derive $\frac{\partial^2}{\partial \gamma^2 }\log L_n (\Theta) =\sum_{i = 1}^n d_i(-\gamma^{-2} - {Y_i^*}^2)$. One can easily show that $\E[d_i + \gamma^2 d_i {Y_i^*}^2] = \xvec_i'\dvec \phi(-\xvec_i'\dvec) + ((\xvec_i'\dvec)^2 + 2)(1 - \Phi(-\xvec_i'\dvec))$. We know that 
	$
	\gamma \frac{\partial}{\partial \gamma } \log L_1 (\Theta) = d_i - \gamma {Y_i^*}(\gamma {Y_i^*} - \xvec_i'\dvec )d_i = d_i - \gamma^2 {Y_i^*}^2 d_i + \gamma \xvec_i'\dvec {Y_i^*} d_i
	$. From  this, we see that
	{\small
		\begin{align*}
			-d_i - \gamma^2 {Y_i^*}^2 d_i - \E[-d_i - \gamma^2 {Y_i^*}^2 d_i ] = \gamma \frac{\partial}{\partial \gamma } \log L_1 (\Theta) - 2d_i + \E[2d_i] - \gamma \xvec_i'\dvec {Y_i^*} d_i + \E[\gamma \xvec_i'\dvec {Y_i^*} d_i] \text{.}
	\end{align*}}
	Lemma \ref{lemma:subExponential_score_gamma} establishes that $\gamma\frac{\partial}{\partial \gamma } \log L_1 (\Theta) \sim \subExp(16 + 4(\xvec_i'\dvec)^2, 4)$. 
	We've already shown that $\gamma \xvec_i'\dvec d_i Y_i^* \sim \subG\left( \frac{1}{4}(\xvec_i'\dvec)^2(2 + \xvec_i'\dvec + g(-\xvec_i'\dvec))^2  \right)$. Hoeffding's lemma yields that $2 d_i \sim \subG(1)$. With the Cauchy-Schwarz inequality, one can show that if $Z_1 \sim \subG(\sigma_1^2)$ and $Z_2 \sim \subExp(\sigma_2^2, \nu)$, then $Z_1 + Z_2 \sim \subExp(2(\sigma_1^2 + \sigma_2^2), 2 \nu)$. As such, we see that
	$$
	-d_i - \gamma^2 {Y_i^*}^2 d_i \sim \subExp\left(\frac{1}{2} (\xvec_i'\dvec)^2(2 + \xvec_i'\dvec + g(- \xvec_i'\dvec))^2 +34 + 8 (\xvec_i'\dvec)^2, 8\right) \text{.}
	$$
	Since the $Y_i$ are independent, this implies that 
	$$
	\gamma^2 \frac{\partial^2}{\partial \gamma^2 }\log L_n (\Theta) \sim \subExp\left(34n + \sum_{i = 1}^n \frac{1}{2} (\xvec_i'\dvec)^2(2 + \xvec_i'\dvec + g(- \xvec_i'\dvec))^2 + 8 (\xvec_i'\dvec)^2, 8\right) \text{.}
	$$
	Applying a Chernoff bound, we find that
	\begin{align*}
		P\bigg( \bigg| \frac{\partial^2}{\partial \gamma^2 }& \log L_n (\Theta)  - \E\bigg[ \frac{\partial^2}{\partial \gamma^2 }  \log L_n (\Theta) \bigg]  \bigg| >  c  \bigg)  \\
		& = P\bigg( \bigg|\gamma^2 \frac{\partial^2}{\partial \gamma^2 }\log L_n (\Theta)  - \E\bigg[ \gamma ^2 \frac{\partial^2}{\partial \gamma^2 }  \log L_n (\Theta) \bigg]  \bigg| >  c\gamma^2  \bigg) \\
		& \leq 2 \exp\left(-\frac{1}{2} \min \left\{ \frac{c\gamma^2}{8}, \frac{c^2 \gamma^4}{34n + \sum_{i = 1}^n \frac{1}{2} (\xvec_i'\dvec)^2(2 + \xvec_i'\dvec + g(- \xvec_i'\dvec))^2 + 8 (\xvec_i'\dvec)^2} \right\} \right) \text{.}
	\end{align*}
	We complete the proof by applying the union bound.
	
\end{proof}

\begin{lemma}\label{lemma:Q_3 matrix prob bound}
	Suppose that $Y_i^* = \xvec_i' \beta + \epsilon_i$ where $\epsilon \iid N(0, \sigma^2)$ and define $Y_i = Y_i^* \ind_{Y_i^* > 0}$ for $i = 1, \ldots, n$. Then for $c > 0$
	{\footnotesize 
		\begin{align*}
			P\bigg( \bigg \lVert 
			\xmat_{(\mathcal{A} \cup \{ 0 \})^c}'
			\begin{bmatrix}
				\mathbf{D}(\dvec) & \mathbf{0}\\
				\mathbf{0} & \mathbf{I}
			\end{bmatrix}
			&
			\begin{bmatrix}
				\xmat_{\mathcal{A} \cup \{ 0 \}} & -\yvec
			\end{bmatrix}
			- 
			\E\left[
			\xmat_{(\mathcal{A} \cup \{ 0 \})^c}'
			\begin{bmatrix}
				\mathbf{D}(\dvec) & \mathbf{0}\\
				\mathbf{0} & \mathbf{I}
			\end{bmatrix}
			\begin{bmatrix}
				\xmat_{(\mathcal{A} \cup \{ 0 \})} & -\yvec
			\end{bmatrix}
			\right]
			\bigg \rVert_{\max} > c \bigg)\\
			& \leq  2(s+1)(p-s) \exp \left( - \frac{2c^2}{ \max_{j \in (\mathcal{A} \cup \{0\})^c, k\in\mathcal{A} \cup \{0\}} \sum_{i = 1}^{n} x_{ij}^2x_{ik}^2 } \right) \\
			& + 2(p - s) \exp \left( - \frac{2c^2\gamma^2}{ \max_{j\in(\mathcal{A} \cup \{0\})^c } \sum_{i = 1}^{n} x_{ij}^2 (2 + \xvec_i'\dvec + g(- \xvec_i'\dvec) )^2 } \right) \text{.}
	\end{align*}}
\end{lemma}

\begin{proof}
	Note that each element of 
	$\xmat_{(\mathcal{A} \cup \{ 0 \})^c}'
	\begin{bmatrix}
		\mathbf{D}(\dvec) & \mathbf{0}\\
		\mathbf{0} & \mathbf{I}
	\end{bmatrix}
	\begin{bmatrix}
		\xmat_{(\mathcal{A} \cup \{ 0 \})} & -\yvec
	\end{bmatrix}$ 
	is equal to either $-\frac{\partial^2}{\partial \delta_j \partial \delta_k }\log L_n (\Theta)$ or $-\frac{\partial^2}{\partial \delta_j \partial \gamma }\log L_n (\Theta)$. As such, the arguments used in the proof of Lemma \ref{lemma:hessian prob bound} apply here as well, with the only change being that $j \in (\mathcal{A} \cup \{0\})^c$ rather than $\mathcal{A} \cup \{0\}$.
\end{proof}

\begin{lemma}
	Suppose that $Y_i^* = \xvec_i' \beta + \epsilon_i$ where $\epsilon_i \iid N(0, \sigma^2)$ and define $Y_i = Y_i^* \ind_{Y_i^* > 0}$ for $i = 1, \ldots, n$. Then for $c > 0$
	{\small \begin{align*}
			P\bigg( \bigg \lVert 
			\begin{bmatrix}
				- \xmat_1' \\ \yvec_1'
			\end{bmatrix} 
			&
			\begin{bmatrix}
				- \xmat_1 & \yvec_1
			\end{bmatrix}
			-
			\E \left[
			\begin{bmatrix}
				- \xmat_1' \\ \yvec_1'
			\end{bmatrix} 
			\begin{bmatrix}
				- \xmat_1 & \yvec_1
			\end{bmatrix}
			\right]
			\bigg \rVert_{\max} > c \bigg)\\
			& \leq  2(p+1)^2\exp\left(-\frac{2c^2}{\max_{j,k\in\{0, \ldots, p\}}\sum_{i = 1}^nx_{ij}^2x_{ik}^2} \right) \\
			& + 4(p+1)\exp\left(-\frac{2c^2\gamma^2}{\max_{j\in\{0, \ldots, p\}}\sum_{i = 1}^n x_{ij}^2(2 + \xvec_i'\dvec + g(-\xvec_i'\dvec))^2} \right)\\
			& + 2 \exp\left(-\frac{1}{2} \min \left\{ \frac{c\gamma^2}{8}, \frac{c^2 \gamma^4}{\frac{65}{2}n + \sum_{i = 1}^n \frac{1}{2} (\xvec_i'\dvec)^2(2 + \xvec_i'\dvec + g(- \xvec_i'\dvec))^2 + 8 (\xvec_i'\dvec)^2} \right\} \right) \text{.}
	\end{align*}}
	\label{lemma:RE matrix prob bound}
\end{lemma}

\begin{proof}
	Our argument will follow the same lines as the proof of Lemma \ref{lemma:hessian prob bound}.
	We see that 
	$\begin{bmatrix}
		- \xmat_1' \\ \yvec_1'
	\end{bmatrix}
	\begin{bmatrix}
		- \xmat_1 & \yvec_1
	\end{bmatrix}$
	has three types of entries: (i) $\sum_{i = 1}^n x_{ij}x_{ik}d_i$, (ii) $-\sum_{i = 1}^n x_{ij}{Y_i^*}d_i$, and (iii) $\sum_{i = 1}^n {Y_i^*}^2 d_i$
	where $j,k \in \{0, \ldots, p\}$. We will bound the upper tail probabilities for each type of entry.
	
	Starting with (i), Hoeffding's lemma provides that $x_{ij}x_{ik}d_i \sim \subG\left(\frac{1}{4} x_{ij}^2x_{ik}^2\right)$. Since the $Y_i$ are independent, this implies
	$\sum_{i = 1}^n x_{ij}x_{ik}d_i \sim \subG\left(\frac{1}{4}\sum_{i = 1}^nx_{ij}^2x_{ik}^2\right)$. Applying a Chernoff bound, we have
	$$
	P\left(\left|\sum_{i = 1}^n x_{ij}x_{ik}d_i - \E\left[\sum_{i = 1}^n x_{ij}x_{ik}d_i\right] \right| >c \right)
	\leq 2\exp\left(- \frac{2c^2}{\sum_{i = 1}^nx_{ij}^2x_{ik}^2} \right)\text{.}
	$$
	
	Note that we already covered (ii) in the proof of Lemma \ref{lemma:hessian prob bound}, as $-\sum_{i = 1}^n x_{ij}{Y_i^*}d_i = -\frac{\partial^2}{\partial \delta_j \partial \gamma }\log L_n (\Theta)$.
	
	Moving on to (iii), we know that
	$
	\gamma\frac{\partial}{\partial \gamma} \log L_1(\Theta) = d_i - \gamma^2 {Y_i^*}^2 d_i + \gamma	Y_i^* \xvec_i'\dvec d_i
	$. As such, 
	$
	\gamma^2 {Y_i^*}^2 d_i = -\gamma\frac{\partial}{\partial \gamma} \log L_1(\Theta) + d_i + \gamma Y_i^* \xvec_i'\dvec d_i
	$. We established in the proof of Lemma \ref{lemma:hessian prob bound} that $\gamma\frac{\partial}{\partial \gamma} \log L_1(\Theta) \sim \subExp(16 + 4(\xvec_i'\dvec)^2, 4)$ and $\gamma Y_i^* \xvec_i'\dvec d_i \sim \subG(\frac{1}{4} (\xvec_i'\dvec)^2(2 + \xvec_i'\dvec + g(-\xvec_i'\dvec))^2)$. Hoeffding's lemma provides that $d_i \sim \subG(\frac{1}{4})$. Applying our earlier findings from the proof of Lemma \ref{lemma:hessian prob bound} about sums of sub-exponential and sub-Gaussian random variables and the fact that the $Y_i$ are independent, we can conclude that
	$$
	\sum_{i =1}^n \gamma^2 {Y_i^*}^2 d_i \sim \subExp\left(\frac{65}{2}n + \sum_{i = 1}^n \frac{1}{2} (\xvec_i'\dvec)^2(2 + \xvec_i'\dvec + g(-\xvec_i'\dvec))^2 + 8(\xvec_i'\dvec)^2 ,8\right) \text{.}
	$$ From here, the Chernoff bound yields
	\begin{align}
		P\bigg( \bigg| \sum_{i =1}^n & {Y_i^*}^2 d_i   - \E\bigg[ \sum_{i =1}^n {Y_i^*}^2 d_i \bigg]  \bigg| >  c  \bigg) \notag \\
		& = P\bigg( \bigg| \sum_{i =1}^n \gamma^2 {Y_i^*}^2 d_i  - \E\bigg[ \sum_{i =1}^n \gamma^2 {Y_i^*}^2 d_i \bigg]  \bigg| > c \gamma^2  \bigg) \notag \\
		& \leq 2 \exp\left(-\frac{1}{2} \min \left\{ \frac{c\gamma^2}{8}, \frac{c^2 \gamma^4}{\frac{65}{2}n + \sum_{i = 1}^n \frac{1}{2} (\xvec_i'\dvec)^2(2 + \xvec_i'\dvec + g(- \xvec_i'\dvec))^2 + 8 (\xvec_i'\dvec)^2} \right\} \right) \text{.}
	\end{align}
	We apply the union bound to arrive at the final result.
\end{proof}

\subsubsection{Proof of Theorem \ref{thm:a0_prob}}
\begin{proof}[Proof of Theorem \ref{thm:a0_prob}]
	Let $\Delta \in \reals^{p+2}$. We partition $\Delta$ into $\Delta = (\Delta_{\dvec}', \Delta_{\gamma})'$ for ease of notation.
	We will begin by showing that for all $\Theta$ such that $\gamma > 0$ and $\Delta$ such that $\Delta_{\gamma} > -\gamma$, 
	$(\nabla \ell_n(\Theta + \Delta) - \nabla \ell_n(\Theta))'\Delta \geq \frac{1}{n} \norm{
		\begin{bmatrix}
			- \xmat_1 & \yvec_1
		\end{bmatrix}\Delta}_2^2
	$.
	We see that
	\begin{align*}
		(\nabla \ell_n(\Theta + \Delta) - \nabla \ell_n(\Theta))'\Delta
		= & \text{ } \frac{1}{n}
		\begin{bmatrix}
			-\xmat_1'(\Delta_{\gamma} \yvec_1 - \xmat_1 \Delta_{\dvec}) + \xmat_0' (\mathbf{g}(\dvec + \Delta_{\dvec}) - \mathbf{g}(\dvec)) \\
			-n_1 ((\gamma + \Delta_{\gamma})^{-1} - \gamma^{-1}) + \yvec_1'(\Delta_{\gamma} \yvec_1 - \xmat_1 \Delta_{\dvec})
		\end{bmatrix}' \Delta\\
		= & \text{ } \frac{1}{n} (\Delta_{\gamma} \yvec_1 - \xmat_1 \Delta_{\dvec})'(\Delta_{\gamma} \yvec_1 - \xmat_1 \Delta_{\dvec}) \\
		& + \frac{1}{n} (\xmat_0 \Delta_{\dvec} )' (\mathbf{g}(\dvec + \Delta_{\dvec}) - \mathbf{g}(\dvec)) - \frac{n_1}{n} ((\gamma + \Delta_{\gamma})^{-1} - \gamma^{-1})\Delta_{\gamma} \text{.}
	\end{align*}
	We see that $g(s) \geq 0$ for all $s \in \reals$ as $\phi(s) \geq 0$ and $\Phi(s) \geq 0$. Since $g'(s) = -h(s) \in [-1,0]$ (as we showed in the proof of Theorem \ref{thm:GCD majorization}), $g$ is decreasing on $\reals$. Suppose that $\xvec_i' \Delta_{\dvec} \geq 0$. Since $g$ is decreasing, we have
	$g(-\xvec_i ' \dvec -  \xvec_i'\Delta_{\dvec} ) \geq g(-\xvec_i ' \dvec )$ and, consequently, $(g(-\xvec_i ' \dvec -  \xvec_i'\Delta_{\dvec} ) - g(-\xvec_i ' \dvec ) )\xvec_i'\Delta_{\dvec} \geq 0$. Now suppose that $\xvec_i' \Delta_{\dvec} < 0$. Then $g(-\xvec_i ' \dvec -  \xvec_i'\Delta_{\dvec} ) \leq g(-\xvec_i ' \dvec )$, so $(g(-\xvec_i ' \dvec -  \xvec_i'\Delta_{\dvec} ) - g(-\xvec_i ' \dvec ) )\xvec_i'\Delta_{\dvec} \geq 0$. As such, we see that $\frac{1}{n} (\xmat_0 \Delta_{\dvec} )' (\mathbf{g}(\dvec + \Delta_{\dvec}) - \mathbf{g}(\dvec)) \geq 0$ for all vectors $\xmat_0 \Delta_{\dvec}$.
	Additionally, we recognize that $ - \frac{n_1}{n} ((\gamma + \Delta_{\gamma})^{-1} - \gamma^{-1})\Delta_{\gamma} = \frac{n_1}{n} \frac{\Delta_{\gamma}^2}{\gamma (\gamma + \Delta_{\gamma}) } \geq 0$, since $\gamma >0$ and $\Delta_{\gamma} > - \gamma$.
	As such, we have
	\begin{equation}
		(\nabla \ell_n(\Theta + \Delta) - \nabla \ell_n(\Theta))'\Delta 
		\geq  \frac{1}{n} (\Delta_{\gamma} \yvec_1 - \xmat_1 \Delta_{\dvec})'(\Delta_{\gamma} \yvec_1 - \xmat_1 \Delta_{\dvec}) = \frac{1}{n} \norm{
			\begin{bmatrix}
				- \xmat_1 & \yvec_1
			\end{bmatrix}\Delta}_2^2 \text{.}
	\end{equation}
	
	With this general inequality established, we can begin to assess the lasso estimator specifically. Let $\hat{\Delta} = \hat{\Theta}^{\lasso} - \Theta^*$. We know that $\hat{\Theta}^{\lasso}$ satisfies the following KKT conditions:
	\begin{equation*}
		\nabla \ell_n(\hat{\Theta}^{\lasso}) + \mathbf{v} = \mathbf{0}
	\end{equation*}
	where
	\begin{equation*}
		v_j
		\begin{cases}
			\in [-\lambda_{\lasso}, \lambda_{\lasso}] & \text{ if } \hat{\delta}_j^{\lasso} = 0 \text{, } j \in \{1, \ldots, p\}\\
			=  \lambda_{\lasso}\sgn(\hat{\delta}_j^{\lasso}) & \text{ if } \hat{\delta}_j^{\lasso} \neq 0 \text{, } j \in \{1, \ldots, p\}\\
			= 0 & \text{ if } j \in \{0, p+1\} \text{.}
		\end{cases}
	\end{equation*}
	Note that $\hat{\delta}_j^{\lasso} v_j = \lambda_{\lasso}|\hat{\delta}_j^{\lasso}|$ for $j = 1, \ldots, p$ and that  $\hat{\Delta}_{\mathcal{A}'^c} = \hat{\Theta}_{\mathcal{A}'^c}^{\lasso}$. 
	We apply H\"{o}lder's inequality to find
	\begin{align}
		0 
		& \leq \frac{1}{n} \norm{
			\begin{bmatrix}
				- \xmat_1 & \yvec_1
			\end{bmatrix}\hat{\Delta} }_2^2 \nonumber \\
		& \leq (\nabla \ell_n(\Theta^* + \hat{\Delta}) - \nabla \ell_n(\Theta^*))'\hat{\Delta} = (- \mathbf{v} - \nabla \ell_n(\Theta^*))'\hat{\Delta} \nonumber \\
		& = (- \mathbf{v}_{\mathcal{A}'} - \nabla_{\mathcal{A}'} \ell_n(\Theta^*))'\hat{\Delta}_{\mathcal{A}'} + (- \mathbf{v}_{\mathcal{A}'^c} - \nabla_{\mathcal{A}'^c} \ell_n(\Theta^*))'\hat{\Delta}_{\mathcal{A}'^c} \nonumber \\
		& = (- \mathbf{v}_{\mathcal{A}'} - \nabla_{\mathcal{A}'} \ell_n(\Theta^*))'\hat{\Delta}_{\mathcal{A}'} + (- \mathbf{v}_{\mathcal{A}'^c} - \nabla_{\mathcal{A}'^c} \ell_n(\Theta^*))'\hat{\Theta}_{\mathcal{A}'^c}^{\lasso} \nonumber \\
		& = (- \mathbf{v}_{\mathcal{A}'} - \nabla_{\mathcal{A}'} \ell_n(\Theta^*))'\hat{\Delta}_{\mathcal{A}'} - \mathbf{v}_{\mathcal{A}'^c}'\hat{\Theta}_{\mathcal{A}'^c}^{\lasso} - \nabla_{\mathcal{A}'^c} \ell_n(\Theta^*)'\hat{\Theta}_{\mathcal{A}'^c}^{\lasso} \nonumber \\
		& \leq \norm{- \mathbf{v}_{\mathcal{A}'} - \nabla_{\mathcal{A}'} \ell_n(\Theta^*)}_{\max} \norm{\hat{\Delta}_{\mathcal{A}'}}_1 - \lambda_{\lasso}\norm{\hat{\Theta}_{\mathcal{A}'^c}^{\lasso}}_1 + \norm{ - \nabla_{\mathcal{A}'^c} \ell_n(\Theta^*)}_{\max} \norm{\hat{\Theta}_{\mathcal{A}'^c}^{\lasso}}_1 \nonumber \\
		& \leq (\lambda_{\lasso} + \norm{\nabla_{\mathcal{A}'} \ell_n(\Theta^*)}_{\max})\norm{\hat{\Delta}_{\mathcal{A}'}}_1 + (- \lambda_{\lasso} + \norm{\nabla_{\mathcal{A}'^c} \ell_n(\Theta^*)}_{\max})\norm{\hat{\Delta}_{\mathcal{A}'^c}}_1 \label{a0_thm_ineq} \text{.}
	\end{align}
	Under the event 
	\begin{equation}
		\left\{ \norm{\nabla \ell_n(\Theta^*)}_{\max} \leq \frac{\lambda_{\lasso}}{2}  \right\} \text{,} \label{a0_proof_score}
	\end{equation}
	\eqref{a0_thm_ineq} implies
	\begin{equation*}
		\norm{\hat{\Delta}_{\mathcal{A}'^c}}_1 
		\leq \frac{\lambda_{\lasso} + \norm{\nabla_{\mathcal{A}'} \ell_n(\Theta^*)}_{\max}}{\lambda_{\lasso} - \norm{\nabla_{\mathcal{A}'^c} \ell_n(\Theta^*)}_{\max}}\norm{\hat{\Delta}_{\mathcal{A}'}}_1 
		\leq 3 \norm{\hat{\Delta}_{\mathcal{A}'}}_1 \text{,}
	\end{equation*}
	meaning that $\hat{\Delta} \in \mathcal{C}$. 
	
	For any $\mathbf{u} \in \mathcal{C}$,
	we see that
	\begin{align*}
		\bigg| \frac{1}{n} 
		\norm{
			\begin{bmatrix}
				- \xmat_1 & \yvec_1
			\end{bmatrix}
			\mathbf{u} }_{2}^2
		- & \frac{1}{n} 
		\E \left[ 
		\norm{
			\begin{bmatrix}
				- \xmat_1 & \yvec_1
			\end{bmatrix}
			\mathbf{u} }_{2}^2
		\right]
		\bigg| \\
		& = \left|  
		\mathbf{u}'\left(
		\frac{1}{n}
		\begin{bmatrix}
			- \xmat_1' \\ \yvec_1'
		\end{bmatrix}
		\begin{bmatrix}
			- \xmat_1 & \yvec_1
		\end{bmatrix}
		- \frac{1}{n}
		\E \left[ 
		\begin{bmatrix}
			- \xmat_1' \\ \yvec_1'
		\end{bmatrix}
		\begin{bmatrix}
			- \xmat_1 & \yvec_1
		\end{bmatrix}
		\right]
		\right)\mathbf{u}
		\right|\\
		& \leq  \norm{
			\frac{1}{n}
			\begin{bmatrix}
				- \xmat_1' \\ \yvec_1'
			\end{bmatrix}
			\begin{bmatrix}
				- \xmat_1 & \yvec_1
			\end{bmatrix}
			- \frac{1}{n}
			\E \left[
			\begin{bmatrix}
				- \xmat_1' \\ \yvec_1'
			\end{bmatrix} 
			\begin{bmatrix}
				- \xmat_1 & \yvec_1
			\end{bmatrix}
			\right]
		}_{\max} \norm{\mathbf{u}}_1^2 \text{.}
	\end{align*}
	Since $\mathbf{u} \in \mathcal{C}$,
	$
	\norm{\mathbf{u}}_1 = \norm{\mathbf{u}_{\mathcal{A}'}}_1 + \norm{\mathbf{u}_{\mathcal{A}'^c}}_1\leq 4\norm{\mathbf{u}_{\mathcal{A}'}}_1 \leq 4\sqrt{s+2} \norm{\mathbf{u}_{\mathcal{A}'}}_2 \leq 4\sqrt{s+2} \norm{\mathbf{u}}_2 
	$.
	Therefore, under the additional event
	\begin{equation}
		\left\{\norm{
			\frac{1}{n}
			\begin{bmatrix}
				- \xmat_1' \\ \yvec_1'
			\end{bmatrix}
			\begin{bmatrix}
				- \xmat_1 & \yvec_1
			\end{bmatrix}
			- \frac{1}{n}
			\E \left[
			\begin{bmatrix}
				- \xmat_1' \\ \yvec_1'
			\end{bmatrix} 
			\begin{bmatrix}
				- \xmat_1 & \yvec_1
			\end{bmatrix}
			\right]
		}_{\max} \leq \frac{\kappa}{32(s+2)} 
		\right\} \text{,} \label{event:max bound for RE}
	\end{equation}
	we have
	$$
	\left| \frac{
		\norm{
			\begin{bmatrix}
				- \xmat_1 & \yvec_1
			\end{bmatrix}
			\mathbf{u} }_{2}^2
	}{n\norm{\mathbf{u}}_2^2}
	- \frac{ 
		\E \left[ 
		\norm{
			\begin{bmatrix}
				- \xmat_1 & \yvec_1
			\end{bmatrix}
			\mathbf{u} }_{2}^2
		\right]
	}{n\norm{\mathbf{u}}_2^2}
	\right| \leq \frac{\kappa}{2} \text{.}
	$$
	By \eqref{RE_condition}, this implies that for all $\mathbf{u} \in \mathcal{C}$
	$$
	\frac{
		\norm{
			\begin{bmatrix}
				- \xmat_1 & \yvec_1
			\end{bmatrix}
			\mathbf{u} }_{2}^2
	}{n\norm{\mathbf{u}}_2^2}
	\geq \frac{ 
		\E \left[ 
		\norm{
			\begin{bmatrix}
				- \xmat_1 & \yvec_1
			\end{bmatrix}
			\mathbf{u} }_{2}^2
		\right]
	}{n\norm{\mathbf{u}}_2^2} - \frac{\kappa}{2}
	\geq  \frac{\kappa}{2} \text{.}
	$$
	
	\noindent Since $\hat{\Delta} \in \mathcal{C}$, we can use this result to tighten the lower bound in \eqref{a0_thm_ineq}:
	\begin{align*}
		\frac{\kappa}{2} \norm{\hat{\Delta} }_2^2 
		& \leq \frac{1}{n} \norm{
			\begin{bmatrix}
				- \xmat_1 & \yvec_1
			\end{bmatrix}\hat{\Delta} }_2^2 \\
		& \leq (\lambda_{\lasso} + \norm{\nabla_{\mathcal{A}'} \ell_n(\Theta^*)}_{\max})\norm{\hat{\Delta}_{\mathcal{A}'}}_1 + (- \lambda_{\lasso} + \norm{\nabla_{\mathcal{A}'^c} \ell_n(\Theta^*)}_{\max})\norm{\hat{\Delta}_{\mathcal{A}'^c}}_1 \\
		& \leq  \frac{3}{2} \lambda_{\lasso} \norm{\hat{\Delta}_{\mathcal{A}'}}_1 \\
		& \leq \frac{3}{2} \lambda_{\lasso} \sqrt{s+2} \norm{\hat{\Delta}_{\mathcal{A}'}}_2 \\
		& \leq \frac{3}{2} \lambda_{\lasso} \sqrt{s+2} \norm{\hat{\Delta}}_2 \text{.}
	\end{align*}
	Therefore if \eqref{a0_proof_score} and \eqref{event:max bound for RE} both hold, then $ \norm{\hat{\Delta} }_2 \leq \frac{3 \sqrt{s+2} \lambda_{\lasso}}{\kappa}$.
	
	We want to derive an upper bound for the probability that \eqref{a0_proof_score} does not hold.
	The union bound provides that
	$$
	P\left( \norm{\nabla \ell_n(\Theta^*)}_{\max} > \frac{\lambda_{\lasso}}{2}  \right)
	\leq \sum_{j =0}^{p+1} P\left( |\nabla_j \ell_n(\Theta^*)| > \frac{\lambda_{\lasso}}{2}  \right) \text{.}
	$$
	As such, we can handle each element of $\nabla \ell_n(\Theta^*)$ separately.
	Let $j \in \{0, 1, \ldots, p\} $. Lemma \ref{lemma:subGaussian_score} implies that $\nabla_j \log L_n(\Theta^*) \sim \subG(\norm{ \xvec_{(j)} }_2^2)$. From here, we apply the Chernoff bound to find
	\begin{equation}
		P\bigg( \left| \frac{1}{n} \nabla_j \log L_n(\Theta^*) \right|  >  \frac{\lambda_{\lasso}}{2}  \bigg)
		\leq 2 \exp \left( - \frac{ n^2 \lambda_{\lasso}^2 }{ 8  \norm{\xvec_{(j)} }_2^2 } \right) \text{.}
	\end{equation}
	As an immediate corollary to Lemma \ref{lemma:subExponential_score_gamma}, we have that $ \gamma^* \nabla_{\gamma} \log L_n(\dvec^*, \gamma^*) \sim \subExp( 16n + 4 \sum_{i = 1}^n (\xvec_i ' \dvec)^2, 4) $.
	Applying the Chernoff Bound, we find
	{\small \begin{align}
			P \bigg( \left| \frac{1}{n} \nabla_{\gamma} \log L_n(\Theta^*) \right| > & \frac{\lambda_{\lasso}}{2}  \bigg) \notag\\
			& = P \bigg( \bigg| \gamma^* \nabla_{\gamma} \log L_n(\Theta^*) \bigg| > n \gamma^* \frac{\lambda_{\lasso}}{2}  \bigg) \notag \\
			& \leq \begin{cases}
				2 \exp\left( \frac{ -n^2 \lambda_{\lasso}^2 {\gamma^*}^2  }{8(16n + 4 \sum_{i = 1}^n (\xvec_i'\dvec^*)^2 )} \right) &\text{ if } 0 \leq \frac{n\lambda_{\lasso}\gamma^*}{2} \leq \frac{16n + 4  \sum_{i = 1}^n (\xvec_i'\dvec^*)^2}{4} \text{,} \\
				2 \exp(\frac{-n \lambda_{\lasso} \gamma^*}{16}) &\text{ otherwise }
			\end{cases} \notag \\
			& = 2 \exp\left( - \frac{n}{2} \min \left\{ \frac{\lambda_{\lasso} \gamma^*}{8},  \frac{ \lambda_{\lasso}^2 {\gamma^*}^2  }{64 + 16 n^{-1} \sum_{i = 1}^n (\xvec_i'\dvec^*)^2 } \right \} \right) \text{.} 
	\end{align}}
	
	Pulling this all together, we have  
	{\small \begin{equation*}
			P(\eqref{a0_proof_score}^c) \leq 2(p+1)\exp\left( - \frac{n \lambda_{\lasso}^2 }{8M} \right) + 2 \exp \left( - \frac{n}{2} \min \left\{ \frac{\lambda_{\lasso} \gamma^*}{8}, \frac{\lambda_{\lasso}^2 {\gamma^*}^2 }{64 + 16n^{-1} \sum_{i = 1}^n (\xvec_i'\dvec^*)^2 } \right\} \right) \text{.}
	\end{equation*}}
	
	Additionally, we can use Lemma \ref{lemma:RE matrix prob bound} to bound the probability that \eqref{event:max bound for RE} does not hold:
	{\footnotesize \begin{align*}
			P(&\eqref{event:max bound for RE}^c) \\
			& \leq  2(p+1)^2\exp\left(-\frac{2n \left(\frac{\kappa}{32(s+2)}\right)^2}{ \max_{j,k\in\{0, \ldots, p\}} n^{-1}\sum_{i = 1}^nx_{ij}^2x_{ik}^2} \right) \\
			& + 4(p+1)\exp\left(-\frac{2n \left(\frac{\kappa}{32(s+2)}\right)^2{\gamma^*}^2}{\max_{j\in\{0, \ldots, p\}}n^{-1}\sum_{i = 1}^n x_{ij}^2(2 + \xvec_i'\dvec^* + g(-\xvec_i'\dvec^*))^2} \right)\\
			& + 2 \exp\left(-\frac{n}{2} \min \left\{ \frac{\left(\frac{\kappa}{32(s+2)}\right){\gamma^*}^2}{8}, \frac{\left(\frac{\kappa}{32(s+2)}\right)^2 {\gamma^*}^4}{\frac{65}{2} + n^{-1}\sum_{i = 1}^n \frac{1}{2} (\xvec_i'\dvec^*)^2(2 + \xvec_i'\dvec^* + g(- \xvec_i'\dvec^*))^2 + 8 (\xvec_i'\dvec^*)^2} \right\} \right)  \text{.}
	\end{align*}}
	We then apply the union bound to arrive at our final result.
	
\end{proof}

\subsubsection{Additional lemmas for Theorem \ref{thm:a1_a2_prob}}
\begin{lemma}\label{lemma:g_2nd_bound}
	$\left| g''(s) \right| < 4.3$ for all $s \in \reals$.
\end{lemma}

\begin{proof}
	We note that $g''(s) = -h'(s)$.
	\cite{Sampford1953} showed that $0 < h(s) < 1$ and $h'(s) <0$ for all $s \in \reals$. As such, we only need a lower bound for $h'(s)$. 
	We see that
	\begin{equation*}
		h'(s) 
		= -h(s)(s+g(s)) + g(s)(1 - h(s)) 
		\geq -h(s)(s+g(s))
	\end{equation*}
	since $g(s) \geq 0$ and $0 < h(s) < 1$.
	We will find lower bounds for $h'(s)$ in two cases: (i) $s\leq 0$ and (ii) $s > 0$.
	
	Let $s \leq 0$. Since $g(s)$ is decreasing, we have
	$
	-h'(s) 
	\leq \frac{h^2(s)}{g(s)} 
	< \frac{1}{g(0)} < 1.26 \text{.}
	$
	
	Let $s > 0$. We will start by showing that there exists some constant $c$ such that $h(s) < \frac{c}{s+1}$. We note that $\Phi(s) > \frac{1}{2}$ for $s > 0$. As such,
	$
	h(s)  = \frac{\phi(s)}{\Phi(s)}\left( s + \frac{\phi(s)}{\Phi(s)} \right)
	< 2 \phi(s) (s + 2 \phi(s)) \text{.}
	$
	We will derive bounds for $2s\phi(s)$ and $4\phi^2(s)$ separately.
	
	We start with $2s\phi(s)$. One can show that if there exists some $k$ such that $s^2 - 2\log(s(s+1)) >  k $, then $2s\phi(s) < \frac{2  e^{-k/2}}{\sqrt{2 \pi}} \frac{1}{s + 1}$. We note that for $s >0$
	$$ s^2 - 2\log(s(s+1)) > s^2 - 4\log\left(s+\frac{1}{2}\right) \text{.}$$
	Simple calculus yields that $ \frac{-1 + \sqrt{33} }{4} = \argmin_{s >0} s^2 - 4\log\left(s+\frac{1}{2}\right)$. This gives us $\min_{s > 0} s^2 - 2\log(s(s+1)) > -0.69$ and, by extension, 
	$2s\phi(s) < \frac{1.6}{s+1} \text{.}$
	
	We use a similar approach to bound $4\phi^2(s)$.
	One can show that if there exists some $k$ such that $s^2 - \log(s+1) >  k $, then $4\phi^2(s) < \frac{2  e^{-k}}{\pi } \frac{1}{s + 1}$. Using simple calculus we find $\min_{s > 0} s^2 - \log(s+1) > -0.18$ and, by extension, 
	$4\phi^2(s) < \frac{0.77}{s+1} \text{.}$
	
	Thus $h(s) < \frac{2.37}{s+1}$ for $s > 0$. 
	Leveraging this result and the fact that $g(s)$ and $\frac{1}{s+1}$ are decreasing for $s > 0$, we find that
	\begin{align*}
		- h'(s) 
		& \leq h(s)(s + g(s))\\
		& < \frac{2.37}{s+1}(s+g(s)) \\
		& < 2.37\left( \frac{s}{s+1} + \frac{g(s)}{s+1} \right)\\
		& < 2.37\left( 1 + \frac{g(s)}{s+1} \right)\\
		& < 2.37 (1 + g(0)) < 4.3 \text{.}
	\end{align*}
	Considering our two cases together, we see that $\left| h'(s) \right| < 4.3$ for all $s \in \reals$.
\end{proof}

\begin{lemma}\label{lemma:hessian bound to inverse bound}
	Suppose that the Tobit log-likelihood satisfies \emph{\ref{hessian positive definite}}. If
	{\small\begin{equation}
			\norm{\frac{1}{n}\nabla_{\mathcal{A}'}^2 \log L_n(\Theta^*)  - \E\left[ \frac{1}{n} \nabla_{\mathcal{A}'}^2 \log L_n(\Theta^*) \right] }_{\max} \leq \min\left\{ \frac{1}{6K_1K_2(s+2)}, \frac{1}{6K_1^3K_2(s+2)} \right\} \text{,} 
			\label{event:hessian bound}
	\end{equation}}
	then
	{\small\begin{multline*}
			\norm{  \left( \frac{1}{n} \nabla_{\mathcal{A}'}^2 \log L_n(\Theta^*) \right)^{-1}  - \left( \E\left[ \frac{1}{n} \nabla_{\mathcal{A}'}^2 \log L_n(\Theta^*) \right] \right)^{-1}   }_{\max}\\ 
			\leq 2K_2\norm{\frac{1}{n}\nabla_{\mathcal{A}'}^2 \log L_n(\Theta^*)  - \E\left[ \frac{1}{n} \nabla_{\mathcal{A}'}^2 \log L_n(\Theta^*) \right] }_{\max} \text{.}
	\end{multline*}}
\end{lemma}

\begin{proof}
	Define the map $F: \VEC(\mathbb{B}(r)) \to \reals^{(s+2)^2}$ by
	\begin{multline}
		F(\VEC(\Delta)) = (H_{\mathcal{A}', \mathcal{A}'}^*)^{-1} \bigg[ \VEC\left( \left( \left( \E\left[ \frac{1}{n} \nabla_{\mathcal{A}'}^2 \log L_n(\Theta^*) \right] \right)^{-1}  + \Delta  \right)^{-1} \right) \\ - \VEC\left( \frac{1}{n} \nabla_{\mathcal{A}'}^2 \log L_n(\Theta^*) \right)  \bigg] + \VEC(\Delta) \label{Fvec_definition}
	\end{multline}
	where $\mathbb{B}(r)$ is the convex compact set $\mathbb{B}(r)= \{\Delta: \norm{\Delta}_{\max} \leq r \} \subseteq \reals^{(s+2) \times (s+2)}$ with \\
	$
	r = 2K_2\norm{ \frac{1}{n} \nabla_{\mathcal{A}'}^2 \log L_n(\Theta^*)  - \E\left[ \frac{1}{n} \nabla_{\mathcal{A}'}^2 \log L_n(\Theta^*) \right] }_{\max}
	$ and $\VEC(\mathbb{B}(r)) = \{ \VEC(\Delta) : \Delta \in \mathbb{B}(r) \} \subseteq \reals^{(s+2)^2} $.
	We aim to show that $F(\VEC(\mathbb{B}(r))) \subseteq \VEC(\mathbb{B}(r))$ under \eqref{event:hessian bound}.
	
	Let $\Delta \in \mathbb{B}(r)$. Define $J = \sum_{j = 0}^{\infty} (-1)^j \left( \E\left[ \frac{1}{n} \nabla_{\mathcal{A}'}^2 \log L_n(\Theta^*) \right] \Delta\right)^j$. 
	We recognize that for two matrices $\mathbf{A}$ and $\mathbf{B}$, $ \norm{\mathbf{AB} }_{\infty} \leq \norm{ \mathbf{A} }_{\infty} \norm{ \mathbf{B} }_1$. As such, we see that under \eqref{event:hessian bound} we have
	{\small \begin{align*}
			\norm{ \E\left[ \frac{1}{n} \nabla_{\mathcal{A}'}^2 \log L_n(\Theta^*) \right] \Delta  }_{\infty}
			& \leq \norm{ \E\left[ \frac{1}{n} \nabla_{\mathcal{A}'}^2 \log L_n(\Theta^*) \right] }_{\infty} \norm{\Delta}_{1}\\
			& \leq K_1 \norm{\Delta}_{1}\\
			& \leq K_1 (s+2)r\\
			& = K_1 (s+2)2K_2\norm{\frac{1}{n}\nabla_{\mathcal{A}'}^2 \log L_n(\Theta^*)  - \E\left[ \frac{1}{n} \nabla_{\mathcal{A}'}^2 \log L_n(\Theta^*) \right] }_{\max}\\
			& \leq \frac{1}{3} \text{.}
	\end{align*}}
	Therefore $J$ is a convergent matrix series, with $J = (I + \E\left[ \frac{1}{n} \nabla_{\mathcal{A}'}^2 \log L_n(\Theta^*) \right] \Delta)^{-1}$. 
	We use the series expansion of $J$ to rewrite a key part of $F( \VEC(\Delta) ) $:
	\begin{align}
		\bigg( \bigg( \E\bigg[ \frac{1}{n} \nabla_{\mathcal{A}'}^2 & \log L_n (\Theta^*) \bigg] \bigg)^{-1}  +  \Delta  \bigg)^{-1} \notag \\
		& = \left\{ \left( \E\left[ \frac{1}{n} \nabla_{\mathcal{A}'}^2 \log L_n(\Theta^*) \right] \right)^{-1} \left(I + \E\left[ \frac{1}{n} \nabla_{\mathcal{A}'}^2 \log L_n(\Theta^*) \right] \Delta \right) \right\}^{-1} \notag \\
		& = \left(I + \E\left[ \frac{1}{n} \nabla_{\mathcal{A}'}^2 \log L_n(\Theta^*) \right] \Delta \right)^{-1} \E\left[ \frac{1}{n} \nabla_{\mathcal{A}'}^2 \log L_n(\Theta^*) \right] \notag \\
		& =  \E\left[ \frac{1}{n} \nabla_{\mathcal{A}'}^2 \log L_n(\Theta^*) \right] - \E\left[ \frac{1}{n} \nabla_{\mathcal{A}'}^2 \log L_n(\Theta^*) \right]\Delta\E\left[ \frac{1}{n} \nabla_{\mathcal{A}'}^2 \log L_n(\Theta^*) \right] \notag \\
		& \text{ \; \; } + \sum_{j = 2}^{\infty} (-1)^j \left( \E\left[ \frac{1}{n} \nabla_{\mathcal{A}'}^2 \log L_n(\Theta^*) \right] \Delta\right)^j \E\left[ \frac{1}{n} \nabla_{\mathcal{A}'}^2 \log L_n(\Theta^*) \right] \notag \\
		& =  \E\left[ \frac{1}{n} \nabla_{\mathcal{A}'}^2 \log L_n(\Theta^*) \right] - \E\left[ \frac{1}{n} \nabla_{\mathcal{A}'}^2 \log L_n(\Theta^*) \right]\Delta\E\left[ \frac{1}{n} \nabla_{\mathcal{A}'}^2 \log L_n(\Theta^*) \right] \notag \\
		& \text{ \; \; } +  \left( \E\left[ \frac{1}{n} \nabla_{\mathcal{A}'}^2 \log L_n(\Theta^*) \right] \Delta\right)^2 J \E\left[ \frac{1}{n} \nabla_{\mathcal{A}'}^2 \log L_n(\Theta^*) \right] \text{.}  \label{plug in for F}
	\end{align} 
	For ease of presentation, define $R(\Delta) = \left( \E\left[ \frac{1}{n} \nabla_{\mathcal{A}'}^2 \log L_n(\Theta^*) \right] \Delta\right)^2 J \E\left[ \frac{1}{n} \nabla_{\mathcal{A}'}^2 \log L_n(\Theta^*) \right]$. We can express \eqref{plug in for F} in a vectorized form:
	\begin{multline*}
		\VEC\bigg( \bigg( \left( \E\left[ \frac{1}{n} \nabla_{\mathcal{A}'}^2 \log L_n(\Theta^*) \right] \right)^{-1} + \Delta  \bigg)^{-1} \bigg) - \VEC \left( \frac{1}{n} \nabla_{\mathcal{A}'}^2 \log L_n(\Theta^*) \right) \\
		= \VEC\left(\E\left[ \frac{1}{n} \nabla_{\mathcal{A}'}^2 \log L_n(\Theta^*) \right]\right) - \VEC \left( \frac{1}{n} \nabla_{\mathcal{A}'}^2 \log L_n(\Theta^*) \right)\\
		- \VEC\left( \E\left[ \frac{1}{n} \nabla_{\mathcal{A}'}^2 \log L_n(\Theta^*) \right]\Delta\E\left[ \frac{1}{n} \nabla_{\mathcal{A}'}^2 \log L_n(\Theta^*) \right] \right) + \VEC(R(\Delta)) \text{.}
	\end{multline*}
	Applying the ``vec-trick'' for the Kronecker product, we find
	$$
	\VEC\left( \E\left[ \frac{1}{n} \nabla_{\mathcal{A}'}^2 \log L_n(\Theta^*) \right]\Delta\E\left[ \frac{1}{n} \nabla_{\mathcal{A}'}^2 \log L_n(\Theta^*) \right] \right) = H_{\mathcal{A}', \mathcal{A}'}^* \VEC(\Delta) \text{,}
	$$
	giving us
	\begin{align}
		\VEC\bigg( \bigg(  \bigg( \E\bigg[ \frac{1}{n} \nabla_{\mathcal{A}'}^2 \log L_n & (\Theta^*) \bigg] \bigg)^{-1}  + \Delta  \bigg)^{-1} \bigg) - \VEC \left( \frac{1}{n} \nabla_{\mathcal{A}'}^2 \log L_n(\Theta^*) \right) \notag \\
		& = \VEC\left(\E\left[ \frac{1}{n} \nabla_{\mathcal{A}'}^2 \log L_n(\Theta^*) \right]\right) - \VEC \left( \frac{1}{n} \nabla_{\mathcal{A}'}^2 \log L_n(\Theta^*) \right) \notag \\
		& \text{\; \; } -  H_{\mathcal{A}', \mathcal{A}'}^* \VEC(\Delta) + \VEC(R(\Delta)) \text{.} \label{vec_simplification}
	\end{align}
	Plugging \eqref{vec_simplification} into \eqref{Fvec_definition}, we have
	{\small $$
		F(\VEC(\Delta)) = (H_{\mathcal{A}', \mathcal{A}'}^*)^{-1} \left[ \VEC\left(\E\left[ \frac{1}{n} \nabla_{\mathcal{A}'}^2 \log L_n(\Theta^*) \right]\right) - \VEC \left( \frac{1}{n} \nabla_{\mathcal{A}'}^2 \log L_n(\Theta^*) \right) + \VEC(R(\Delta)) \right] \text{.}
		$$}
	
	With this new expression for $F(\VEC(\Delta))$ we aim show that $\norm{F(\VEC(\Delta))}_{\max} \leq r$. We start by following the proof of Lemma 5 in \cite{Ravikumar2011} to bound $\norm{R(\Delta)}_{\max}$.
	Let $e_i$ denote the vector with $1$ in position $i$ and zeros elsewhere. We have
	\begin{align*}
		\norm{R(\Delta)}_{\max} 
		& = \max_{i,j} |e_i' \left( \E\left[ \frac{1}{n} \nabla_{\mathcal{A}'}^2 \log L_n(\Theta^*) \right] \Delta\right)^2 J \E\left[ \frac{1}{n} \nabla_{\mathcal{A}'}^2 \log L_n(\Theta^*) \right] e_j | \\
		& \leq \max_{i} \norm{e_i'  \E\left[ \frac{1}{n} \nabla_{\mathcal{A}'}^2 \log L_n(\Theta^*) \right] \Delta}_{\max} \\
		& \text{ \; \; }\times \max_{j} \norm{  \E\left[ \frac{1}{n} \nabla_{\mathcal{A}'}^2 \log L_n(\Theta^*) \right] \Delta J \E\left[ \frac{1}{n} \nabla_{\mathcal{A}'}^2 \log L_n(\Theta^*) \right] e_j }_{1}\\
		& \leq \max_{i} \norm{e_i'  \E\left[ \frac{1}{n} \nabla_{\mathcal{A}'}^2 \log L_n(\Theta^*) \right]}_{1} \norm{\Delta}_{\max}\\
		& \text{ \; \; }\times \max_{j} \norm{  \E\left[ \frac{1}{n} \nabla_{\mathcal{A}'}^2 \log L_n(\Theta^*) \right] \Delta J \E\left[ \frac{1}{n} \nabla_{\mathcal{A}'}^2 \log L_n(\Theta^*) \right] e_j }_{1}\\
		& \leq \norm{\E\left[ \frac{1}{n} \nabla_{\mathcal{A}'}^2 \log L_n(\Theta^*) \right]}_{\infty} \norm{\Delta}_{\max}\\
		& \text{ \; \; }\times \norm{  \E\left[ \frac{1}{n} \nabla_{\mathcal{A}'}^2 \log L_n(\Theta^*) \right] \Delta J \E\left[ \frac{1}{n} \nabla_{\mathcal{A}'}^2 \log L_n(\Theta^*) \right] }_1 \text{.}
	\end{align*}
	We know that for any matrix $ \mathbf{A} $, $\norm{\mathbf{A}}_1 = \norm{\mathbf{A}'}_{\infty}$. As such, by the sub-multiplicativity of the $\ell_{\infty}$-norm we see that
	\begin{align*}
		\norm{R(\Delta)}_{\max} 
		& \leq \norm{\E\left[ \frac{1}{n} \nabla_{\mathcal{A}'}^2 \log L_n(\Theta^*) \right]}_{\infty} \norm{\Delta}_{\max}\\
		& \text{ \; \; }\times \norm{  \E\left[ \frac{1}{n} \nabla_{\mathcal{A}'}^2 \log L_n(\Theta^*) \right] J' \Delta' \E\left[ \frac{1}{n} \nabla_{\mathcal{A}'}^2 \log L_n(\Theta^*) \right] }_{\infty}\\
		& \leq \norm{\E\left[ \frac{1}{n} \nabla_{\mathcal{A}'}^2 \log L_n(\Theta^*) \right]}_{\infty} \norm{\Delta}_{\max}\\
		& \text{ \; \; }\times \norm{  \E\left[ \frac{1}{n} \nabla_{\mathcal{A}'}^2 \log L_n(\Theta^*) \right]}_{\infty} \norm{J'}_{\infty} \norm{\Delta'}_{\infty} \norm{\E\left[ \frac{1}{n} \nabla_{\mathcal{A}'}^2 \log L_n(\Theta^*) \right] }_{\infty}\\
		& \leq \norm{\Delta}_{\max} K_1^3 \norm{J'}_{\infty} \norm{\Delta}_1\\
		& \leq (s+2) \norm{\Delta}_{\max}^2 K_1^3 \norm{J'}_{\infty} \text{.}
	\end{align*}
	Focusing on $J$, we find
	\begin{align*}
		\norm{J'}_{\infty} 
		& \leq \sum_{j = 0}^\infty \norm{\Delta' \E\left[ \frac{1}{n} \nabla_{\mathcal{A}'}^2 \log L_n(\Theta^*) \right]}_{\infty}^j \\
		& \leq \frac{1}{1 - \norm{\Delta' \E\left[ \frac{1}{n} \nabla_{\mathcal{A}'}^2 \log L_n(\Theta^*) \right]}_{\infty}} \\
		& \leq \frac{1}{1 - \norm{\Delta'}_{\infty} \norm{\E\left[ \frac{1}{n} \nabla_{\mathcal{A}'}^2 \log L_n(\Theta^*) \right]}_{\infty}} \\
		& = \frac{1}{1 - \norm{\Delta}_{1} \norm{\E\left[ \frac{1}{n} \nabla_{\mathcal{A}'}^2 \log L_n(\Theta^*) \right]}_{\infty}}\\
		& \leq \frac{3}{2}
	\end{align*}
	since $\norm{\Delta}_{1} \norm{\E\left[ \frac{1}{n} \nabla_{\mathcal{A}'}^2 \log L_n(\Theta^*) \right]}_{\infty} \leq \frac{1}{3}$, as we showed previously. Plugging this into our bound for $\norm{R(\Delta)}_{\max}$, we have
	$
	\norm{R(\Delta)}_{\max} \leq \frac{3}{2} (s+2) \norm{\Delta}_{\max}^2 K_1^3  \leq \frac{3}{2} (s+2)r^2 K_1^3 \text{.}
	$
	
	Returning to our expression for $F(\VEC(\Delta))$, we see that under \eqref{event:hessian bound}	
	{\small \begin{align*}
			\lVert F(&\VEC(\Delta))\rVert_{\max} \\
			& = \norm{
				(H_{\mathcal{A}', \mathcal{A}'}^*)^{-1} \left[ \VEC\left(\E\left[ \frac{1}{n} \nabla_{\mathcal{A}'}^2 \log L_n(\Theta^*) \right]\right) - \VEC \left( \frac{1}{n} \nabla_{\mathcal{A}'}^2 \log L_n(\Theta^*) \right) + \VEC(R(\Delta)) \right]
			}_{\max}\\
			& \leq K_2 \left( \norm{ \frac{1}{n} \nabla_{\mathcal{A}'}^2 \log L_n(\Theta^*) - \E\left[ \frac{1}{n} \nabla_{\mathcal{A}'}^2 \log L_n(\Theta^*) \right]}_{\max}  + \norm{R(\Delta)}_{\max} \right)\\
			& \leq K_2\norm{ \frac{1}{n} \nabla_{\mathcal{A}'}^2 \log L_n(\Theta^*) - \E\left[ \frac{1}{n} \nabla_{\mathcal{A}'}^2 \log L_n(\Theta^*) \right]}_{\max}  + \frac{3}{2} (s+2)r^2 K_1^3K_2 \\
			& \leq r \text{,}
	\end{align*}}
	meaning that $F(\VEC(\mathbb{B}(r))) \subseteq \VEC(\mathbb{B}(r))$. 
	
	Since $F(\VEC(\mathbb{B}(r))) \subseteq \VEC(\mathbb{B}(r))$ under \eqref{event:hessian bound}, Brouwer's fixed point theorem guarantees that there exists a fixed point $\hat{\Delta} \in \mathbb{B}(r)$ such that $F(\VEC(\hat{\Delta} )) = \VEC(\hat{\Delta} )$. That is, $\hat{\Delta}$ satisfies
	{\small $$
		\mathbf{0} = (H_{\mathcal{A}', \mathcal{A}'}^*)^{-1} \bigg[ \VEC\left( \left( \left( \E\left[ \frac{1}{n} \nabla_{\mathcal{A}'}^2 \log L_n(\Theta^*) \right] \right)^{-1}  + \hat{\Delta}  \right)^{-1} \right) \\ - \VEC\left( \frac{1}{n} \nabla_{\mathcal{A}'}^2 \log L_n(\Theta^*) \right)  \bigg]\text{.}
		$$}
	Because $(H_{\mathcal{A}', \mathcal{A}'}^*)^{-1}$ is positive definite, we see that this condition holds only if $\hat{\Delta} = \left( \frac{1}{n} \nabla_{\mathcal{A}'}^2 \log L_n(\Theta^*) \right)^{-1} - \left( \E\left[ \frac{1}{n} \nabla_{\mathcal{A}'}^2 \log L_n(\Theta^*) \right] \right)^{-1} $.
	As such, we can conclude that when \eqref{event:hessian bound} holds
	$$
	\norm{\left( \frac{1}{n} \nabla_{\mathcal{A}'}^2 \log L_n(\Theta^*) \right)^{-1} - \left( \E\left[ \frac{1}{n} \nabla_{\mathcal{A}'}^2 \log L_n(\Theta^*) \right] \right)^{-1}}_{\max} = \norm{\hat{\Delta}}_{\max} \leq r \text{.}
	$$
\end{proof}

\subsubsection{Proof of Theorem \ref{thm:a1_a2_prob}}
\begin{proof}[Proof of Theorem \ref{thm:a1_a2_prob}]
	Recall our two events of interest $\mathcal{E}_1 = \{|| \nabla_{\mathcal{A}'^c} \ell_n (\hat{\Theta}^{\oracle}) ||_{\max} < a_1 \lambda \}$ and
	$\mathcal{E}_2 = \{ || \hat{\dvec}_{\mathcal{A}}^{\oracle} ||_{\min} > a\lambda\}$.
	We define the following key events for the proof:\\
	{\small $
		\mathcal{E}_1^* =  \left\{ \norm{ \frac{1}{n} \nabla_{\mathcal{A}'} \log L_n(\Theta^*) }_{\max} \leq C_1(s)  \right \} \text{,} 
		$
		$
		\mathcal{E}_2^* =  \bigg\{ \norm{\frac{1}{n}\nabla_{\mathcal{A}'}^2 \log L_n(\Theta^*)  - \E\left[ \frac{1}{n} \nabla_{\mathcal{A}'}^2 \log L_n(\Theta^*) \right]   }_{\max} 
		\leq \min\left\{ \frac{1}{6K_1K_2(s+2)}, \frac{1}{6K_1^3K_2(s+2)}, \frac{Q_2}{4 K_2 (s+2)}  \right\} 
		\bigg \}\text{,} 
		$
		$
		\mathcal{E}_3^* =  \left\{ \norm{ \frac{1}{n} \nabla_{\mathcal{A}'} \log L_n(\Theta^*) }_{\max} < \frac{1}{3 Q_2} ( \norm{ \dvec_{\mathcal{A}}^* }_{\min} - a\lambda  ) \right \}\text{,} 
		$
		$
		\mathcal{E}_4^* =  \left\{ \norm{ \frac{1}{n} \nabla_{\mathcal{A}'^c} \log L_n(\Theta^*) }_{\max} < \frac{a_1 \lambda}{2}\right\} \cap 
		\left\{ \norm{ \frac{1}{n} \nabla_{\mathcal{A}'} \log L_n(\Theta^*) }_{\max} \leq \frac{a_1 \lambda}{9 Q_3 + 2} \right\} \text{,} 
		$
		and}
	\begin{multline*}
		\mathcal{E}_5^* =  \bigg\{  
		\bigg \lVert 
		\begin{smallmatrix}
			\frac{1}{n}
			\xmat_{(\mathcal{A} \cup \{ 0 \})^c}'
		\end{smallmatrix}
		\bigl[\begin{smallmatrix}
			\mathbf{D}(\dvec^*) & \mathbf{0}\\
			\mathbf{0} & \mathbf{I}
		\end{smallmatrix}\bigr]
		\bigl[\begin{smallmatrix}
			\xmat_{(\mathcal{A} \cup \{ 0 \})} & -\yvec
		\end{smallmatrix}\bigr]
		- 
		\E\left[
		\begin{smallmatrix}
			\frac{1}{n}
			\xmat_{(\mathcal{A} \cup \{ 0 \})^c}'
		\end{smallmatrix}
		\bigl[\begin{smallmatrix}
			\mathbf{D}(\dvec^*) & \mathbf{0}\\
			\mathbf{0} & \mathbf{I}
		\end{smallmatrix}\bigr]
		\bigl[\begin{smallmatrix}
			\xmat_{(\mathcal{A} \cup \{ 0 \})} & -\yvec
		\end{smallmatrix}\bigr]
		\right]
		\bigg \rVert_{\max} \\
		\leq \frac{1}{2(s+2)}
		\norm{
			\E\left[
			\begin{smallmatrix}
				\frac{1}{n}
				\xmat_{(\mathcal{A} \cup \{ 0 \})^c}'
			\end{smallmatrix}
			\bigl[\begin{smallmatrix}
				\mathbf{D}(\dvec^*) & \mathbf{0}\\
				\mathbf{0} & \mathbf{I}
			\end{smallmatrix}\bigr]
			\bigl[\begin{smallmatrix}
				\xmat_{(\mathcal{A} \cup \{ 0 \})} & -\yvec
			\end{smallmatrix}\bigr]
			\right]
		}_{\infty}
		\bigg\}\text{.} 
	\end{multline*}
	We introduce these new events because we can easily bound $P( {\mathcal{E}_1^*}^c ), \ldots, P( {\mathcal{E}_5^*}^c )$
	---a luxury $\mathcal{E}_1$ and $\mathcal{E}_2$ do not afford us. 
	Our strategy for bounding $P(\mathcal{E}_1^c)$ and $P((\mathcal{E}_1 \cap \mathcal{E}_2)^c)$ will be to show that
	$\mathcal{E}_1^* \cap \mathcal{E}_2^* \cap \mathcal{E}_3^* \subseteq \mathcal{E}_2$ and 
	$\mathcal{E}_1^* \cap \mathcal{E}_2^* \cap \mathcal{E}_4^* \cap \mathcal{E}_5^* \subseteq \mathcal{E}_1$ 
	and then leverage the fact that, as a consequence, $P(\mathcal{E}_1^c) \leq P( (\mathcal{E}_1^* \cap \mathcal{E}_2^* \cap \mathcal{E}_4^* \cap \mathcal{E}_5^*)^c)$ and $P((\mathcal{E}_1 \cap \mathcal{E}_2)^c) \leq P( (\mathcal{E}_1^* \cap \mathcal{E}_2^* \cap \mathcal{E}_3^* \cap \mathcal{E}_4^* \cap \mathcal{E}_5^*)^c)$.
	
	We will start by showing that $\mathcal{E}_1^* \cap \mathcal{E}_2^* \cap \mathcal{E}_3^* \subseteq \mathcal{E}_2$.
	Throughout this section of the proof, keep in mind that $C_1(s) = \min \left\{ \frac{1}{38.7 Q_1 Q_2^2 (s+1)}, \frac{ {\gamma^*}^3 }{144 Q_2^2}, \frac{\gamma^*}{6 Q_2} \right \}$.
	From \eqref{oracle_prop}, we know that 
	$$ \nabla_{\mathcal{A}'}\log L_n(\hat{\Theta}^{\oracle} ) = \mathbf{0}\text{.} $$
	Define $F: \mathbb{B}(r) \to \reals^{p+2} $ by
	$F(\Delta) = (F_{\mathcal{A}'}(\Delta), \mathbf{0}) $ 
	where 
	\begin{equation}
		F_{\mathcal{A}'}(\Delta) = - ( \nabla_{\mathcal{A}'}^2 \log L_n(\Theta^*) )^{-1} \nabla_{\mathcal{A}'}  \log L_n(\Theta^* + \Delta) + \Delta_{A'} \label{F_defn}
	\end{equation}
	and $\mathbb{B}(r)$ is the convex compact set $\mathbb{B}(r) = \{\Delta \in \reals^{p+2} : \norm{ \Delta_{\mathcal{A}'}  }_{\max} \leq r, \Delta_{\mathcal{A}'^c} = \mathbf{0}\}$ 
	with $r = 3 Q_2 \norm{ \frac{1}{n} \nabla_{\mathcal{A}'} \log L_n(\Theta^*) }_{\max}$.
	
	We aim to show that $F(\mathbb{B}(r)) \subseteq \mathbb{B}(r)$ when both $\mathcal{E}_1^*$ and $\mathcal{E}_2^*$ hold.
	Let $\Delta \in \mathbb{B}(r)$. Applying the mean value theorem with respect to $\theta_j$, for $j = 0, 1, \ldots, p+1$, we have
	\begin{equation}
		\nabla_j \log L_n(\Theta^* + \Delta) = \nabla_j \log L_n(\Theta^*) + [\nabla^2 \log L_n(\Theta^*)]_j \Delta + R_j(\tilde{\Delta}^{(j)})
	\end{equation}
	where $[\nabla^2 \log L_n(\Theta^*)]_j$ denotes the $j$th row of $\nabla^2 \log L_n(\Theta^*)$, $\tilde{\Delta}^{(j)}$ is on the line segment joining $\mathbf{0}$ and $\Delta$, and
	$R_j(\tilde{\Delta}^{(j)}) =  [\nabla^2 \log L_n(\Theta^* + \tilde{\Delta}^{(j)}) - \nabla^2 \log L_n(\Theta^*)]_j \Delta \text{.}$
	Define $R(\tilde{\Delta}) = (R_0(\tilde{\Delta}^{(0)}), R_1(\tilde{\Delta}^{(1)}), \ldots, R_{p+1}(\tilde{\Delta}^{(p+1)}) )' $. Then
	\begin{equation}
		\nabla \log L_n(\Theta^* + \Delta) = \nabla \log L_n(\Theta^*) + \nabla^2 \log L_n(\Theta^*) \Delta + R(\tilde{\Delta}) \text{.}
		\label{mvt_for_score}
	\end{equation} 
	For ease of notation, we rewrite $R(\tilde{\Delta})$ as $R(\tilde{\Delta}) = (R_{\mathcal{A}'}'(\tilde{\Delta}),R_{\mathcal{A}'^c}'(\tilde{\Delta}))'$.
	Since $\Delta_{\mathcal{A}'^c} = \mathbf{0}$, we see that $[\nabla^2 \log L_n(\Theta^*)]_{\mathcal{A}'} \Delta = \nabla_{\mathcal{A}'}^2 \log L_n(\Theta^*) \Delta_{\mathcal{A}'} $. 
	As such, we have
	\begin{equation}
		\nabla_{\mathcal{A}'} \log L_n(\Theta^* + \Delta) = \nabla_{\mathcal{A}'} \log L_n(\Theta^*) + \nabla_{\mathcal{A}'}^2 \log L_n(\Theta^*) \Delta_{\mathcal{A}'} + R_{\mathcal{A}'}(\tilde{\Delta}) \text{.} \label{MVT_for_logL}
	\end{equation}
	Plugging \eqref{MVT_for_logL} back into \eqref{F_defn}, we find
	\begin{align*}
		F_{\mathcal{A}'}(\Delta) 
		& = - ( \nabla_{\mathcal{A}'}^2 \log L_n(\Theta^*) )^{-1} ( \nabla_{\mathcal{A}'} \log L_n(\Theta^*) + \nabla_{\mathcal{A}'}^2 \log L_n(\Theta^*) \Delta_{\mathcal{A}'} + R_{\mathcal{A}'}(\tilde{\Delta}) ) + \Delta_{A'} \\
		& = - ( \nabla_{\mathcal{A}'}^2 \log L_n(\Theta^*) )^{-1} ( \nabla_{\mathcal{A}'} \log L_n(\Theta^*) + R_{\mathcal{A}'}(\tilde{\Delta}) )
	\end{align*}
	
	We will use this expression for $F_{\mathcal{A}'}(\Delta)$ to show that $\norm{ F_{\mathcal{A}'}(\Delta)  }_{\max} \leq r$. We start by bounding $\norm{ R(\tilde{\Delta}) }_{\max}$.  We partition $\Delta$ into $\Delta = (\Delta_{\dvec}', \Delta_{\gamma})'$. Let $j \in \{0,1,\ldots,p \} $. Based on \eqref{logL_Hessian}, we see that
	$$ [ \nabla^2 \log L_n(\Theta)]_j = 
	- \xvec_{(j)}'
	\begin{bmatrix}
		\mathbf{D}(\dvec) & \mathbf{0}\\
		\mathbf{0} & \mathbf{I}
	\end{bmatrix}
	\begin{bmatrix}
		\xmat & -\yvec
	\end{bmatrix} \text{.}$$
	Since $\Delta_{A'^c} = \mathbf{0}$, we can simplify $R_j(\tilde{\Delta}^{(j)})$. We find that
	\begin{align*}
		R_j(\tilde{\Delta}^{(j)}) 
		& =
		- \xvec_{(j)}'
		\begin{bmatrix}
			\mathbf{D}(\dvec^* + \tilde{\Delta}_{\dvec}^{(j)} ) - \mathbf{D}(\dvec^* ) & \mathbf{0}\\
			\mathbf{0} & \mathbf{0}
		\end{bmatrix}
		\begin{bmatrix}
			\xmat & -\yvec
		\end{bmatrix} \Delta \\
		& = - \xmat_{0, j}'
		\begin{bmatrix}
			\mathbf{D}(\dvec^* + \tilde{\Delta}_{\dvec}^{(j)} ) - \mathbf{D}(\dvec^* )
		\end{bmatrix}
		\xmat_{0, \mathcal{A} \cup \{0\} } \Delta_{\dvec_{\mathcal{A} \cup \{0\}}} \text{.}
	\end{align*}
	Applying the mean value theorem, we derive the following inequality:
	$$|R_j(\tilde{\Delta}^{(j)})| \leq \Delta_{\dvec_{\mathcal{A} \cup \{0\}}}'\xmat_{0, \mathcal{A} \cup \{0\}}'\diag\{ |\xmat_{0, j}| \circ |g''(\dvec^* + \bar{\Delta}_{\dvec}^{(j)} ) | \} \xmat_{0, \mathcal{A} \cup \{0\}}  \Delta_{\dvec_{\mathcal{A} \cup \{0\}}} \text{,}$$
	where $\bar{\Delta}_{\dvec}^{(j)}$ is on the line segment joining $\mathbf{0}$ and $\tilde{\Delta}^{(j)} $.
	Lemma \ref{lemma:g_2nd_bound} provides that $|g''(s)| < 4.3$ for all $s$, giving us
	\begin{equation}
		|R_j(\tilde{\Delta}^{(j)})| \leq 4.3 n \lambda_{\max}\left( \frac{1}{n} \xmat_{0, \mathcal{A} \cup \{0\}}'\diag\{ |\xmat_{0, j}|\} \xmat_{0, \mathcal{A} \cup \{0\}}   \right) \norm{\Delta_{\dvec_{\mathcal{A} \cup \{0\}}}}_2^2  \text{.} \label{eqn:Rj MVT bound}
	\end{equation}
	
	Because $ \xmat_{0, \mathcal{A} \cup \{0\}}'\diag\{ |\xmat_{0, j}|\} \xmat_{0, \mathcal{A} \cup \{0\}}$ is a random matrix, we must take the additional step of bounding its largest eigenvalue with a non-random quantity.
	We see that 
	\begin{align*}
		\xmat_{ (\mathcal{A} \cup \{0\}) }' \diag\{ |\xmat_{(j)} |\} \xmat_{(\mathcal{A} \cup \{0\})}
		& = \begin{bmatrix}
			\xmat_{0, \mathcal{A} \cup \{0\}}' & \xmat_{1, \mathcal{A} \cup \{0\}}'
		\end{bmatrix}
		\diag\{ |\xmat_{(j)} |\}
		\begin{bmatrix}
			\xmat_{0, \mathcal{A} \cup \{0\}} \\  \xmat_{1, \mathcal{A} \cup \{0\}}
		\end{bmatrix}\\
		& = \xmat_{0, \mathcal{A} \cup \{0\}}' \diag\{ |\xmat_{0, j} |\} \xmat_{0, \mathcal{A} \cup \{0\}} \\ 
		& \text{ \; \; } + \xmat_{1, \mathcal{A} \cup \{0\}}' \diag\{ |\xmat_{1, j} |\} \xmat_{1, \mathcal{A} \cup \{0\}} \text{.}
	\end{align*}
	It is straightforward to show that if $\mathbf{A} \in \reals ^{n \times p}$ and $ \mathbf{D} \in \reals^{n \times n}$ is a diagonal matrix with non-negative entries, then $\mathbf{A}'\mathbf{D}\mathbf{A}$ is symmetric and positive semidefinite.
	Weyl's inequality provides that if $\mathbf{A} \in \reals^{n \times n}$ and $\mathbf{B} \in \reals^{n \times n}$ are Hermitian, then $\lambda_{\max}( \mathbf{A} + \mathbf{B} ) \geq \lambda_{\max}(\mathbf{A}) + \lambda_{\min}(\mathbf{B})$. Additionally, if $\mathbf{B}$ is positive semidefinite, then $\lambda_{\max}(\mathbf{A}) + \lambda_{\min}(\mathbf{B}) \geq \lambda_{\max}(\mathbf{A})$. As such, for any $j \in \{0, \ldots, p\}$, we have
	$\lambda_{\max}\left(\frac{1}{n}\xmat_{0, \mathcal{A} \cup \{0\}}' \diag\{ |\xmat_{0, j} |\} \xmat_{0, \mathcal{A} \cup \{0\}}   \right)
	\leq \lambda_{\max}\left(\frac{1}{n} \xmat_{ (\mathcal{A} \cup \{0\}) }' \diag\{ |\xmat_{(j)} |\} \xmat_{(\mathcal{A} \cup \{0\})}   \right)
	\leq Q_1
	$. Plugging this last inequality into \eqref{eqn:Rj MVT bound}, we have 
	$|R_j(\tilde{\Delta}^{(j)})| \leq 4.3 n Q_1 \norm{\Delta_{\dvec_{\mathcal{A} \cup \{0\}}}}_2^2  \text{.}$
	Finally, we note that  
	$\norm{\Delta_{\dvec_{\mathcal{A} \cup \{0\}}}}_2^2 \leq (s+1)r^2$ and, by extension,
	$
	|R_j(\tilde{\Delta}^{(j)})| \leq 4.3 n Q_1 (s+1)r^2 \text{.}
	$
	
	Next, we bound $|R_{p+1}(\tilde{\Delta}^{(p+1)})|$. From \eqref{logL_Hessian}, we know that
	$$[ \nabla^2 \log L_n(\Theta)]_{p+1} =
	\begin{bmatrix}
		\yvec_1' \xmat_1 & - \yvec_1 '\yvec_1 - n_1 \gamma^{-2} 
	\end{bmatrix}
	$$
	As such, we see that
	\begin{align*}
		R_{p+1}(\tilde{\Delta}^{(p+1)}) 
		& = \begin{bmatrix}
			\mathbf{0} &  - \yvec_1 '\yvec_1 - n_1 (\gamma^* + \tilde{\Delta}_{\gamma}^{(p+1)} )^{-2} - ( - \yvec_1 '\yvec_1 - n_1 {\gamma^*}^{-2})
		\end{bmatrix} \Delta
		\\
		& = - n_1 ( (\gamma^* + \tilde{\Delta}_{\gamma}^{(p+1)} )^{-2} - {\gamma^*}^{-2}) \Delta_{\gamma} \text{.}
	\end{align*}
	Applying the mean value theorem, we find 
	$ |R_{p+1}(\tilde{\Delta}^{(p+1)})| \leq 2 n_1 (\gamma^* + \bar{\Delta}_{\gamma} ^{(p+1)} )^{-3} \Delta_{\gamma}^2$, where $\bar{\Delta}_{\gamma} ^{(p+1)}$ is on the line segment joining $\mathbf{0}$ and $\tilde{\Delta}_{\gamma}^{(p+1)}$. If $\mathcal{E}_1^*$ holds, then by the definitions of $r$ and $C_1(s)$, $|\bar{\Delta}_{\gamma}^{(p+1)}| \leq |\Delta_{\gamma}| \leq r \leq \frac{\gamma^*}{2} $ and, consequently,
	$|R_{p+1}(\tilde{\Delta}^{(p+1)})| \leq 16 n_1 {\gamma^*}^{-3}  \Delta_{\gamma}^2 \text{.}$
	
	Pulling this all together, we see that under $\mathcal{E}_1^*$, 
	\begin{equation}
		\norm{R(\tilde{\Delta})}_{\max} \leq \max \{ 4.3n Q_1 (s+1) r^2, 16 n {\gamma^*}^{-3} r^2\} \text{.} \label{eqn:R delta bound}
	\end{equation}
	We apply the triangle inequality to find
	\begin{align}
		\norm{ F_{\mathcal{A}'}(\Delta) }_{\max} 
		& = \norm{  - ( \nabla_{\mathcal{A}'}^2 \log L_n(\Theta^*) )^{-1} ( \nabla_{\mathcal{A}'} \log L_n(\Theta^*) + R_{\mathcal{A}'}(\tilde{\Delta}) )  }_{\max} \notag \\ 
		& \leq\widehat{Q}_2 \norm{ \frac{1}{n} \nabla_{\mathcal{A}'} \log L_n(\Theta^*) }_{\max} + \widehat{Q}_2 \norm{ \frac{1}{n} R_{\mathcal{A}'}(\tilde{\Delta})}_{\max} \text{,} \label{eqn:F bound with Q2 hat}
	\end{align}
	where $\widehat{Q}_2 = \norm{ \left(\frac{1}{n} \nabla_{\mathcal{A}'}^2 \log L_n(\Theta^*) \right)^{-1} }_{\infty}$. Lemma \ref{lemma:hessian bound to inverse bound} implies that if $\mathcal{E}_2^*$ holds, then
	$$
	\norm{  \left( \frac{1}{n}\nabla_{\mathcal{A}'}^2 \log L_n(\Theta^*) \right)^{-1}  - \left( \E\left[ \frac{1}{n}\nabla_{\mathcal{A}'}^2 \log L_n(\Theta^*) \right] \right)^{-1}   }_{\max} \leq \frac{Q_2}{2(s+2)} \text{.}
	$$
	As an immediate consequence of this, we see that
	for $1 \leq i \leq s+2$
	$$
	\sum_{j} \left| \left[\left( \frac{1}{n} \nabla_{\mathcal{A}'}^2 \log L_n(\Theta^*) \right)^{-1}\right]_{i,j} \right| \leq \sum_{j} \left| \left[\left( \E\left[ \frac{1}{n} \nabla_{\mathcal{A}'}^2 \log L_n(\Theta^*) \right] \right)^{-1}\right]_{i,j} \right|
	+ \frac{Q_2}{2} \leq \frac{3}{2}Q_2
	$$
	and, by extension, 
	$
	\widehat{Q}_2 \leq \frac{3}{2}Q_2
	$. Returning to \eqref{eqn:F bound with Q2 hat}, we see that when both $\mathcal{E}_1^*$ and $\mathcal{E}_2^*$ hold
	\begin{align*}
		\norm{ F_{\mathcal{A}'}(\Delta) }_{\max} 
		& \leq \frac{3}{2} Q_2 \norm{ \frac{1}{n} \nabla_{\mathcal{A}'} \log L_n(\Theta^*) }_{\max} + \frac{3}{2} Q_2 \norm{ \frac{1}{n} R_{\mathcal{A}'}(\tilde{\Delta})}_{\max} \\
		& \leq \frac{r}{2} + \frac{3}{2}Q_2 \max \{ 4.3 Q_1 (s+1) r^2, 16 {\gamma^*}^{-3} r^2\}\\
		& \leq r \text{,}
	\end{align*}
	where the last inequality follows immediately from the definitions of $r$ and $C_1(s)$.
	Thus we have shown that $ F(\mathbb{B}(r)) \subseteq \mathbb{B}(r) $ under $\mathcal{E}_1^*$ and $\mathcal{E}_2^*$.
	
	Since $ F(\mathbb{B}(r)) \subseteq \mathbb{B}(r) $, Brouwer's fixed point theorem guarantees that there exists a fixed point $\hat{\Delta} \in \mathbb{B}(r)$ such that $F(\hat{\Delta} ) = \hat{\Delta} $. As a result, $ \nabla_{\mathcal{A}'}\log L_n(\Theta^* + \hat{\Delta}) = \mathbf{0}$ and $\hat{\Delta}_{\mathcal{A}'^c} = \mathbf{0}$. Since $\hat{\Theta}^{\oracle}$ is the unique solution to $ \nabla_{\mathcal{A}'}\log L_n(\Theta ) = \mathbf{0}$, it follows that $\Theta^* + \hat{\Delta} = \hat{\Theta}^{\oracle} $. As such,
	$$\norm{ \hat{\Theta}^{\oracle} - \Theta^* }_{\max} = \norm{\hat{\Delta}}_{\max} \leq r \text{.} $$
	If $\mathcal{E}_3^*$ also holds, then
	$r < \norm{ \dvec_{\mathcal{A}}^* }_{\min} - a\lambda $ and, by extension, $\norm{ \hat{\dvec}_{\mathcal{A}}^{\oracle} }_{\min} > a\lambda$. That is, $\mathcal{E}_1^* \cap \mathcal{E}_2^* \cap \mathcal{E}_3^* \subseteq \mathcal{E}_2$.
	
	We will now shift to proving that $\mathcal{E}_1^* \cap \mathcal{E}_2^* \cap \mathcal{E}_4^* \cap \mathcal{E}_5^* \subseteq \mathcal{E}_1$. 
	We have shown that if both $\mathcal{E}_1^*$ and $\mathcal{E}_2^*$ hold, then $\Theta^* + \hat{\Delta} = \hat{\Theta}^{\oracle}$. Since $\hat{\Delta}_{\mathcal{A}'^c} = \mathbf{0}$, \eqref{mvt_for_score} gives us that
	\begin{equation}
		\nabla_{\mathcal{A}'} \log L_n(\hat{\Theta}^{\oracle}) = \nabla_{\mathcal{A}'} \log L_n(\Theta^*) + \nabla_{\mathcal{A}'}^2 \log L_n(\Theta^*) \hat{\Delta}_{\mathcal{A}'} + R_{\mathcal{A}'}(\tilde{\Delta})
		\label{LA_expansion}
	\end{equation}
	and 
	\begin{align*}
		\nabla_{\mathcal{A}'^c} \log L_n(\hat{\Theta}^{\oracle}) 
		& = \nabla_{\mathcal{A}'^c} \log L_n(\Theta^*) + \left[\nabla^2 \log L_n(\Theta^*)\right]_{\mathcal{A}'^c} \hat{\Delta} + R_{\mathcal{A}'^c}(\tilde{\Delta})  \\
		& = \nabla_{\mathcal{A}'^c} \log L_n(\Theta^*)  
		- 
		\begin{smallmatrix}
			\xmat_{(\mathcal{A} \cup \{ 0\})^c}'
		\end{smallmatrix}
		\bigl[ 
		\begin{smallmatrix}
			\mathbf{D}(\dvec^*) & \mathbf{0}\\
			\mathbf{0} & \mathbf{I}
		\end{smallmatrix} 
		\bigr]
		\bigl[ 
		\begin{smallmatrix}
			\xmat_{\mathcal{A} \cup \{ 0\}} & -\yvec
		\end{smallmatrix}
		\bigr]
		\hat{\Delta}_{\mathcal{A}'} + R_{\mathcal{A}'^c}(\tilde{\Delta}),
	\end{align*}
	where $R(\tilde{\Delta})$ is defined as before, with $\tilde{\Delta}^{(j)}$ on the line segment joining $\mathbf{0}$ and $\hat{\Delta}$ for $j =0, 1, \ldots, p+1$. 
	Using the fact that $\nabla_{\mathcal{A}'} \log L_n(\hat{\Theta}^{\oracle}) = \mathbf{0}$, we solve for $\hat{\Delta}_{\mathcal{A}}$ in \eqref{LA_expansion} and substitute it into the previous expression to find
	\begin{multline}
		\nabla_{\mathcal{A}'^c} \log L_n(\hat{\Theta}^{\oracle}) 
		= \nabla_{\mathcal{A}'^c} \log L_n(\Theta^*) + R_{\mathcal{A}'^c}(\tilde{\Delta})  
		- 
		\begin{smallmatrix}
			\xmat_{(\mathcal{A} \cup \{ 0 \})^c}'
		\end{smallmatrix}
		\bigl[\begin{smallmatrix}
			\mathbf{D}(\dvec^*) & \mathbf{0}\\
			\mathbf{0} & \mathbf{I}
		\end{smallmatrix}\bigr]
		\bigl[\begin{smallmatrix}
			\xmat_{(\mathcal{A} \cup \{ 0 \})} & -\yvec
		\end{smallmatrix}\bigr] \\ 
		\times
		\left(
		\nabla_{\mathcal{A}'}^2 \log L_n(\Theta^*)
		\right)^{-1} 
		\left(
		- \nabla_{\mathcal{A}'} \log L_n(\Theta^*) -  R_{\mathcal{A}'}(\tilde{\Delta})
		\right) \text{.} \label{eqn:Ac score expression}
	\end{multline}
	We have already shown that if both $\mathcal{E}_1^*$ and $\mathcal{E}_2^*$ hold, then $\norm{\frac{1}{n} R(\tilde{\Delta})}_{\max} \leq \frac{1}{3 Q_2} r $. As a consequence of this, we see that
	$\norm{\frac{1}{n}R(\tilde{\Delta})}_{\max} \leq \norm{ \frac{1}{n} \nabla_{\mathcal{A}'} \log L_n(\Theta^*) }_{\max} $ under $\mathcal{E}_1^*$ and $\mathcal{E}_2^*$.
	Returning to \eqref{eqn:Ac score expression}, the triangle inequality provides that, under $\mathcal{E}_1^*$ and $\mathcal{E}_2^*$, 
	{\small \begin{align}
			\norm{ \frac{1}{n} \nabla_{\mathcal{A}'^c} \log L_n(\hat{\Theta}^{\oracle}) }_{\max}
			& \leq  \norm{ \frac{1}{n} \nabla_{\mathcal{A}'^c} \log L_n(\Theta^*) }_{\max} + \norm{\frac{1}{n} R_{\mathcal{A}'^c}(\tilde{\Delta})}_{\max} \notag \\
			& \text{ \: \: } + \widehat{Q}_3 \left(  \norm{ \frac{1}{n} \nabla_{\mathcal{A}'} \log L_n(\Theta^*) }_{\max} + \norm{\frac{1}{n} R_{\mathcal{A}'}(\tilde{\Delta})}_{\max}   \right) \notag \\
			& \leq (2\widehat{Q}_3 + 1) \norm{ \frac{1}{n} \nabla_{\mathcal{A}'} \log L_n(\Theta^*) }_{\max} +  \norm{ \frac{1}{n} \nabla_{\mathcal{A}'^c} \log L_n(\Theta^*) }_{\max}  \label{eqn:logLAc bound with Q3 hat} \text{,}
	\end{align}}
	where 
	$\widehat{Q}_3 = \norm{\frac{1}{n}
		\xmat_{(\mathcal{A} \cup \{ 0 \})^c}'
		\begin{bmatrix}
			\mathbf{D}(\dvec^*) & \mathbf{0}\\
			\mathbf{0} & \mathbf{I}
		\end{bmatrix}
		\begin{bmatrix}
			\xmat_{\mathcal{A} \cup \{ 0 \}} & -\yvec
		\end{bmatrix}
		\left(
		\frac{1}{n}\nabla_{\mathcal{A}'}^2 \log L_n(\Theta^*)
		\right)^{-1}
	}_{\infty} \text{.}$
	
	The sub-multiplicativity of the $\ell_{\infty}$-norm gives us
	\begin{equation*}
		\widehat{Q}_3 \leq \norm{ \frac{1}{n}
			\xmat_{(\mathcal{A} \cup \{ 0 \})^c}'
			\begin{bmatrix}
				\mathbf{D}(\dvec^*) & \mathbf{0}\\
				\mathbf{0} & \mathbf{I}
			\end{bmatrix}
			\begin{bmatrix}
				\xmat_{\mathcal{A} \cup \{ 0 \}} & -\yvec
			\end{bmatrix}
		}_{\infty} \widehat{Q}_2 \text{.}
	\end{equation*}
	Recall that under $\mathcal{E}_2^*$, $ \widehat{Q}_2 \leq \frac{3}{2}Q_2$. We see that under $\mathcal{E}_5^*$,
	{\small $$
		\norm{ 
			\begin{smallmatrix}
				\frac{1}{n}
				\xmat_{(\mathcal{A} \cup \{ 0 \})^c}'
			\end{smallmatrix}
			\bigl[\begin{smallmatrix}
				\mathbf{D}(\dvec^*) & \mathbf{0}\\
				\mathbf{0} & \mathbf{I}
			\end{smallmatrix}\bigr]
			\bigl[\begin{smallmatrix}
				\xmat_{(\mathcal{A} \cup \{ 0 \})} & -\yvec
			\end{smallmatrix}\bigr]
		}_{\infty}
		\leq
		\frac{3}{2} 
		\norm{ 
			\E\left[
			\begin{smallmatrix}
				\frac{1}{n}
				\xmat_{(\mathcal{A} \cup \{ 0 \})^c}'
			\end{smallmatrix}
			\bigl[\begin{smallmatrix}
				\mathbf{D}(\dvec^*) & \mathbf{0}\\
				\mathbf{0} & \mathbf{I}
			\end{smallmatrix}\bigr]
			\bigl[\begin{smallmatrix}
				\xmat_{(\mathcal{A} \cup \{ 0 \})} & -\yvec
			\end{smallmatrix}\bigr]
			\right]
		}_{\infty} \text{.}
		$$}
	Therefore if $\mathcal{E}_2^*$ and $\mathcal{E}_5^*$ both hold, then $\widehat{Q}_3 \leq \frac{9}{4}Q_3$. Considering this finding in conjunction with \eqref{eqn:logLAc bound with Q3 hat}, we see that if $\mathcal{E}_1^*$, $\mathcal{E}_2^*$, $\mathcal{E}_4^*$, and $\mathcal{E}_5^*$ all hold, then
	\begin{align*}
		\norm{ \frac{1}{n} \nabla_{\mathcal{A}'^c} \log L_n(\hat{\Theta}^{\oracle}) }_{\max}
		& \leq \left(\frac{9}{2} Q_3 + 1\right) \norm{ \frac{1}{n} \nabla_{\mathcal{A}'} \log L_n(\Theta^*) }_{\max} +  \norm{ \frac{1}{n} \nabla_{\mathcal{A}'^c} \log L_n(\Theta^*) }_{\max} \\
		& < a_1 \lambda
		\text{.}
	\end{align*}
	That is, $\mathcal{E}_1^* \cap \mathcal{E}_2^* \cap \mathcal{E}_4^* \cap \mathcal{E}_5^* \subseteq \mathcal{E}_1$.
	
	\bigskip
	\noindent So far, we have shown that
	$\mathcal{E}_1^* \cap \mathcal{E}_2^* \cap \mathcal{E}_3^* \subseteq \mathcal{E}_2$ and 
	$\mathcal{E}_1^* \cap \mathcal{E}_2^* \cap \mathcal{E}_4^* \cap \mathcal{E}_5^* \subseteq \mathcal{E}_1$. We now move to deriving bounds for $P(\mathcal{E}_1^c)$ and $P((\mathcal{E}_1 \cap \mathcal{E}_1)^c)$ using the $\mathcal{E}_i^*$. 
	
	We will start by bounding $P(\mathcal{E}_1^c)$. Note that $P( \mathcal{E}_1^c ) \leq P\left(	(\mathcal{E}_1^* \cap \mathcal{E}_4^*)^c \right) + P({\mathcal{E}_2^*}^c) + P({\mathcal{E}_5^*}^c)$. We will bound the terms on the right hand side of this expression one-by-one.
	We note that
	{\footnotesize $$
		\mathcal{E}_1^* \cap \mathcal{E}_4^* = \left\{ \norm{ \frac{1}{n} \nabla_{\mathcal{A}'} \log L_n(\Theta^*) }_{\max} \leq \min \left\{ C_1(s), \frac{a_1 \lambda}{9Q_3 +2} \right\} \right\}
		\cap \left\{ \norm{ \frac{1}{n} \nabla_{\mathcal{A}'^c} \log L_n(\Theta^*) }_{\max} < \frac{a_1 \lambda}{2}\right\} \text{.}
		$$} 
	We focus first on the event involving $\norm{ \frac{1}{n} \nabla_{\mathcal{A}'} \log L_n(\Theta^*) }_{\max}$. Recall that we defined $C_2(s, \lambda) =\min \left\{ C_1(s), \frac{a_1 \lambda}{9Q_3 +2} \right\}$ for ease of notation.
	The union bound provides that
	$$
	P\left( \norm{ \frac{1}{n} \nabla_{\mathcal{A}'} \log L_n(\Theta^*) }_{\max} > C_2(s, \lambda)  \right)
	\leq \sum_{j \in \mathcal{A}'} P\left( \left| \frac{1}{n} \nabla_j \log L_n(\Theta^*) \right| > C_2(s, \lambda)  \right) \text{.}
	$$
	As such, we can handle each element of $\frac{1}{n} \nabla_{\mathcal{A}'} \log L_n(\Theta^*)$ separately.
	Let $j \in \mathcal{A} \cup \{0\} $. Lemma \ref{lemma:subGaussian_score} implies that $\nabla_j \log L_n(\Theta^*) \sim \subG(\norm{ \xvec_{(j)} }_2^2)$. From here, we apply the Chernoff bound to find
	\begin{equation}
		P\bigg( \left| \frac{1}{n} \nabla_j \log L_n(\Theta^*) \right|  >  C_2(s, \lambda)  \bigg)
		\leq 2 \exp \left( - \frac{ n^2 C_2^2(s, \lambda) }{ 2  \norm{\xvec_{(j)} }_2^2 } \right) \text{.}
		\label{eqn:scoreA chernoff bound}
	\end{equation}
	As an immediate corollary to Lemma \ref{lemma:subExponential_score_gamma}, we know $ \gamma^* \nabla_{\gamma} \log L_n(\dvec^*, \gamma^*) \sim \subExp( 16n + 4 \sum_{i = 1}^n (\xvec_i ' \dvec)^2, 4) $.
	Applying the Chernoff Bound, we find
	{\small \begin{align}
			P \bigg( \left| \frac{1}{n} \nabla_{\gamma} \log L_n(\Theta^*) \right| > & C_2(s, \lambda) \bigg) \notag\\
			& = P \bigg( \bigg| \gamma^* \nabla_{\gamma} \log L_n(\Theta^*) \bigg| > n \gamma^* C_2(s, \lambda) \bigg) \notag \\
			& \leq \begin{cases}
				2 \exp\left( \frac{ -n^2 C_2^2(s, \lambda) {\gamma^*}^2  }{2(16n + 4 \sum_{i = 1}^n (\xvec_i'\dvec^*)^2 )} \right) &\text{ if } 0 \leq nC_2(s, \lambda)\gamma^* \leq \frac{16n + 4  \sum_{i = 1}^n (\xvec_i'\dvec^*)^2}{4} \text{,} \\
				2 \exp(\frac{-n C_2(s, \lambda) \gamma^*}{8}) &\text{ otherwise }
			\end{cases} \notag \\
			& = 2 \exp\left( - \frac{n}{2} \min \left\{ \frac{C_2(s, \lambda) \gamma^*}{4},  \frac{ C_2^2(s, \lambda) {\gamma^*}^2  }{16 + 4 n^{-1} \sum_{i = 1}^n (\xvec_i'\dvec^*)^2 } \right \} \right) \text{.} \label{eqn:score gamma chernoff bound}
	\end{align}}
	
	We give the event involving $\norm{ \frac{1}{n} \nabla_{\mathcal{A}'^c} \log L_n(\Theta^*) }_{\max}$ a similar treatment. Let $j \in \mathcal{A}'^c$. Again using Lemma \ref{lemma:subGaussian_score} and the Chernoff bound, we find
	\begin{equation}
		P\left( \left| \frac{1}{n} \nabla_j  \log L_n(\Theta^*) \right| > \frac{a_1 \lambda}{2} \right) \leq 2\exp \left( - \frac{n^2 a_1^2 \lambda^2}{8 \norm{\xvec_{(j)}}_2^2 } \right) \text{.} \label{eqn:scoreAc chernoff bound}
	\end{equation}
	Combining \eqref{eqn:scoreA chernoff bound}, \eqref{eqn:score gamma chernoff bound}, and \eqref{eqn:scoreAc chernoff bound} with the union bound, we find
	\begin{align}
		P\left(	(\mathcal{E}_1^* \cap \mathcal{E}_4^*)^c \right) 
		& \leq 2(s+1) \exp \left( - \frac{ n C_2^2(s, \lambda) }{ 2 M } \right) \notag \\
		& + 2 \exp\left( - \frac{n}{2} \min \left\{ \frac{C_2(s, \lambda) \gamma^*}{4},  \frac{ C_2^2(s, \lambda) {\gamma^*}^2  }{16 + 4 n^{-1} \sum_{i = 1}^n (\xvec_i'\dvec^*)^2 } \right \} \right) \notag \\ 
		& + 2(p-s)\exp \left( - \frac{n a_1^2 \lambda^2}{8 M } \right) \text{.} \label{eqn:score prob bound 1}
	\end{align}
	
	Our next step is to bound $P({\mathcal{E}_2^*}^c)$ and $P({\mathcal{E}_5^*}^c)$.
	\noindent Using Lemma \ref{lemma:hessian prob bound}, we derive the following upper bound for $P({\mathcal{E}_2^*}^c)$:
	{\footnotesize \begin{align}
			P({\mathcal{E}_2^*}^c)
			& \leq  2 (s + 1)^2 \exp \left( - \frac{2nC_{Q_2}^2(s)}{ \max_{j,k\in\mathcal{A} \cup \{0\}} n^{-1} \sum_{i = 1}^{n} x_{ij}^2x_{ik}^2 } \right) \notag \\ 
			& + 4(s+1) \exp \left( - \frac{2nC_{Q_2}^2(s){\gamma^*}^2}{ \max_{j\in\mathcal{A} \cup \{0\} } n^{-1}\sum_{i = 1}^{n} x_{ij}^2 (2 + \xvec_i'\dvec^* + g(- \xvec_i'\dvec^*) )^2 } \right) \notag \\ 
			& + 2 \exp\left(-\frac{n}{2} \min \left\{ \frac{C_{Q_2}(s){\gamma^*}^2}{8}, \frac{C_{Q_2}^2(s) {\gamma^*}^4}{34 + n^{-1}\sum_{i = 1}^n \frac{1}{2}(\xvec_i'\dvec^*)^2 (2 + \xvec_i'\dvec^* + g(- \xvec_i'\dvec^*))^2 + 8 (\xvec_i'\dvec^*)^2} \right\} \right) \text{,}
	\end{align}}
	where $C_{Q_2}(s) = \min\left\{ \frac{1}{6K_1K_2(s+2)}, \frac{1}{6K_1^3K_2(s+2)}, \frac{Q_2}{4 K_2 (s+2)}  \right\}$.
	Likewise, we use Lemma \ref{lemma:Q_3 matrix prob bound} to derive the following upper bound for $P({\mathcal{E}_5^*}^c)$:
	\begin{align}
		P({\mathcal{E}_5^*}^c)
		& \leq  2(s+1)(p-s) \exp \left( - \frac{2 n C_{Q_3}^2(s)}{ \max_{j \in (\mathcal{A} \cup \{0\})^c, k\in\mathcal{A} \cup \{0\}} n^{-1}\sum_{i = 1}^{n} x_{ij}^2x_{ik}^2 } \right) \notag \\
		& + 2(p - s) \exp \left( - \frac{2nC_{Q_3}^2(s){\gamma^*}^2}{ \max_{j\in(\mathcal{A} \cup \{0\})^c } n^{-1}\sum_{i = 1}^{n} x_{ij}^2 (2 + \xvec_i'\dvec^* + g(- \xvec_i'\dvec^*) )^2 } \right) \text{,}
	\end{align}
	where $C_{Q_3}(s) = \frac{1}{2(s+2)}
	\norm{
		\E\left[
		\frac{1}{n}
		\xmat_{(\mathcal{A} \cup \{ 0 \})^c}'
		\begin{bmatrix}
			\mathbf{D}(\dvec^*) & \mathbf{0}\\
			\mathbf{0} & \mathbf{I}
		\end{bmatrix}
		\begin{bmatrix}
			\xmat_{(\mathcal{A} \cup \{ 0 \})} & -\yvec
		\end{bmatrix}
		\right]
	}_{\infty}$.
	With that, we have all the pieces need to compute our upper bound for $P( \mathcal{E}_1^c )$.
	
	Next we will bound $P((\mathcal{E}_1 \cap \mathcal{E}_2)^c)$. Since $P((\mathcal{E}_1 \cap \mathcal{E}_2)^c) \leq P\left(	(\mathcal{E}_1^* \cap \mathcal{E}_3^* \cap \mathcal{E}_4^*)^c \right) + P({\mathcal{E}_2^*}^c) + P({\mathcal{E}_5^*}^c)$, the only term that we still need to bound is $ P\left(	(\mathcal{E}_1^* \cap \mathcal{E}_3^* \cap \mathcal{E}_4^*)^c \right) $. We see that
	\begin{multline*}
		\mathcal{E}_1^* \cap \mathcal{E}_3^* \cap \mathcal{E}_4^* = \left\{ \norm{ \frac{1}{n} \nabla_{\mathcal{A}'} \log L_n(\Theta^*) }_{\max} < \min \left\{ C_1(s), \frac{a_1 \lambda}{9Q_3 +2}, \frac{1}{3 Q_2} ( \norm{ \dvec_{\mathcal{A}}^* }_{\min} - a\lambda  )  \right\} \right\} \\
		\cap \left\{ \norm{ \frac{1}{n} \nabla_{\mathcal{A}'^c} \log L_n(\Theta^*) }_{\max} < \frac{a_1 \lambda}{2}\right\} \text{.}
	\end{multline*}
	Recall that $C_3(s, \lambda) =\min \left\{ C_1(s), \frac{a_1 \lambda}{9Q_3 +2}, \frac{1}{3 Q_2} ( \norm{ \dvec_{\mathcal{A}}^* }_{\min} - a\lambda  ) \right\}$. Taking the same approach we used to derive \eqref{eqn:score prob bound 1}, we find
	\begin{align}
		P\left(	(\mathcal{E}_1^* \cap \mathcal{E}_3^* \cap \mathcal{E}_4^*)^c \right) 
		& \leq 2(s+1) \exp \left( - \frac{ n C_3^2(s, \lambda) }{ 2 M } \right) \notag \\
		& + 2 \exp\left( - \frac{n}{2} \min \left\{ \frac{C_3(s, \lambda) \gamma^*}{4},  \frac{ C_3^2(s, \lambda) {\gamma^*}^2  }{16 + 4 n^{-1} \sum_{i = 1}^n (\xvec_i'\dvec^*)^2 } \right \} \right) \notag \\ 
		& + 2(p-s)\exp \left( - \frac{n a_1^2 \lambda^2}{8 M } \right) \text{.} \label{eqn:score prob bound 2}
	\end{align}
	With that final addition, we have all the pieces need to compute our upper bound for $P((\mathcal{E}_1 \cap \mathcal{E}_2)^c)$.
\end{proof}

\end{document}